\documentclass[a4paper, 12pt]{report}

\usepackage{calc, amsmath, amsthm, latexsym, amssymb, color, url}
\usepackage{graphicx}
\usepackage[a4paper,margin=1in]{geometry}
\usepackage{amsfonts,mathtools,bm}
\usepackage{upgreek}
\usepackage{subcaption}
\usepackage{booktabs,multirow}
\usepackage{enumitem}
\usepackage{physics}
\usepackage{siunitx}
\usepackage{algorithm}
\usepackage{thmtools}

\usepackage[colorlinks,linkcolor=blue,citecolor=blue,urlcolor=blue]{hyperref}
\usepackage[nameinlink,noabbrev]{cleveref}
\usepackage{textcomp}
\usepackage{titlesec}
\titleformat{\chapter}[hang]
  {\normalfont\bfseries\fontsize{16pt}{19pt}\selectfont}
  {\thechapter.}{0.75em}{\MakeUppercase}
\titlespacing*{\chapter}{0pt}{1.0\baselineskip}{0.6\baselineskip}
\titleformat{\section}
  {\normalfont\bfseries\fontsize{14pt}{17pt}\selectfont}
  {\thesection}{0.6em}{}
\titlespacing*{\section}{0pt}{0.9\baselineskip}{0.4\baselineskip}
\titleformat{\subsection}
  {\normalfont\normalsize\bfseries}
  {\thesubsection}{0.5em}{}
\titlespacing*{\subsection}{0pt}{*0.7}{0.3em}
\declaretheorem[name=Proposition]{prop}
\declaretheorem[name=Corollary]{cor}

\declaretheorem[name=Lemma]{lemma}
\declaretheorem[name=Assumption]{assump}
\declaretheorem[name=Definition,style=definition]{definition}
\declaretheorem[name=Remark,style=remark]{remark}
\declaretheorem[name=Example,style=definition,sibling=definition]{example}
\newcommand{\R}{\mathbb{R}}
\newcommand{\N}{\mathbb{N}}

\newcommand{\diag}{\operatorname{diag}}

\usepackage[english]{babel}
\usepackage{tikz}
\usepackage{tikz-3dplot}
\usepackage{float}

\title{Arbitrage-Free Option Price Surfaces via Chebyshev Tensor Bases and a Hamiltonian Fog Post-Fit}
\author{Robert Jenkinson \'Alvarez}
\date{December 1, 2025}

\begin{document}

\makeatletter
    \begin{titlepage}
        \begin{center}
            \vspace{1cm}
            {\huge \bfseries  \@title }\\[15ex] 
            {\Large  \@author}\\[50ex] 
            {\large \@date}
        \end{center}
    \end{titlepage}
\makeatother
\thispagestyle{empty}
\newpage

\thispagestyle{empty}
\newpage


\begin{abstract}
We study the construction of arbitrage-free option price surfaces from noisy bid-ask quotes across strike and maturity. Our starting point is a Chebyshev representation of the call price surface on a warped log-moneyness/maturity rectangle, together with linear sampling and no-arbitrage operators acting on a collocation grid. Static no-arbitrage requirements are enforced as linear inequalities, while the surface is fitted directly to prices via a coverage-seeking quadratic objective that trades off squared band misfit against spectral and transport-inspired regularisation of the Chebyshev coefficients. This yields a strictly convex quadratic program in the modal coefficients, solvable at practical scales with off-the-shelf solvers (OSQP).

On top of the global backbone, we introduce a local post-fit layer based on a discrete fog of risk-neutral densities on a three-dimensional lattice \((m,\tau,u)\) and an associated Hamiltonian-type energy. On each patch of the \((m,\tau)\) plane, the fog variables are coupled to a nodal price field obtained from the baseline surface, yielding a joint convex optimisation problem that reweights noisy quotes and applies noise-aware local corrections while preserving global static no-arbitrage and locality.

The method is designed such that for equity options panels, the combined procedure achieves high inside–spread coverage in stable regimes (in calm years, $98-99\%$ of quotes are priced inside the bid–ask intervals) and low rates of static no–arbitrage violations (below $1\%$). In stressed periods, the fog layer provides a mechanism for controlled leakage outside the band: when local quotes are mutually inconsistent or unusually noisy, the optimiser allocates fog mass outside the bid–ask tube and justifies small out–of–band deviations of the post–fit surface, while preserving a globally arbitrage–free and well–regularised description of the option surface.
\end{abstract}

\tableofcontents

\chapter{Introduction}

Liquid option markets require a smooth, stable and \emph{arbitrage-free}
surface of prices or implied volatilities over strike and maturity.
Such a surface underpins marking, risk management and model calibration, and
feeds directly into trading and hedging decisions.
In practice, the raw quote grid is sparse, noisy and often inconsistent with
the static no-arbitrage conditions implied by absence of butterfly and
calendar spreads.
Production systems therefore interpolate and regularise the observed quotes
into a dense surface subject to no-arbitrage constraints.

There is substantial literature on arbitrage-free surface
construction.
Parametric approaches such as SVI and its extensions impose functional forms
on implied volatility and derive analytical sufficient conditions for absence
of static arbitrage.
Alternatively, nonparametric smoothing methods reconstruct prices or
volatilities on a grid while enforcing no-arbitrage inequalities either as
hard constraints or via penalisation.
These methods have been successfully deployed in practice, but there remains
a trade-off between fidelity to the bid-ask quotes, strict enforcement of
no-arbitrage on a dense grid, and computational cost on large universes of
names and dates.

This paper contributes a practical convex-programming formulation for option price
surfaces that aims to balance these considerations, together with a local geometric
post-fit layer that explicitly models quote noise on difficult regions of the surface.
The key ingredients are:
\begin{itemize}
  \item a global Chebyshev representation of the price surface on a warped
  log-moneyness / maturity rectangle, which provides high approximation power;

  \item linear operators that encode static no-arbitrage constraints on a dense
  collocation grid (monotonicity in strike, convexity in strike, calendar
  monotonicity and simple bounds), so that absence of butterfly and calendar
  arbitrage is enforced directly in price space;

  \item a coverage-seeking quadratic objective aligned to the bid–ask geometry,
  augmented by spectral and transport-inspired regularisers (ridge in the
  Chebyshev coefficients, discrete $H^{-1}$ smoothing of the density, short-end
  anchoring and frequency tapering) that stabilise the fit while preserving
  convexity;

  \item a patch-wise post-fit in price space, built on a discrete ``fog'' of
  risk-neutral densities on a three-dimensional lattice $(m,\tau,u)$ and a
  Hamiltonian-type energy on that fog, which yields a convex, noise-aware local
  correction of the baseline surface on problematic regions while preserving
  global static no-arbitrage.
\end{itemize}

The resulting baseline surface is obtained as the solution of a single medium-scale
QP with sparse structure, solvable reliably with off-the-shelf solvers, and tuned to
reach high within-band coverage and low static no-arbitrage violation rates on a
dense grid. The discrete Hamiltonian fog layer appears as a second, fully convex
post-fit stage defined on local patches in $(m,\tau)$; it is implemented in this
paper in a finite-dimensional setting (Chapter \ref{sec:postfit-hamiltonian-fog}) and used to refine the baseline surface in stressed regimes.

The rest of the paper is organised as follows. Chapter \ref{sec:markt_coor_trg} defines the market
coordinates, targets and static no-arbitrage axioms. Chapters \ref{sec:wrapped_cheb_dsg}-\ref{sec:noarb-discrete} build the warped
Chebyshev tensor basis and the no-arbitrage operators on a collocation grid.
Chapters \ref{sec:coverage-data}-\ref{sec:global-noarb-grid} specify the coverage-seeking data term and the soft no-arbitrage
penalties, and Chapter \ref{sec:convex_prog} assembles the global QP and discusses convexity and
solution. Chapters \ref{sec:ridge-etc} and \ref{sec:dw-omega} develop the spectral and transport-inspired quadratic
regularisers, and Chapter \ref{sec:diag_and_imp} collects structure monitors that diagnose stability.
Chapter \ref{sec:postfit-hamiltonian-fog} then introduces the patch-wise Hamiltonian fog post-fit in price space,
formulated as a joint convex optimisation in the nodal surface and fog variables.
Chapter \ref{sec:conclusion} concludes and outlines a continuum version of the fog/Hamiltonian
geometry, which is deferred to a separate theoretical follow-up paper.

\chapter[Market coordinates, targets, and no-arbitrage axioms]%
  {Market coordinates, targets, and\\ no-arbitrage axioms}
\label{sec:markt_coor_trg}

Let $t$ be a trading date and $F_t(T)$ denote the discount adjusted forward for maturity time $T$. Set $\tau=T-t>0$, the time to maturity. We work in forward discounted prices:
\[
C_f(K,\tau) \;\coloneqq\; e^{r(\tau)\tau} C(K,\tau)
\quad\text{and}\quad
m \;\coloneqq\; \log\!\frac{K}{F_t(\tau)}.
\]
Throughout, we fit a surface $C_f(m,\tau)$ from quoted calls. Puts follow from put-call parity.

These are sufficient conditions for Static no-arbitrage for calls (for a.e. $m,\tau$):
\begin{align}
\partial_m C_f(m,\tau) &\le 0, \label{eq:monoK}\\
\partial_{mm} C_f(m,\tau) &\ge 0, \label{eq:convK}\\
\partial_{\tau} C_f(m,\tau)\big|_{K} &\ge 0. \label{eq:cal}
\end{align}
Bounds: $0 \le C_f(m,\tau)\le F_t(\tau)$ and $C_f(m,0^+)=\big(F_t(0)-K\big)^+$.

Throughout this section we tacitly assume enough regularity for the continuum
derivatives in \eqref{eq:monoK}–\eqref{eq:cal} to be well defined on the compact rectangle
where we approximate the surface. In particular, on the box $[m_{\min},m_{\max}] \times [\tau_{\min},\tau_{\max}]$ used in Section~\ref{sec:wrapped_cheb_dsg}, we work under
\[
C_f \in C^{2,1}\bigl([m_{\min},m_{\max}] \times [\tau_{\min},\tau_{\max}]\bigr),
\qquad
F \in C^{1}\bigl([\tau_{\min},\tau_{\max}]\bigr),
\]
so that $\partial_m C_f$, $\partial_{mm} C_f$, $(\partial_\tau C_f)\vert_{K}$ and
$\frac{\mathrm{d}}{\mathrm{d}\tau}\log F(\tau)$ all exist and are continuous. The later
discrete QP only uses linear operators on a grid, but these smoothness conditions
provide the natural continuum axioms behind \eqref{eq:monoK}–\eqref{eq:cal}.

\begin{remark}[Calendar derivative at fixed strike]
Since the basis uses $(m,\tau)$, the calendar derivative at fixed $K$ becomes
\(
\partial_{\tau} C_f|_{K}=\partial_{\tau} C_f + \partial_m C_f \cdot \partial_{\tau} m\big|_{K}
= \partial_{\tau} C_f - \partial_m C_f \cdot \big(r_{\text{cc}}(\tau)+\tau r'_{\text{cc}}(\tau)\big).
\)
This is implemented exactly in the operators below.
\end{remark}

\section*{Notation}
\begin{table}[h]
\centering
\begin{tabular}{@{}ll@{}}
\toprule
Symbol & Meaning \\
\midrule
\(K\) & Strike \\
\(\tau\) & Time to maturity \(T-t\) \\
\(F_t(\tau)\) & Forward (discount‑adjusted) underlying for \(\tau\) \\
\(m\) & Log‑moneyness \(\log(K/F_t(\tau))\) \\
\(C_f\) & Forward‑discounted call price \\
\(\rho\) & Risk‑neutral density \(\partial_{KK} C_f\) \\
\(A,\ A_m,\ A_{mm},\ A_\tau\) & Design/derivative blocks in coefficient space \\
\bottomrule
\end{tabular}
\end{table}

\chapter{Warped tensor Chebyshev basis and design matrices}\label{sec:wrapped_cheb_dsg}

This chapter builds the approximation space and its derivative blocks.

\section{Why normalise to \texorpdfstring{$[-1,1]^2$}{[-1,1]\string^2} and why Chebyshev?}

The problem is to approximate a continuous surface \(C_f(m,\tau)\) on a compact rectangle
\([m_{\min},m_{\max}]\times[\tau_{\min},\tau_{\max}]\).
On a compact interval, polynomials are dense (Weierstrass), and Chebyshev polynomials are numerically preferred because:
\begin{enumerate}[label=(\roman*)]
    \item they minimise Runge oscillations on \([-1,1]\),
    \item they admit stable three–term recurrences and Clenshaw evaluation,
    \item they possess explicit derivative identities useful for Greeks.
\end{enumerate}
We therefore map each axis to \([-1,1]\) and expand in a tensor-product
Chebyshev basis.

\section{Coordinate warps (endpoint preserving)}

Let $[m_{\min},m_{\max}]$ and $[\tau_{\min},\tau_{\max}]$ be robust, date-adaptive intervals.

Define the warps
\[
x=\Phi_m(m)\in[-1,1],\qquad y=\Phi_\tau(\tau)\in[-1,1],
\]
so that the interval endpoints map exactly to \(\pm1\).

\paragraph{Log-moneyness warp (asinh).}
Let \(c_m\) be a centre (e.g.\ the liquidity-weighted median of \(m\)),
and \(\lambda_m>0\) a tail-compression parameter. Set
\[
\phi_{m,\pm}\coloneqq \operatorname{asinh}\!\big(\lambda_m(m_{\max/\min}-c_m)\big),\quad
W_m\coloneqq \phi_{m,+}-\phi_{m,-},
\]
\[
\boxed{\ \Phi_m(m)=\frac{2}{W_m}\Big(\operatorname{asinh}(\lambda_m(m-c_m))-\phi_{m,-}\Big)-1\ }.
\]
Then \(\Phi_m(m_{\min})=-1\), \(\Phi_m(m_{\max})=+1\). Derivatives (by the chain rule) are
\[
\Phi_m'(m)=\frac{2\lambda_m}{W_m}\frac{1}{\sqrt{1+\lambda_m^2(m-c_m)^2}},\qquad
\Phi_m''(m)=-\frac{2\lambda_m^3}{W_m}\frac{(m-c_m)}{\big(1+\lambda_m^2(m-c_m)^2\big)^{3/2}}.
\]

\paragraph{Maturity warp (square-root).}
Let \(\Delta_\tau\coloneqq \tau_{\max}-\tau_{\min}>0\) and \(s(\tau)\coloneqq
(\tau-\tau_{\min})/\Delta_\tau\in[0,1]\). Set
\[
\boxed{\ \Phi_\tau(\tau)=2\sqrt{s(\tau)}-1\ },
\qquad
\Phi_\tau'(\tau)=\frac{1}{\sqrt{\Delta_\tau\,(\tau-\tau_{\min})}}\;\;(\text{finite if }\tau>\tau_{\min}).
\]
The square-root allocates higher resolution near short maturities.
(If \(\tau_{\min}\) is very close to \(0\), a small positive floor avoids the
endpoint singularity in \(\Phi_\tau'\).)

\begin{remark}[Why these warps]
The asinh warp allocates more resolution near \(m\approx c_m\) (ATM) while compressing
deep wings; the square-root warp concentrates basis power near short maturities
where curvature in \(\tau\) is largest. Both preserve endpoints and expose simple
chain-rule factors for derivatives.
\end{remark}

\section{Chebyshev polynomials on \texorpdfstring{$[-1,1]$}{[-1,1]}}
For \(x\in[-1,1]\), the Chebyshev polynomials of the first kind are
\[
T_k(x)=\cos\big(k\arccos x\big),\qquad
T_0=1,\ \ T_1=x,\ \ T_{k+1}=2x\,T_k-T_{k-1}.
\]
The derivatives needed later are available in closed form:
\[
T_k'(x)=k\,U_{k-1}(x),\qquad
\]
and, for $|x|<1$,
\[
(1-x^2)\,T_k''(x)-x\,T_k'(x)+k^2\,T_k(x)=0
\;\Rightarrow\;
\boxed{\ T_k''(x)=\frac{x\,k\,U_{k-1}(x)-k^2 T_k(x)}{1-x^2}\ },
\]
where \(U_n\) are Chebyshev polynomials of the second kind, which are defined recursively
(\(U_0(x)=1,\,U_1(x)=2x,\,U_{n+1}(x)= 2x \cdot U_n(x)-U_{n-1}(x)\)).
In practice, we evaluate \(T_k, U_{k-1}\) stably via Clenshaw recurrences.

\section{Tensor-product basis for the surface}
Let \(K,L\in\N\) be polynomial degrees in \(m\) and \(\tau\), and define
the coefficient array \(a=\{a_{k\ell}\}_{k=0,\dots,K;\ \ell=0,\dots,L}\).
We approximate
\[
\boxed{\ C_f(m,\tau)\;=\;\sum_{k=0}^{K}\sum_{\ell=0}^{L}
a_{k\ell}\,T_k\!\big(\Phi_m(m)\big)\,T_\ell\!\big(\Phi_\tau(\tau)\big)\ }.
\]
Stacking \(a\) into a vector in \(\R^{P}\) with
\(P=(K{+}1)(L{+}1)\) yields a linear map from coefficients to prices.

\section{Design matrices at arbitrary points}
Given data points \(\{(m_i,\tau_i)\}_{i=1}^N\), set \(x_i=\Phi_m(m_i)\),
\(y_i=\Phi_\tau(\tau_i)\). Define the following row vectors
\[
\bm{t}(x_i)=\big[T_0(x_i),\dots,T_K(x_i)\big],
\]
\[
\bm{s}(y_i)=\big[T_0(y_i),\dots,T_L(y_i)\big].
\]
The pair \((k,\ell)\) defines the index of the column (where \(k \in \{0,\dots , K\}\) and \(l \in \{0,\dots , L\}\)). 

\begin{definition}
    The (price) design matrix \(A\in\R^{N\times P}\) is defined as
    \[
    \boxed{\ A_{i,(k,\ell)} = T_k(x_i)\,T_\ell(y_i)\ }.
    \]
\end{definition}

Equivalently, if \(\Phi_K\in\R^{N\times(K+1)}\) stacks \(\bm{t}(x_i)\) and
\(\Phi_L\in\R^{N\times(L+1)}\) stacks \(\bm{s}(y_i)\), then \(A\) is the
row-wise Khatri–Rao product \(A=\Phi_K \odot \Phi_L\); for grid evaluations,
Kronecker structure \((\Phi_L\otimes \Phi_K)\) can be exploited.

\section{Derivative design blocks via the chain rule}

\begin{prop}[Closed-form derivative design blocks]
\label{prop:derivative-blocks}
Fix integers $K,L\ge 0$ and let $P=(K{+}1)(L{+}1)$. For each data point $(m_i,\tau_i)$ define
\[
x_i\coloneqq \Phi_m(m_i),\qquad y_i\coloneqq \Phi_\tau(\tau_i),
\]
where $\Phi_m\in C^2$ and $\Phi_\tau\in C^1$ on their domains. Let the price design matrix
$A\in\mathbb{R}^{N\times P}$ be
\[
A_{i,(k,\ell)} \;=\; T_k(x_i)\,T_\ell(y_i),\qquad 0\le k\le K,\ 0\le \ell\le L,
\]
with any fixed stacking $(k,\ell)\mapsto (k,\ell)$-column. Define the ``inner-variable'' derivative
matrices by
\[
(\partial_x A)_{i,(k,\ell)} \coloneqq T_k'(x_i)\,T_\ell(y_i),\qquad
(\partial_{xx} A)_{i,(k,\ell)} \coloneqq T_k''(x_i)\,T_\ell(y_i),
\]
\[
(\partial_y A)_{i,(k,\ell)} \coloneqq T_k(x_i)\,T_\ell'(y_i).
\]
For any coefficient vector $a\in\mathbb{R}^P$, consider the model values
\[
\widehat C_i = (Aa)_i = \sum_{k,\ell} a_{k\ell}\,T_k(x_i)T_\ell(y_i).
\]
Then the vectors of physical derivatives evaluated at the same points are linear images of $a$:
\[
\big(\partial_m \widehat C\big)_i = (A_m a)_i,\qquad
\big(\partial_{mm} \widehat C\big)_i = (A_{mm} a)_i,\qquad
\big(\partial_{\tau} \widehat C\big)_i = (A_{\tau} a)_i,
\]
where the derivative design blocks are
\[
\boxed{\ A_m \;=\; \diag\!\big(\Phi_m'(m)\big)\,\partial_x A,
}
\]
\[
\boxed{
A_{mm}\;=\;\diag\!\big((\Phi_m')^2\big)\,\partial_{xx}A \;+\; \diag\!\big(\Phi_m''(m)\big)\,\partial_x A,}
\]
\[
\boxed{\ A_\tau \;=\; \diag\!\big(\Phi_\tau'(\tau)\big)\,\partial_y A. }
\]
\end{prop}
\begin{proof}
All statements follow from linearity and the chain rule, applied row-wise.

\paragraph{Setup:}
Write the one–dimensional warped basis functions
\[
\phi_k(m)\coloneqq T_k(\Phi_m(m)),\qquad \psi_\ell(\tau)\coloneqq T_\ell(\Phi_\tau(\tau)).
\]
Then the model at $(m_i,\tau_i)$ is
\[
\widehat C_i \;=\; \sum_{k=0}^K\sum_{\ell=0}^L a_{k\ell}\,\phi_k(m_i)\,\psi_\ell(\tau_i)
\;=\; \sum_{k,\ell} a_{k\ell}\,T_k(x_i)\,T_\ell(y_i).
\]
By construction, the $i$-th row of $A$ consists of the basis values
$\{T_k(x_i)T_\ell(y_i)\}_{k,\ell}$, so $\widehat C = A a$.

\paragraph{First derivative in $m$:}
Differentiating $\phi_k(m)=T_k(\Phi_m(m))$ with respect to $m$, by the chain rule we obtain the following:
\[
\frac{d}{dm}\phi_k(m) \;=\; T_k'\!\big(\Phi_m(m)\big)\,\Phi_m'(m).
\]
Fixing $\tau$ , the derivative of each product is
\[
\frac{\partial}{\partial m}\big[\phi_k(m)\psi_\ell(\tau)\big]
= T_k'\!\big(\Phi_m(m)\big)\,\Phi_m'(m)\cdot T_\ell\!\big(\Phi_\tau(\tau)\big).
\]
Evaluating at $(m_i,\tau_i)$ and summing over $(k,\ell)$:
\[
\big(\partial_m \widehat C\big)_i
= \Phi_m'(m_i)\sum_{k,\ell} a_{k\ell}\,T_k'(x_i)\,T_\ell(y_i)
= \big(\ \Phi_m'(m_i)\cdot (\partial_x A\,a)_i\ \big).
\]
Taking the scalar formula for each $i$ and writing as a vector-matrix equation yields 
\[
\partial_m \widehat C = \diag(\Phi_m'(m))\,\partial_x A\,a,
\]
where
$A_m=\diag(\Phi_m'(m))\,\partial_x A$.

\paragraph{Second derivative in $m$:}
Differentiate once more, using the product rule and chain rule:
\[
\frac{d^2}{dm^2}\phi_k(m)
= \frac{d}{dm}\Big(T_k'\!\big(\Phi_m(m)\big)\,\Phi_m'(m)\Big)
\]
\[
= T_k''\!\big(\Phi_m(m)\big)\,(\Phi_m'(m))^2 + T_k'\!\big(\Phi_m(m)\big)\,\Phi_m''(m).
\]
Therefore,
\[
\frac{\partial^2}{\partial m^2}\big[\phi_k(m)\psi_\ell(\tau)\big]
= \Big(T_k''(\Phi_m(m))(\Phi_m')^2 + T_k'(\Phi_m(m))\Phi_m''\Big)\,T_\ell(\Phi_\tau(\tau)).
\]
Evaluating at $(m_i,\tau_i)$ and summing,
\[
\big(\partial_{mm}\widehat C\big)_i
= (\Phi_m'(m_i))^2\sum_{k,\ell} a_{k\ell}\,T_k''(x_i)\,T_\ell(y_i)
  + \Phi_m''(m_i)\sum_{k,\ell} a_{k\ell}\,T_k'(x_i)\,T_\ell(y_i).
\]
In matrix form,
\[
\partial_{mm}\widehat C
= \diag\!\big((\Phi_m'(m))^2\big)\,\partial_{xx}A\,a
 \;+\; \diag\!\big(\Phi_m''(m)\big)\,\partial_x A\,a,
\]
so $A_{mm}=\diag((\Phi_m')^2)\partial_{xx}A+\diag(\Phi_m'')\partial_x A$.

\paragraph{First derivative in $\tau$:}
Analogously, with $y=\Phi_\tau(\tau)$,
\[
\frac{d}{d\tau}\psi_\ell(\tau) \;=\; T_\ell'\!\big(\Phi_\tau(\tau)\big)\,\Phi_\tau'(\tau),
\]
Fixing $m$ , 
\[
\frac{\partial}{\partial\tau}\big[\phi_k(m)\psi_\ell(\tau)\big]
= T_k(\Phi_m(m))\,T_\ell'(\Phi_\tau(\tau))\,\Phi_\tau'(\tau).
\]
Evaluating at $(m_i,\tau_i)$ and summing over $(k,\ell)$ yields:
\[
\big(\partial_\tau \widehat C\big)_i
= \Phi_\tau'(\tau_i)\sum_{k,\ell} a_{k\ell}\,T_k(x_i)\,T_\ell'(y_i)
= \big(\ \Phi_\tau'(\tau_i)\cdot (\partial_y A\,a)_i\ \big),
\]
Taking the scalar formula for each $i$ and writing as a vector-matrix equation yields 
\[
\partial_\tau \widehat C = \diag(\Phi_\tau'(\tau))\,\partial_y A\,a
\]
where
$A_\tau=\diag(\Phi_\tau'(\tau))\,\partial_y A$.

\paragraph{Conclusion.}
In each case the derivative vector equals a fixed matrix (depending only on the warps and basis) times $a$, establishing the stated formulas.
\end{proof}

\begin{remark}
Precompute \(\{T_k(x_i),U_{k-1}(x_i)\}_{k\le K}\) and
\(\{T_\ell(y_i),U_{\ell-1}(y_i)\}_{\ell\le L}\) via Clenshaw recurrences.
Then obtain \(T_k'(x_i)=k\,U_{k-1}(x_i)\) and, for interior points $|x_i|<1$,
\(T_k''(x_i)=\big(x_i\,k\,U_{k-1}(x_i)-k^2 T_k(x_i)\big)/(1-x_i^2)\).
At the Chebyshev–Lobatto endpoints $x=\pm1$ the denominator $1-x^2$ vanishes, but
$T_k$ is a polynomial so $T''_k(x)$ exists and is finite there. In practice we define
$T''_k(\pm 1)$ by continuity (or via the closed forms
$T''_k(1)=\tfrac{k^2(k^2-1)}{3}$ and
$T''_k(-1)=(-1)^k\tfrac{k^2(k^2-1)}{3}$ for $k\ge2$) and use these values whenever
$|1-x_i^2|$ is numerically small. With this convention, all entries of $\partial_{xx}A$
are well defined and the assembly of $A,\partial_x A,\partial_{xx}A,\partial_y A$ uses
only closed-form expressions, with no numerical differencing.
\end{remark}

\chapter{No-arbitrage operators on a collocation grid}\label{sec:no-arb-grid}\label{sec:noarb-discrete}

From the previous chapter we have for any set of evaluation points $(m,\tau)$, the
price design matrix $A$ and the derivative blocks $A_m, A_{mm}, A_\tau$ are defined by
\[
(Aa)(m,\tau)=C_f(m,\tau),\quad
(A_m a)(m,\tau)=\partial_m C_f,\quad
\]
\[
(A_{mm} a)(m,\tau)=\partial_{mm} C_f,\quad
(A_\tau a)(m,\tau)=\partial_\tau C_f.
\]
All maps are linear in $a$ and are computed pointwise via the chain rule.

\section{Collocation grid and evaluation}
Let $\{(m_g,\tau_g)\}_{g=1}^{G}$ be a fixed collocation grid used to test the
no-arbitrage shape conditions (Chebyshev nodes in $m$, uniform in $\tau$ are a robust choice).
On this grid define the forward (no-arbitrage price of receiving one unit of the underlying at time $(t+\tau)$) and strike (rearrangement from the definition of $m$)
\[
F_g \coloneqq F_t(\tau_g),\qquad K_g \coloneqq F_g\,e^{m_g}.
\]
We evaluate the same derivative blocks on the grid (rather than on \(\{(m_i,\tau_i)\}_{i=1}^N\), we evaluate on $\{(m_g,\tau_g)\}_{g=1}^{G}$); keeping the
symbols $A, A_m, A_{mm}, A_\tau$ for the $G\times P$ versions where the $g$-th row corresponds
to $(m_g,\tau_g)$.

\begin{remark}[Why a separate grid]
Quotes can be sparse or clustered. A collocation grid decouples shape testing from where data happen to lie and gives uniform control of violations over the rectangle in $(m,\tau)$.
\end{remark}

\section{Strike-space operators (monotonicity and convexity)}
Static no-arb for \emph{calls} requires $\partial_K C_f\le 0$ and $\partial_{KK} C_f\ge 0$ at fixed $\tau$. 
However, our derivative blocks are in the $m$ coordinate, where $m=\ln\!\big(K/F(\tau)\big)$.
At fixed $\tau$,
\[
\frac{\partial m}{\partial K}=\frac{1}{K},\qquad
\frac{\partial^2 m}{\partial K^2}=-\frac{1}{K^2}.
\]
For any smooth $f(m,\tau)$, by the chain rule
\[
\frac{\partial f}{\partial K}\Big|_{\tau}=\frac{1}{K}\,f_m,\qquad
\frac{\partial^2 f}{\partial K^2}\Big|_{\tau}=\frac{1}{K^2}\,(f_{mm}-f_m).
\]

Apply this with $f=C_f$ row-wise on the grid, replacing $f_m$ and $f_{mm}$ by $A_m a$ and
$A_{mm} a$.
\[
\partial _K C_f (K_g, \tau_g) = \frac{1}{K_g} \partial _m C_f (m_g, \tau_g) = \frac{1}{K_g}(A_m a)_g
\]
\[
\partial _{KK} C_f (K_g, \tau_g) = - \frac{1}{K_g^2} \partial _m C_f (m_g, \tau_g) + \frac{1}{K_g^2} \partial _{mm} C_f (m_g, \tau_g) = \frac{(A_{mm} a)_g - (A_m a)_g}{K_g^2}
\]
This yields the \emph{linear operators} that map coefficients $a$ to strike derivatives:
\begin{align}
\text{Monotonicity in strike:}\qquad
&A_{K}\;=\;\diag(K_g)^{-1}\,A_m,\label{eq:A_K}\\[4pt]
\text{Convexity in strike:}\qquad
&A_{KK}\;=\;\diag(K_g)^{-2}\,\big(A_{mm}-A_m\big).\label{eq:A_KK}
\end{align}
Thus, $(A_K a)_g=\partial_K C_f(K_g,\tau_g)$ and $(A_{KK} a)_g=\partial_{KK} C_f(K_g,\tau_g)$.

\section{Calendar derivative at fixed strike}\label{cal_der}\label{sec:calendar-operator}

Calendar no-arbitrage requires $\partial_\tau C_f\big|_{K}\ge 0$.
The block $A_\tau$ computes $\partial_\tau C_f$ at fixed $m$; to switch to fixed $K$ use the following relation
\[
\Big(\partial_\tau C_f\Big)_{\!K}
= \Big(\partial_\tau C_f\Big)_{\!m}
+ \Big(\partial_m C_f\Big)\Big(\partial_\tau m\Big)_{\!K}.
\]
With $m=\log\big(K/F(\tau)\big)$ and $K$ fixed,
\[
\Big(\partial_\tau m\Big)_{\!K}=-\,\frac{d}{d\tau}\log F(\tau).
\]
Two equivalent parameterisations of $F$ give:

\smallskip
\noindent\emph{(i) General form.} Let $\rho(\tau)\coloneqq \tfrac{d}{d\tau}\log F(\tau)$.
Then
\[
A_{\tau|K} \;=\; A_\tau \;-\; \diag\!\big(\rho(\tau_g)\big)\,A_m.
\]

\smallskip
\noindent\emph{(ii) Report convention.} If $\log F(\tau)=\tau\,r_{\mathrm{cc}}(\tau)$ (net continuously
compounded carry), then
$\rho(\tau)=r_{\mathrm{cc}}(\tau)+\tau\,r'_{\mathrm{cc}}(\tau)$ and
\begin{equation}
A_{\tau|K}\;=\;A_\tau \;+\; \diag\!\big(-r_{\mathrm{cc}}(\tau_g)-\tau_g r'_{\mathrm{cc}}(\tau_g)\big)\,A_m.\label{eq:A_tK}
\end{equation}

\begin{remark}[Sanity checks]
If $F$ is flat (zero carry), then $\rho\equiv 0$ and $A_{\tau|K}=A_\tau$.
If carry is constant $r$, then $\rho\equiv r$ and $A_{\tau|K}=A_\tau - r\,A_m$.
\end{remark}

\section{Price map and bound operators}
Price non-negativity and upper bounds by the forward read on the grid as
\[
0 \ \le\ (A a)_g \ \le\ F_g,\qquad g=1,\dots,G.
\]
Simply write $A_{\text{price}}=A$ and use the known vector $F=(F_g)_g$ when
imposing hard constraints or soft penalties for violations.

\section{Row scaling and a single no-arb weight}\label{sec:row-scaling-single-weight}
The three no-arbitrage defect maps have different natural magnitudes and units:
\[
A_K a \quad(\text{``price per strike''}),\qquad
-A_{KK} a \quad(\text{``price per strike}^2\text{''}),
\]
\[
-A_{\tau|K} a \quad(\text{``price per time''}).
\]
If a single penalty weight $\lambda_{\mathrm{NA}}$ is applied to all three without normalisation, the largest magnitude block dominates and the others become numerically inert. Therefore, normalise each block by a positive scalar so that a single $\lambda_{\mathrm{NA}}$ can control them comparably.
\paragraph{Blocks to be normalised:}
On the collocation grid (size $G\times P$), set
\[
B_1 \coloneqq A_K,\qquad
B_2 \coloneqq -\,A_{KK},\qquad
B_3 \coloneqq -\,A_{\tau|K}.
\]
\paragraph{Robust block scales:}
For each $j\in\{1,2,3\}$, compute Euclidean row $\ell_2$ norms 
\[
r^{(j)}_g \coloneqq \| (B_j)_{g,:}\|_2 = \sqrt{\sum_{p=1}^{P}(B_j)^2_{gp}}, \qquad g=1,\dots,G
\]
Sorting the list in ascending order, pick the value $s_j$ below which $95\%$ of the $r^{(j)}_g$ fall. 

\begin{definition}
    The robust scale is defined as
    \[
    s_j \;\coloneqq\; \mathrm{q}_{0.95}\!\big(\{\,r^{(j)}_g\,:\,g=1,\dots,G\}\big),
    \]
    (Other robust choices are possible; $q_{0.95}$ balances outliers vs.\ typical rows.)
\end{definition}

\paragraph{Scaled blocks and unified weight:}
Define the scaled operators
\[
\widetilde B_j \;\coloneqq\; \frac{1}{s_j}\,B_j,\qquad j=1,2,3.
\]
Using a single $\lambda_{\mathrm{NA}}$ for all three terms, the soft no--arbitrage penalty becomes
\[
\frac{\lambda_{\mathrm{NA}}}{2}\sum_{j=1}^3 \big\|(\widetilde B_j a)_+\big\|_2^2
\;=\;
\frac{1}{2}\sum_{j=1}^3 \underbrace{\Big(\frac{\lambda_{\mathrm{NA}}}{s_j^2}\Big)}_{\text{effective weight for block $j$}}\;\big\|(B_j a)_+\big\|_2^2,
\]
so that the typical (p95) row magnitude of each block is $\approx1$ and one knob $\lambda_{\mathrm{NA}}$ moves all three violation shares on a comparable scale.

\begin{prop}[Invariance of hard constraints under positive scaling]
\label{prop:scaling-invariance}
Let $D$ be any positive diagonal matrix (in particular $D=\alpha I$ with $\alpha>0$). Then, for any $B\in\R^{G\times P}$ and any $a\in\R^P$,
\[
Ba\le 0 \;\Longleftrightarrow\; (DB)a\le 0.
\]
Hence replacing $B_j$ by $\widetilde B_j=\frac{1}{s_j}B_j$ leaves the hard no--arbitrage feasible set unchanged; only numerical conditioning and relative penalty weights are affected.
\end{prop}

\begin{proof}
All inequalities are understood componentwise.

Let $D = \operatorname{diag}(d_1,\dots,d_G)$ with $d_i > 0$ for all $i$. For any
$a \in \R^P$ we have
\[
(DB)a = D(Ba),
\]
so on the $i$-th component
\[
\bigl((DB)a\bigr)_i = d_i \,(Ba)_i.
\]

($\Rightarrow$) Suppose $Ba \le 0$. Then for every $i$,
\[
(Ba)_i \le 0 \quad\Longrightarrow\quad d_i (Ba)_i \le 0
\]
because $d_i > 0$. Hence $(DB)a = D(Ba) \le 0$.

($\Leftarrow$) Conversely, suppose $(DB)a \le 0$. Then for every $i$,
\[
d_i (Ba)_i = \bigl((DB)a\bigr)_i \le 0.
\]
Since $d_i > 0$, dividing by $d_i$ preserves the inequality sign and yields
\[
(Ba)_i \le 0 \quad\text{for all } i,
\]
i.e.\ $Ba \le 0$.

Thus $\{a : Ba \le 0\} = \{a : (DB)a \le 0\}$, proving the equivalence
$Ba \le 0 \;\Longleftrightarrow\; (DB)a \le 0$.

For the final claim, take $D = \frac{1}{s_j} I$ with $s_j > 0$ and $B$ replaced
by a given block $B_j$. Then
\[
B_j a \le 0 \;\Longleftrightarrow\; \Bigl(\frac{1}{s_j} I\, B_j\Bigr) a \le 0
\;\Longleftrightarrow\; \widetilde B_j a \le 0,
\]
so replacing $B_j$ by $\widetilde B_j = \frac{1}{s_j} B_j$ leaves the hard
no–arbitrage feasible set $\{a : B_j a \le 0\}$ (and hence the intersection over
all blocks $j$) unchanged. Only the numerical conditioning of the operators and
the effective relative weights in any soft penalties involving $B_j$ are affected.
\end{proof}

\begin{remark}[Exact recipe used in this paper]
\label{rem:our-scaling}
\leavevmode\par\vspace{0.25\baselineskip}
\begin{enumerate}[leftmargin=2em,itemsep=2pt,topsep=2pt]
\item \emph{Where scaling is applied.} We first apply any coefficient reparameterisation $U$ (price--invariant transform), i.e.\ replace each block by $A_\bullet U$. Scaling is computed and applied to these \emph{post-$U$} blocks.
\item \emph{Which blocks.} We scale $B_1=A_K$, $B_2=-A_{KK}$, $B_3=-A_{\tau|K}$ by \emph{one scalar per block}: $s_K,s_{KK},s_{\tau}$ given by the p95 of row $\ell_2$ norms on the collocation grid.
\item \emph{How it enters the objective.} The no--arb penalty uses the scaled operators $\widetilde B_j=B_j/s_j$ with a \emph{single} weight $\lambda_{\mathrm{NA}}$:
\[
\mathcal{P}_{\mathrm{NA}}(a)=\frac{\lambda_{\mathrm{NA}}}{2}\!\left(
\big\|(\widetilde A_K a)_+\big\|_2^2+
\big\|(-\widetilde A_{KK} a)_+\big\|_2^2+
\big\|(-\widetilde A_{\tau|K} a)_+\big\|_2^2\right).
\]
\item \emph{Bounds kept separate.} Price bounds $0\le Aa\le F$ are handled with a separate weight $\lambda_B$; we do \underline{not} include $A$ in the no--arb scaling group.
\item \emph{Reporting.} Diagnostics/violation shares are computed from the \underline{unscaled} physical operators $A_K, A_{KK}, A_{\tau|K}$.
\end{enumerate}
\end{remark}

\begin{remark}[Alternative (not used): row--by--row equalisation]
One may also scale \emph{each row} to equalise row influence by taking $D_j=\diag(d^{(j)}_g)$ with $d^{(j)}_g=1/\max(\|(B_j)_{g,:}\|_2,\varepsilon)$ and using $D_j B_j$. This preserves feasibility for hard constraints (Prop.~\ref{prop:scaling-invariance}) but reweights the grid non--uniformly. We \emph{do not} use this in our main results; we use the block--scalar scheme of Remark~\ref{rem:our-scaling}.
\end{remark}

\section{Summary (operators used in the optimiser)}
On the collocation grid, the no-arbitrage conditions become linear maps of $a$:
\[
\begin{array}{ll}
\text{strike monotonicity:} & A_K a \ \le\ 0,\\[2pt]
\text{strike convexity:} & -\,A_{KK} a \ \le\ 0,\\[2pt]
\text{calendar at fixed $K$:} & -\,A_{\tau|K} a \ \le\ 0,\\[2pt]
\text{bounds:} & 0\ \le\ A a \ \le\ F.
\end{array}
\]
Enforce these either as hard linear inequalities or as convex quadratic penalties on the positive parts, all while keeping the problem a single QP.

\chapter[Coverage-seeking data term with bid-ask geometry]%
  {Coverage-seeking data term with\\ bid-ask geometry}
\label{sec:coverage-data}

On date $t$, let $\{(m_i,\tau_i)\}_{i=1}^N$ be the quote locations, and let
\[
b_i\coloneqq \text{bid}_i,\qquad
a_i\coloneqq \text{ask}_i,\qquad
y_i\coloneqq \tfrac{1}{2}(\text{bid}_i+\text{ask}_i),
\]
be the forward–discounted band endpoints and mids ($0\le b_i\le a_i$ after standard cleaning).
Let $A\in\R^{N\times P}$ be the price design matrix so that $\widehat y(a)\coloneqq A a$ are model
prices at the quote points. Set heteroscedastic residual weights
\[
w_i \;=\; \frac{\mathrm{liq}_i}{\max(a_i-b_i,\varepsilon)^2},\qquad W\coloneqq \diag(w_1,\dots,w_N),
\]
where $\mathrm{liq}_i = 1 + \sqrt{\mathrm{volume}_i} + 0.1\sqrt{\mathrm{open\_interest}_i}$ and with a small floor $\varepsilon>0$.

\section{Loss components and their roles}
Use two convex terms:
\begin{equation}\label{eq:fit-loss}
\mathcal{L}_{\text{fit}}(a)
\;=\; \underbrace{\frac{1}{2}\,\| W^{1/2} (A a - y)\|_2^2}_{\text{within-band centre anchor}}
\;+\;
\underbrace{\mu \sum_{i=1}^N \ell_{\text{band}}\big((Aa)_i; b_i,a_i\big)}_{\text{coverage/Slack pricing}},
\end{equation}
where the \emph{quadratic band hinge} for a scalar $\hat y$ and interval $[b,a]$ is
\begin{equation}\label{eq:band-hinge}
\ell_{\text{band}}(\hat y; b,a)
\;=\;\frac{1}{2}\big(\max\{b-\hat y,0\}^2 + \max\{\hat y-a,0\}^2\big)
\;=\;\frac{1}{2}\,\mathrm{dist}\!\big(\hat y,[b,a]\big)^2.
\end{equation}
This means that $\ell_{\text{band}}(\hat y;b,a)=0$ iff $\hat y\in[b,a]$, and otherwise equals one–half the
squared Euclidean distance to the band. The first term in \eqref{eq:fit-loss} selects a point
\emph{inside} the band (preferentially near $y$ defined as the mid) whenever that is compatible with the other constraints;
the second term is a convex surrogate that drives \emph{coverage} by penalizing exactly the squared
violation outside the band.

\begin{lemma}[Convexity of band loss and fit objective]
\label{lem:band-convex}
Fix $b\le a$ and define $\ell_{\mathrm{band}}$ as in \eqref{eq:band-hinge}.
Then $\hat y\mapsto \ell_{\mathrm{band}}(\hat y;b,a)$ is a convex function on $\R$.
Consequently, for any $\mu\ge0$, any weight matrix $W\succeq 0$, design matrix
$A\in\R^{N\times P}$ and vector $y\in\R^N$, the loss
\[
\mathcal{L}_{\mathrm{fit}}(a)
= \frac12\big\|W^{1/2}(Aa-y)\big\|_2^2
+ \mu\sum_{i=1}^N \ell_{\mathrm{band}}\big((Aa)_i;b_i,a_i\big)
\]
is convex in $a\in\R^P$.
\end{lemma}
\begin{proof}
Write
\[
\ell_{\mathrm{band}}(\hat y;b,a)
= \frac12\Big(\max\{b-\hat y,0\}^2 + \max\{\hat y-a,0\}^2\Big).
\]

Each map $\hat y\mapsto b-\hat y$ and $\hat y\mapsto \hat y-a$ is affine, hence convex.
The hinge map $t\mapsto \max\{t,0\}$ is convex as a pointwise maximum of two affine
functions ($t$ and $0$). Therefore
\[
\hat y \mapsto \max\{b-\hat y,0\},
\qquad
\hat y \mapsto \max\{\hat y-a,0\}
\]
are convex functions. Moreover, both are nonnegative.

The square map $s\mapsto s^2$ is convex and nondecreasing on $[0,\infty)$.
The composition of a convex, nondecreasing function with a convex, nonnegative
function is convex. Hence
\[
\hat y \mapsto \max\{b-\hat y,0\}^2,
\qquad
\hat y \mapsto \max\{\hat y-a,0\}^2
\]
are convex, and so is their sum. Multiplication by $\tfrac12>0$ preserves convexity,
therefore $\ell_{\mathrm{band}}(\cdot;b,a)$ is convex.

For the second claim, the map $a\mapsto Aa-y$ is affine, $W^{1/2}$ is linear, and
$f(z)=\tfrac12\|z\|_2^2$ is convex; the composition $a\mapsto f\big(W^{1/2}(Aa-y)\big)$
is therefore convex. We also have just shown that $\hat y\mapsto \ell_{\mathrm{band}}(\hat y;b_i,a_i)$
is convex for each $i$. Composition with the affine map $a\mapsto (Aa)_i$ preserves
convexity, so $a\mapsto \ell_{\mathrm{band}}\big((Aa)_i;b_i,a_i\big)$ is convex for all $i$.
Summation over $i$ and scaling by $\mu\ge0$ preserve convexity. Adding the two convex
terms yields that $\mathcal{L}_{\mathrm{fit}}$ is convex in $a$.
\end{proof}

\begin{remark}[Optional dead–zone/margin]
To avoid hugging the band edges, one may widen the interior by a margin $\delta_i\ge 0$ and
replace $[b_i,a_i]$ with $[b_i+\delta_i,a_i-\delta_i]$ in \eqref{eq:band-hinge}. All results below are unchanged.
\end{remark}

\section{Quadratic–program form via auxiliary slacks}\label{sec:band-hinge-qp}
While \eqref{eq:fit-loss} is already convex in $a$, it is possible to cast it as a QP with \emph{only} a quadratic
objective and linear constraints. Introduce non negative slacks $(u_i,v_i)$ per quote:
\begin{equation}\label{eq:hinge-qp}
\ell_{\text{band}}(\hat y;b,a)
\;=\;
\min_{u,v\ge 0}\ \frac{1}{2}(u^2+v^2)\quad
\text{s.t.}\quad u\ge b-\hat y,\ \ v\ge \hat y-a.
\end{equation}

\begin{lemma}[Exact equivalence of \eqref{eq:band-hinge} and \eqref{eq:hinge-qp}]
\label{lem:hinge-eq}
For any $b\le a$ and any $\hat y\in\R$, the optimal slacks in \eqref{eq:hinge-qp} are
$u^\star=(b-\hat y)_+$ and $v^\star=(\hat y-a)_+$, and the optimal value equals
$\tfrac{1}{2}[(b-\hat y)_+^2+(\hat y-a)_+^2]=\ell_{\text{band}}(\hat y;b,a)$.
\end{lemma}
\begin{proof}
If $\hat y\in[b,a]$, feasibility with $u=v=0$ gives value $0$; nonnegativity enforces $u=v=0$ at optimum.
If $\hat y<b$, the constraints reduce to $u\ge b-\hat y>0$ and $v\ge 0$, so the quadratic objective
is minimised at $(u^\star,v^\star)=(b-\hat y,0)$. The case $\hat y>a$ is symmetric.
\end{proof}

Stacking \eqref{eq:hinge-qp} over quotes yields the QP
\begin{equation}\label{eq:data-qp}
\min_{a,u,v}\ 
\frac{1}{2}\|W^{1/2}(Aa-y)\|_2^2
+\frac{\mu}{2}\big(\|u\|_2^2+\|v\|_2^2\big)
\quad\text{s.t.}\quad
\begin{cases}
u\ge b-Aa,\ \ u\ge 0,\\
v\ge Aa-a,\ \ v\ge 0,
\end{cases}
\end{equation}
where all inequalities are coordinate-wise. 

\begin{remark}[KKT and projection viewpoint]
At the solution of \eqref{eq:hinge-qp} for a fixed $\hat y$, $(u^\star,v^\star)$ is precisely the
vector of signed violations projected onto the nonnegative orthant; equivalently,
$\sqrt{2\,\ell_{\text{band}}(\hat y;b,a)}=\mathrm{dist}(\hat y,[b,a])$.
Thus the second term in \eqref{eq:fit-loss} is $\tfrac{\mu}{2}\|\mathrm{dist}(Aa,[b,a])\|_2^2$ (coordinatewise distance).
\end{remark}

\subsection*{Strict convexity and uniqueness of the data QP}\label{sec:strict-convex-data}

\begin{definition}[Positive definiteness on the span of $A$]
We say that the quadratic form $Q(a)=\tfrac12\,a^\top A^\top W A\,a$ is \emph{positive definite on the span of $A$} if
\[
a\neq 0\ \text{ and }\ Aa\neq 0 \quad\Longrightarrow\quad a^\top A^\top W A\,a=(Aa)^\top W (Aa)>0.
\]
In particular, $Q$ is strictly convex in the prediction variable $p:=Aa$, and in coefficient space its only flat directions are those in $\ker(A)$: for each fixed $p$ the restriction of $Q$ to the affine fibre $\{a:\ Aa=p\}$ is constant.
\end{definition}

\begin{prop}[Strict convexity $\Rightarrow$ uniqueness]\label{prop:strict-convex-data}
Assume $W\succ 0$ (symmetric positive definite, ie $x^TWx>0$ for every non-zero vector $x$) and $\mu>0$. If $A^\top W A$ is positive definite on the span of $A$ (in particular, if $A$ has full column rank), then the objective of \eqref{eq:data-qp} is strictly convex in $(a,u,v)$, and hence \eqref{eq:data-qp} has a unique optimiser $(a^\star,u^\star,v^\star)$ whenever the feasible set is nonempty. Here and throughout, uniqueness is understood modulo the nullspace of $A$: if $A$ is rank-deficient and $(a_1,u_1,v_1)$ and $(a_2,u_2,v_2)$ are both optimal solutions of \eqref{eq:data-qp}, then $Aa_1 = Aa_2$, $u_1 = u_2$, $v_1 = v_2$, and $a_2 - a_1 \in \ker(A)$. Moreover, if a ridge term $\tfrac{\lambda}{2}\|a\|_2^2$ with $\lambda>0$ is added, the objective is strictly convex \emph{unconditionally} (regardless of $\mathrm{rank}(A)$), yielding uniqueness of the optimiser.
\end{prop}

\begin{proof}

The slack QP has decision variables $z=(a, u, v)$ and objective
\[
F(a, u, v)  = \min_{a,u,v}\ 
\frac{1}{2}\|W^{1/2}(Aa-y)\|_2^2
+\frac{\mu}{2}\big(\|u\|_2^2+\|v\|_2^2\big)
\]
Expanding the first term
\[
\frac{1}{2} (Aa-y)^\top W (Aa-y)=\frac{1}{2}a^\top A^\top WAa-y^\top WAa+\frac{1}{2}y^\top Wy
\]
Firstly we can see that $\frac{1}{2}y^TWy$ does not depend on the decision variable, so it is a constant. A quadratic function can be written as:
\[
q(z)=\frac{1}{2}z^\top Hz+c^\top z+constant
\]
with $H$ symmetric. We can see from the expansion that $F(a,u,v)$ is a quadratic function of $(a,u,v)$ with Hessian
\[
H \;=\;
\begin{bmatrix}
A^\top W A & 0 & 0\\[2pt]
0 & \mu I & 0\\[2pt]
0 & 0 & \mu I
\end{bmatrix}.
\]
Also note that from the equation, the constant is $\frac{1}{2}y^TWy$ and 
\[
C \;=\;
\begin{bmatrix}
-A^\top W y\\[2pt]
0\\[2pt]
0
\end{bmatrix}.
\]
Since $\mu>0$, the $u$- and $v$-blocks are positive definite (namely $\mu I\succ 0$). 

For the $a$-block, take any direction $\delta a \in \R^P$, and define the prediction perturbation
\[
\delta p =A \delta a \in \R^N
\]
Then
\[
\delta a^\top A^\top W A\,\delta a = (A\delta a)^\top W (A\delta a) = (\delta p)^\top W (\delta p) = \|\delta p\|_{W}^2,
\]
where $\|z\|_{W}^2:=z^\top W z$ is the weighted Euclidean norm (since $W \succ 0$ it is in fact a norm. The interpretation is that curvature in the $a$-block aling $\delta a$ equals the weighted squared change in predictions produced by that $\delta a$:
\begin{itemize}
    \item If $A \delta a \neq 0$, predictions move and the term is $>0$
    \item If $A \delta a = 0$, predictions don't move and the term is $=0$
\end{itemize}
The set $ker(A)=\{ \delta a : A\delta a =0 \}$ is the nullspace (directions in coefficient space that leave predictions unchanged).

By the assumption “positive definite on the span of $A$”, $\|A\delta a\|_{W}^2>0$ for every $\delta a$ with $A\delta a\neq 0$; hence along any nonzero direction $(\delta a,\delta u,\delta v)$ with $(\delta u,\delta v)\neq 0$ or $A\delta a\neq 0$ we have
\[
(\delta a,\delta u,\delta v)^\top H\,(\delta a,\delta u,\delta v)
= \|A\delta a\|_{W}^2 + \mu\|\delta u\|_2^2 + \mu\|\delta v\|_2^2 \;>\; 0.
\]
Thus the objective is strictly convex on $\mathbb{R}^{P}\times \mathbb{R}^{N}\times \mathbb{R}^{N}$ modulo the trivial flat directions $\delta a\in\ker(A)$ with $\delta u=\delta v=0$. If $A$ has full column rank, $\ker(A)=\{0\}$ and $H\succ 0$, so the objective is strictly convex in $(a,u,v)$. A strictly convex objective over a convex (polyhedral) feasible set admits at most one minimiser; feasibility of \eqref{eq:data-qp} then yields uniqueness.

If a ridge term $\tfrac{\lambda}{2}\|a\|_2^2$ with $\lambda>0$ is added, the Hessian becomes
\[
H_\lambda \;=\;
\begin{bmatrix}
A^\top W A + \lambda I & 0 & 0\\[2pt]
0 & \mu I & 0\\[2pt]
0 & 0 & \mu I
\end{bmatrix}\ \succ\ 0,
\]
which is positive definite regardless of $\mathrm{rank}(A)$, hence the objective is strictly convex in $(a,u,v)$ and the minimiser is unique.
\end{proof}

\begin{remark}[What is unique when $A$ is rank-deficient]
If $A$ is rank-deficient and no ridge is used, the objective is strictly convex in the \emph{predictions} $p:=Aa$ and in $(u,v)$, but flat along $\ker(A)$. Consequently, the optimiser’s predictions $p^\star=Aa^\star$ and slacks $(u^\star,v^\star)$ are unique, while $a^\star$ is unique only up to additions by vectors in $\ker(A)$. Adding a small ridge fixes $a^\star$ uniquely.
\end{remark}

\section{Weights, units, and invariance}
The choice $w_i\propto (a_i-b_i)^{-2}$ makes the mid–squared error scale–free with respect to the
local band width; and the multiplicative factor $\mathrm{liq}_i$ up-weights more reliable quotes. The hinge
term already measures squared band distance, so $\mu$ is dimensionless. If one rescales all prices by a
factor $c>0$, then $A\!\gets cA$, $y\!\gets cy$, $b\!\gets cb$, $a\!\gets ca$; the minimiser of \eqref{eq:data-qp}
is unchanged after dividing $\mu$ by $c^2$ and multiplying $W$ by $c^{-2}$—this is the standard
homogeneity of quadratic objectives.

\section{Binned variant (optional)}
To stabilise sparse regions, let $G\in\{0,1\}^{B\times N}$ be a selector that sums quotes in $(m,\tau)$ bins.
Replacing per–quote hinge terms by binned terms yields
\begin{equation}\label{eq:binned-hinge}
\sum_{i=1}^N \ell_{\text{band}}\big((Aa)_i;b_i,a_i\big)
\;\leadsto\;
\sum_{b=1}^B \ell_{\text{band}}\Big( (G Aa)_b;\ (G b)_b,\ (G a)_b\Big),
\end{equation}
which is still a QP by Lemma~\ref{lem:hinge-eq}, with slacks now attached to bins. The binned form
penalises average violations in each cell and reduces sensitivity to isolated outliers.

\section{Feasibility and the role of \texorpdfstring{$\mu$}{mu}}
Let $\mathcal{S}_{\text{band}}=\{a\in\R^P:\ b\le Aa\le a\}$ be the band–feasible set (coordinate-wise).
\begin{itemize}[leftmargin=1.25em]
\item If $\mathcal{S}_{\text{band}}\neq\emptyset$ and other constraints (no–arb penalties or hard
inequalities) admit a feasible intersection, then taking $\mu\to\infty$ in \eqref{eq:data-qp}
forces $Aa$ into the band while the WLS term selects the point closest to $y$ among the band–feasible reconstructions.
\item If the intersection is empty, \eqref{eq:data-qp} finds the unique pair $(a,u,v)$ that minimises the
$\ell_2$–distance of $Aa$ to the rectangle $[b,a]$ while trading off the mid anchor through $W$.
\end{itemize}

\begin{remark}[Targeting coverage]
Define coverage$(a)=\frac{1}{N}\sum_i \mathbf{1}\{b_i\le (Aa)_i\le a_i\}$. Increasing $\mu$ reduces hinge violations and typically (empirically) increases coverage; we adjust $\mu$ with a short controller to hit a target coverage level. Formal monotonicity in $\mu$ is not required for the optimiser or the QP structure.
\end{remark}

\section{What is used in this paper (precise choices)}
\begin{enumerate}[leftmargin=2em,itemsep=2pt]
\item \textbf{Forward–discounted bands and mids:} $(b_i,a_i,y_i)$ constructed at each quote and robustly cleaned so $0\le b_i\le a_i$.
\item \textbf{Weights:} $w_i=\mathrm{liq}_i/\max(a_i-b_i,\varepsilon)^2$ with $\varepsilon$ a small fixed floor; $W=\diag(w)$.
\item \textbf{Band hinge:} quadratic $\ell_{\text{band}}$ as in \eqref{eq:band-hinge}; no interior margin unless stated (set $\delta_i=0$ by default).
\item \textbf{QP form:} auxiliary slacks $(u,v)\ge 0$ with linear constraints \eqref{eq:data-qp}, solved jointly with the rest of the QP (ridge, no–arb penalties, etc.).
\item \textbf{Optional binning:} $G$–aggregation in \eqref{eq:binned-hinge} enabled on sparse books; otherwise per–quote hinge.
\item \textbf{Controller for $\mu$:} simple scheduler that increases $\mu$ until the observed coverage reaches the target (with caps); WLS weight $W$ is held fixed across the schedule.
\end{enumerate}

\chapter[Ridge, spectral geometry, and transport regularisation]%
  {Ridge, spectral geometry, and\\ transport regularisation}
\label{sec:ridge-etc}
This chapter specifies the quadratic regularisers added to the objective, and the
price invariant reparameterisation used to improve conditioning. Every term below is a
fixed quadratic form in the coefficient vector $a\in\R^P$ (or in a linear reparameterisation
$\tilde a$), so the overall problem remains a convex QP.

\section{Spectral ridge (modal energy control)}\label{sec:ridge-gcv}
Let the tensor index be $(k,\ell)$ with $k=0,\dots,K$ (log–moneyness) and $\ell=0,\dots,L$
(maturity). Define a diagonal weight 
\[
\Lambda_{(k,\ell),(k,\ell)} \;=\; \big(1+\alpha\,k^2+\beta\,\ell^2\big)^{\,s},
\qquad \alpha,\beta>0,\ s>0,
\]
and set $\Lambda=\diag(\Lambda_{(k,\ell),(k,\ell)})\in\R^{P\times P}$ with $P=(K{+}1)(L{+}1)$.
The spectral ridge is defined as
\begin{equation}
\label{eq:ridge}
\mathcal{R}_{\mathrm{ridge}}(a)
\;=\; \frac{\lambda_{\mathrm{ridge}}}{2}\,\| \Lambda^{1/2} a \|_2^2
\;=\; \frac{\lambda_{\mathrm{ridge}}}{2}\,a^\top \Lambda a
\;=\; \frac{\lambda_{\mathrm{ridge}}}{2}\sum_{k=0}^{K}\sum_{l=0}^{L}(1+\alpha\,k^2+\beta\,\ell^2\big)^{\,s} a_{kl}^2.
\end{equation}

\paragraph{Interpretation:} Each coefficient is penalised by a weight that grows with its modal index. Low modes (small $k$ and $l$) get weight $\approx 1$; higher $(k,\ell)$ modes carry larger weights (are expensive). $\alpha,\beta$ tune the relative penalisation across $m$ vs.\ $\tau$, and $s$ controls the growth rate (asymptotically the weights grow like $(\alpha k^2 + \beta \ell^2)^s$, so $s=1$ gives quadratic growth in the indices and $s=2$ gives quartic growth). This damps high–frequency oscillations while leaving low modes essentially unchanged. In spectral methods, smooth functions have rapidly decaying coefficients and non-smooth noisy features push energy into high indices. Penalising $a_{kl}^2$ with a weight increasing in $k, l$ is the discrete analogue of a Sobolev $H^s$ seminorm, suppressing high frequency components components while leaving low modes mostly alone. The wraps $\Phi_m$ and $\Phi_\tau$ mean smoothness is enforced in the wrapped coordinates where the basis is well-conditioned (ATM focus and short-$\tau$ density).

Unless otherwise stated, we fix $\alpha=\beta=1$ and $s=2$. The scalar
$\lambda_{\mathrm{ridge}}$ is chosen once per date by a small-subsample generalised
cross–validation (GCV) pass on the \emph{linear} WLS subproblem. Build $A_\mathrm{sub}$ and $W_\mathrm{sub}$ on a random $8\%$ subset of quotes. Namely $A_{sub}\in \R ^{N_{sub}\times P}$, $W_{sub}=diag(w_{sub}) \succ 0$ and $y_{sub}\in \R ^{N_{sub}}$. For any $\lambda > 0$, solve the ridge-regularised weighted least squares:
\[
\min_{a}\;\frac12\|W_\mathrm{sub}^{1/2}(A_\mathrm{sub}a-y_\mathrm{sub})\|_2^2
+\frac{\lambda}{2}\| \Lambda^{1/2} a\|_2^2,
\]
This has the closed form
\[
a(\lambda)= (A_{\mathrm{sub}}^\top W_{\mathrm{sub}} A_{\mathrm{sub}}+\lambda \Lambda)^{-1}A_{\mathrm{sub}}^\top W_{\mathrm{sub}} y_{\mathrm{sub}}
\]
Define the weighted residual and the hat matrix
\[
r(\lambda)=W_\mathrm{sub}^{1/2}(A_\mathrm{sub}a(\lambda)-y_\mathrm{sub}); \quad H(\lambda)=W_\mathrm{sub}^{1/2}A_\mathrm{sub}(A_\mathrm{sub}^\top W_\mathrm{sub}A_\mathrm{sub}
+\lambda\Lambda)^{-1}A_\mathrm{sub}^\top W_\mathrm{sub}^{1/2}.
\]

\begin{remark}
    A useful identity for computation is 
    \[
    tr( H(\lambda)) = tr \left( SG(\lambda)^{-1} \right), \quad S=(A_{\mathrm{sub}}^\top W_{\mathrm{sub}} A_{\mathrm{sub}}), \quad G(\lambda)=S+\lambda \Lambda .  
    \]
\end{remark}

We compute the GCV (Generalised cross-validation) score 
\[
\mathrm{GCV}(\lambda)=\frac{\|r(\lambda)\|_2^2}{\big(N_\mathrm{sub}-\mathrm{tr}\,H(\lambda)\big)^2}.
\]
Choose
$\lambda_{\mathrm{ridge}}=\arg\min_\lambda\mathrm{GCV}(\lambda)$ and use it in the full QP.

\begin{lemma}[Spectral ridge is a fixed quadratic form]
\label{lem:ridge-quadratic}
Fix the tensor grid $(k,\ell)$ with $k=0,\dots,K$, $\ell=0,\dots,L$, and
hyperparameters $\alpha,\beta>0$, $s>0$. Let $\Lambda\in\R^{P\times P}$,
$P=(K{+}1)(L{+}1)$, be the diagonal matrix defined above and let
$\lambda_{\mathrm{ridge}}\ge 0$ be fixed. Then for every $a\in\R^P$,
the spectral ridge \eqref{eq:ridge} can be written as
\[
\mathcal{R}_{\mathrm{ridge}}(a)
\;=\; \frac{1}{2}\,a^\top Q_{\mathrm{ridge}} a,
\qquad
Q_{\mathrm{ridge}} := \lambda_{\mathrm{ridge}}\,\Lambda,
\]
with $Q_{\mathrm{ridge}}$ symmetric positive semidefinite and independent of $a$.
In particular, $\mathcal{R}_{\mathrm{ridge}}$ is a convex quadratic function of
the coefficient vector $a$, and enters any optimisation problem as a fixed
quadratic form (for given $(\alpha,\beta,s)$ and $\lambda_{\mathrm{ridge}}$).
\end{lemma}
\begin{proof}
By definition,
\[
\Lambda_{(k,\ell),(k,\ell)}
= \bigl(1+\alpha k^2+\beta \ell^2\bigr)^s > 0
\quad\text{for all } 0\le k\le K,\ 0\le \ell\le L,
\]
so $\Lambda$ is diagonal with strictly positive diagonal entries and hence
$\Lambda\succeq 0$ (indeed $\Lambda\succ 0$). For any $a\in\R^P$,
\[
\|\Lambda^{1/2}a\|_2^2
= a^\top \Lambda^{1/2}\Lambda^{1/2} a
= a^\top \Lambda a.
\]
Thus \eqref{eq:ridge} can be rewritten as
\[
\mathcal{R}_{\mathrm{ridge}}(a)
= \frac{\lambda_{\mathrm{ridge}}}{2}\,a^\top \Lambda a
= \frac{1}{2}\,a^\top Q_{\mathrm{ridge}} a,
\qquad
Q_{\mathrm{ridge}} := \lambda_{\mathrm{ridge}}\Lambda.
\]
The matrix $Q_{\mathrm{ridge}}$ is symmetric. Since $\lambda_{\mathrm{ridge}}\ge 0$
and $\Lambda\succeq 0$, we have $Q_{\mathrm{ridge}}\succeq 0$, so the map
$a\mapsto \tfrac12 a^\top Q_{\mathrm{ridge}} a$ is a convex quadratic function.
For fixed hyperparameters $(\alpha,\beta,s)$, grid sizes $(K,L)$, and a chosen
value of $\lambda_{\mathrm{ridge}}$, the matrix $Q_{\mathrm{ridge}}$ is completely
determined and does not depend on $a$. Hence $\mathcal{R}_{\mathrm{ridge}}$ is a
fixed quadratic form in the coefficient vector.
\end{proof}

\section{\texorpdfstring{$\Lambda$–module: price–invariant reparameterisation}
                     {Lambda–module: price–invariant reparameterisation}}
\label{sec:Lambda-module}
Now propose a change of coordinates in the coefficient space to make the optimisation numerically well-behaved.  Since the variables and matrices are changed in a consistent way, all prices and constraint values stay identical, and the conditioning of the problem improves.

Let $U\in\R^{P\times P}$ be a fixed, invertible linear map. Define new coefficients
$\tilde a:=U^{-1}a$ and replace every block by post–multiplication with $U$:
\[
A\leftarrow A U,\quad
A_m\leftarrow A_m U,\quad
A_{mm}\leftarrow A_{mm} U,\quad
A_\tau\leftarrow A_\tau U,\quad
\text{etc.}
\]
Predictions are unchanged: $(AU)\tilde a=A(U\tilde a)=Aa$. The ridge becomes
\[
\mathcal{R}_{\mathrm{ridge}}(a)
=\frac{\lambda_{\mathrm{ridge}}}{2}\,\| \Lambda^{1/2} U \tilde a \|_2^2
=\frac{\lambda_{\mathrm{ridge}}}{2}\,\tilde a^\top \underbrace{U^\top \Lambda U}_{\widetilde\Lambda}\,\tilde a,
\]
i.e.\ the same quadratic form in $\tilde a$ with $\widetilde\Lambda=U^\top\Lambda U$.

\paragraph{Blockwise whitening and why it is safe:}
Firstly, partition columns by maturity slice $\ell$ (all $m$–modes for that $\ell$) and define the sets $\{\mathcal G_\ell\}$.
Then for each block form the \emph{weighted} thin QR
\[
W^{1/2}A_{[:,\mathcal G_\ell]} \;=\; Q_\ell R_\ell,\qquad Q_\ell^\top Q_\ell=I,\ R_\ell \text{ invertible upper–triangular},
\]
and assemble a block–diagonal $U$ with $U_{\mathcal G_\ell,\mathcal G_\ell}:=R_\ell^{-1}$ (zeros off–block). Optionally, right–scale
columns so $\|W^{1/2}(A U)_{[:,j]}\|_2=1$ by replacing $U\leftarrow U D^{-1}$ with $D=\diag(d_j)$, $d_j=\|W^{1/2}(AU)_{[:,j]}\|_2$.

We implicitly require that each slice $A[:,G_\ell]$ have full column rank under the
$W$–inner product, so that the thin QR factorisation with a square, invertible
$R_\ell$ exists. This condition is satisfied for the Chebyshev grids used in our
experiments; in degenerate cases one can replace $R_\ell^{-1}$ by a pseudo–inverse
obtained from a rank–revealing QR or SVD, at the price of a slightly lower–dimensional
reparameterisation.

\begin{prop}[W–orthonormality within slices]\label{prop:whiten-block}
With $U$ defined above,
\[
(AU)_{[:,\mathcal G_\ell]} \;=\; W^{-1/2}Q_\ell
\quad\Longrightarrow\quad
(AU)_{[:,\mathcal G_\ell]}^\top W\,(AU)_{[:,\mathcal G_\ell]} \;=\; I_{|\mathcal G_\ell|}.
\]
In particular, the WLS normal matrix becomes block–identity within each slice (collinearity removed).
\end{prop}
\begin{proof}
$A_{[:,\mathcal G_\ell]}=W^{-1/2}Q_\ell R_\ell$ and $U_{\mathcal G_\ell,\mathcal G_\ell}=R_\ell^{-1}$ give
$(AU)_{[:,\mathcal G_\ell]}=W^{-1/2}Q_\ell$. 

Hence $(AU)_{[:,\mathcal G_\ell]}^\top W (AU)_{[:,\mathcal G_\ell]}=Q_\ell^\top Q_\ell=I$.
\end{proof}

\begin{prop}[Price/constraint invariance]\label{prop:inv}\label{prop:equiv}
Let $\tilde a=U^{-1}a$, $A':=AU$, and for any block $A_\bullet$ set $A_\bullet':=A_\bullet U$. Then
\[
A'\tilde a=Aa,\qquad A_\bullet'\tilde a=A_\bullet a.
\]
Consequently, hard inequalities $A_\bullet a\le 0$ are equivalent to $(A_\bullet U)\tilde a\le 0$, and soft penalties that depend on $A_\bullet a$ take the same values when written in $\tilde a$.
\end{prop}
\begin{proof}
We see by definition of $A'$ and $\tilde a$, $A'\tilde a=AUU^{-1}a=Aa$ as required. 
\end{proof}

\begin{prop}[Ridge congruence]\label{prop:ridge-cong}
For any symmetric $Q\succeq 0$,
\[
\frac12\,a^\top Q a \;=\; \frac12\,\tilde a^\top (U^\top Q U)\,\tilde a.
\]
In particular, the spectral ridge becomes $\frac{\lambda_{\mathrm{ridge}}}{2}\,\tilde a^\top \widetilde\Lambda\,\tilde a$
with $\widetilde\Lambda:=U^\top \Lambda U\succeq 0$.
\end{prop}
\begin{proof}
Substitute $a=U\tilde a$ and regroup; congruence preserves positive semidefiniteness.
\end{proof}

\begin{remark}[Global whitening as a special case]
If one QR–factorises $W^{1/2}A=QR$ once and sets $U=R^{-1}$, then $(AU)^\top W (AU)=I_P$ (full whitening); the blockwise construction above is its per–slice counterpart.
\end{remark}

\section{\texorpdfstring
  {DW--module: discrete transport ($H^{-1}$) smoothing of density}
  {DW--module: discrete transport H-1 smoothing of density}}
\label{sec:dw-mod}

This section penalises oscillations of the risk-neutral density $\rho$ along $m$ by measuring how much potential $\phi$ is needed so its discrete derivative matches $\rho$. High frequency wiggles are expensive and slowly varying shapes cost little.

Let $\rho=\partial_{KK} C_f$ denote the risk–neutral density. On each maturity slice
$\tau=\tau_g$, we discretise the $m$–axis on the collocation nodes $m_{j}$ and build:
\begin{itemize}[leftmargin=1.5em]
\item a diagonal mass matrix $M_m=\diag(w^{(m)})$ with Gauss–Lobatto (Chebyshev) quadrature weights.
\item a first–difference matrix $D_m\in\R^{(M_m-1)\times M_m}$ (forward differences with
homogeneous Neumann boundary, i.e. zero–flux ends) with $(D_m \phi )_i=\phi_{i+1}-\phi_i $ for $i=1,...,M_m-1$.
\end{itemize}

Define the discrete Neumann Laplacian in $m$ by
\[
L_m \;:=\; D_m^\top M_m^{-1} D_m \;\in\; \R^{M_m\times M_m}.
\]
$L_m$ is symmetric positive semidefinite and its nullspace is the span of the constant vector
(along each slice). For a discrete function $f\in\R^{M_m}$, the discrete $H^{-1}(m)$
seminorm is defined by
\[
\|f\|_{H^{-1}(m)}^2 \;:=\; f^\top \underbrace{L_m^{+}}_{\text{Moore--Penrose pseudoinverse of }L_m} f,
\]
where $^{+}$ is the Moore–Penrose pseudoinverse on the range.
This is the standard discrete Neumann $H^{-1}$ seminorm: $L_m^{+}$ plays the role of
the inverse Laplacian, so only the mean–zero component of $f$ is penalised and the
constant/mean mode lies in the nullspace.

Assemble the full grid operator
$L^{+}=\mathrm{blkdiag}(L_m^{+},\dots,L_m^{+})$ across slices and the sampling matrix
$S:\R^G\to\R^{M_m(M_\tau{+}1)}$ that reshapes grid values into slice stacks. With
$E:=S\,A_{KK}\in\R^{M_m(M_\tau{+}1)\times P}$ (density map in slice–stacked order), the DW penalty is
\begin{equation}
\label{eq:dw-penalty}
\mathcal{R}_{\mathrm{DW}}(a)
\;=\; \frac{\lambda_{\mathrm{DW}}}{2}\,\| \rho(a)\|_{H^{-1}}^2
\;=\; \frac{\lambda_{\mathrm{DW}}}{2}\, a^\top E^\top L^{+} E\, a,
\end{equation}
a fixed quadratic form once $L^{+}$ is precomputed (e.g. Cholesky on each $L_m$ on the
mean–zero subspace, plus a rank–1 fix for the constant nullspace).

In Fourier language, $\|f\|_{H^{-1}(m)}^2 \sim \sum _k |f_k|^2/k^2$. The high $k$ content is amplified, so the optimiser prefers smooth densities. The constant/mean component sits in the nullspace and is not penalised. The constraint set continues to control positivity/monotonicity and DW just damps ripples that those constraints do not eliminate.

We apply \eqref{eq:dw-penalty} \emph{along $m$ only} on each slice, with Neumann
boundaries and Chebyshev–Lobatto weights. The difference operator $D_m$ only takes
interior forward differences,
\[
(D_m \phi )_i=\phi_{i+1}-\phi_i \qquad i=1,\dots,M_m-1,
\]
so the associated Laplacian $L_m = D_m^\top M_m^{-1} D_m$ has the constant vector in
its nullspace: adding a constant to $\phi$ does not change $D_m\phi$ or the quadratic
form. In the $H^{-1}$ penalty, this means the mean component of $\rho$ is left
unpenalised and only fluctuations around the mean contribute to $\|\rho\|_{H^{-1}}$;
this is the discrete zero–flux (Neumann) condition at the ends.

Equivalently, on each slice we can view $\rho$ as a one–dimensional ``charge
distribution'' along $m$. The matrix $L_m$ is a discrete Neumann Laplacian on
the nodes, and $L_m^{+}$ acts as its inverse on mean–zero densities. For any
profile $f\in\R^{M_m}$ with zero $M_m$–mean there exists a potential
$\phi$ (unique up to an additive constant) solving the discrete Poisson problem
\[
L_m \phi = f \qquad \text{(Neumann in $m$)}.
\]
Among all such potentials, the one with the smallest discrete Dirichlet energy
$\phi^\top L_m \phi$ satisfies
\[
\phi^\top L_m \phi \;=\; f^\top L_m^{+} f \;=\; \|f\|_{H^{-1}(m)}^2.
\]
So $\|f\|_{H^{-1}(m)}^2$ measures how much ``bending'' of the potential $\phi$
is needed to support the density: sharply oscillating $f$ requires a highly
curved potential and incurs a large penalty, while slowly varying $f$ can be
supported by a gentle potential and is cheap. The constant component of $f$
generates no potential at all and is left unpenalised.

Chebyshev-Lobatto nodes cluster near the endpoints, and without quadrature weights, any discrete $L^2$ inner product would outweigh the ends and underweight the middle. $M_m = diag(w^{(m)})$ fixes that. For a smooth function $g$,
\[
\sum_{j=1}^{M_m} w_j^{(m)}g(m_j) \approx \int_{m_{min}}^{m_{max}}g(m) dm.
\]
So $\phi ^\top M_m \phi$ is a proper discretisation of $\int \phi(m)^2dm$, independent of how densely sampled near the ends.
The weights are obtained using standard Clenshaw-Curtis (Gauss-Lobatto) quadrature on the Chebyshev-Lobatto nodes in the reference variable $n \in [-1,1]$ and then rescaled to $m \in [m_{min}, m_{max}]$. 

If $m=\frac{m_{max}-m_{min}}{2}n+\frac{m_{max}+m_{min}}{2}$, the Jacobian is constant and 
\[
w_j^{(m)}= \frac{m_{max}-m_{min}}{2} w_j^{(n)}.
\]
The weight is tapered in maturity:
$\lambda_{\mathrm{DW}}(\tau_g)=\lambda_{\mathrm{DW}}^{(0)}\cdot \min\{1,\tau_\star/\tau_g\}$ to
dampen short maturity ripples; defaults $\tau_\star=5$ trading days. For $\tau_g \leq \tau_*$, we use the full DW smoothing strength $\lambda_{\mathrm{DW}}^{(0)}$ and for $\tau_g > \tau_*$ the smoothing weight decays like $1/\tau_g$ so that the long end is not over-smoothed.

\begin{lemma}[DW penalty as a fixed quadratic form]
\label{lem:dw-quadratic}
Fix the grid operators
\[
  L^+ \in \R^{M_m(M_\tau+1)\times M_m(M_\tau+1)},
  \qquad
  E := S A_{KK} \in \R^{M_m(M_\tau+1)\times P},
\]
as above, and let
$\lambda_{\mathrm{DW}}\ge 0$ be given. Then for every $a\in\R^P$, let
\[
\mathcal{R}_{\mathrm{DW}}(a)
\;=\; \frac{\lambda_{\mathrm{DW}}}{2}\,\|\rho(a)\|_{H^{-1}}^2
\;=\; \frac{1}{2}\, a^\top Q_{\mathrm{DW}} a,
\qquad
Q_{\mathrm{DW}} := \lambda_{\mathrm{DW}}\,E^\top L^{+} E,
\]
with $Q_{\mathrm{DW}}$ symmetric positive semidefinite and independent of $a$.
Then $\mathcal{R}_{\mathrm{DW}}$ is a convex quadratic function of the
coefficient vector $a$ and enters the global problem as a fixed quadratic term,
so the formulation remains a convex QP.
\end{lemma}
\begin{proof}
By construction, $\rho(a)$ is linear in $a$: on the slice-stacked grid,
$\rho(a) = Ea$ with $E = S A_{KK}$ independent of $a$. The discrete
$H^{-1}$ seminorm is
\[
\|\rho(a)\|_{H^{-1}}^2
= \rho(a)^\top L^+ \rho(a)
= (Ea)^\top L^+ (Ea)
= a^\top E^\top L^+ E a.
\]
Thus
\[
\mathcal{R}_{\mathrm{DW}}(a)
= \frac{\lambda_{\mathrm{DW}}}{2}\,a^\top E^\top L^+ E a
= \frac{1}{2}\,a^\top Q_{\mathrm{DW}} a
\]
with $Q_{\mathrm{DW}} := \lambda_{\mathrm{DW}} E^\top L^+ E$. The operator $L^+$ is symmetric positive semidefinite by construction, as the
Moore–Penrose pseudoinverse of the symmetric positive semidefinite block–diagonal
matrix whose blocks are $L_m = D_m^\top M_m^{-1} D_m$. Hence for any $z$,
$z^\top L^+ z \ge 0$ (by definition of symmetric positive semidefinite), and in particular
\[
a^\top Q_{\mathrm{DW}} a
= \lambda_{\mathrm{DW}} (Ea)^\top L^+ (Ea)
\;\ge\; 0
\]
for all $a$, so $Q_{\mathrm{DW}}\succeq 0$ whenever $\lambda_{\mathrm{DW}}\ge 0$.
All ingredients $(E,L^+,\lambda_{\mathrm{DW}})$ are fixed once the grid,
quadrature, and smoothing weight are chosen, so $Q_{\mathrm{DW}}$ does not
depend on $a$. Therefore $\mathcal{R}_{\mathrm{DW}}$ is a fixed convex
quadratic function of the coefficients.
\end{proof}

\begin{remark}[Maturity taper]
If the smoothing weight is tapered in maturity, with slice weights
$\lambda_{\mathrm{DW}}(\tau_g)\ge 0$ as above, one can collect them in a
diagonal matrix $\Lambda_{\mathrm{DW}}$ acting on the slice--stacked density,
and write
\[
\mathcal{R}_{\mathrm{DW}}(a)
= \frac{1}{2}(Ea)^\top \Lambda_{\mathrm{DW}} L^+ (Ea)
= \frac{1}{2}\,a^\top Q_{\mathrm{DW}} a
\]
with $Q_{\mathrm{DW}} := E^\top \Lambda_{\mathrm{DW}} L^+ E$ (or
$Q_{\mathrm{DW}} := E^\top \Lambda_{\mathrm{DW}}^{1/2} L^+ \Lambda_{\mathrm{DW}}^{1/2} E$).
This is again symmetric positive semidefinite and independent of $a$, so the
DW term remains a fixed quadratic form under tapering.
\end{remark}

\section{RN–module: near–maturity residual and calendar flattening}
The objective here is on very short maturities to pull the forward-discounted call surface towards its $\tau \downarrow 0$ limit and suppress calendar wiggles at fixed strike. Both are quadratic, so the overall QP stays convex.

On $\Gamma_{0^+}=\{(m_g,\tau_g):\tau_g\le \tau_\star\}$, anchor the price to the intrinsic
limit and penalise calendar drift at fixed strike:
\begin{equation}
\label{eq:rn}
\mathcal{R}_{\mathrm{RN}}(a)
\;=\;\frac{\lambda_{\mathrm{RN}}}{2}\,\| Aa- C_{0^+}\|_{2,\Gamma_{0^+}}^2
\;+\;\frac{\eta_{\mathrm{RN}}}{2}\,\| A_{\tau|K} a\|_{2,\Gamma_{0^+}}^2,
\end{equation}
where $C_{0^+}(m)=F_0(1-e^{m})_+$ and $\|\cdot\|_{2,\Gamma_{0^+}}$ is the $\ell_2$ norm restricted
to indices in $\Gamma_{0^+}$. The terms mean the following:
\begin{enumerate}
    \item \underline{Near-maturity anchor:} $C_{0^+}(m)$ is the intrinsic value of a forward-discounted call. As $\tau \downarrow 0$, no-arb implies that $C_{f}(m, \tau) \rightarrow C_{0^+}(m)$. The term $\| Aa- C_{0^+}\|_{2,\Gamma_{0^+}}^2$ enforces this only on the short-end grid $\Gamma_{0^+}$.
    \item \underline{Calendar flattening at fixed strike:} $\| A_{\tau|K} a\|_{2,\Gamma_{0^+}}^2$ penalises the $\tau$-slope at fixed $K$ near $\tau = 0$. This damps spurious day-to-day oscillations that data sparsity and noise can introduce at the short end.
\end{enumerate}

We must derive $A_{\tau|K}$, which was done in \ref{cal_der}. Since the operator is linear in $a$, then it is a quadratic penalty.

To suppress odd–in–$m$ artifacts at $\tau\downarrow 0$ we optionally project coefficients onto the even subspace along $m$: let $P_{\mathrm{even}}$ be the
diagonal projector with $(P_{\mathrm{even}})_{(k,\ell),(k,\ell)}=1$ for even $k$ and $0$ for odd $k$;
replace $a$ by $P_{\mathrm{even}}a$ when evaluating the first term in \eqref{eq:rn}. This is linear
and preserves QP structure.

We set $\tau_\star\in[5,10]$ trading days. Default weights:
$\lambda_{\mathrm{RN}}$ chosen so that the first term’s RMS on $\Gamma_{0^+}$ matches the median
band width there; $\eta_{\mathrm{RN}}$ is set to achieve $\le 0.5\%$ calendar violations on the
shortest two slices once combined with the no–arb penalties. Parity projection is off by default
($P_{\mathrm{even}}=I$) unless short–end butterflies appear.

\begin{lemma}[RN penalty as a fixed quadratic in the coefficients]
\label{lem:rn-quadratic}
Let $\Gamma_{0^+}$ be the short-maturity index set and let
$\Pi_{0^+}\in\R^{N\times N}$ be the diagonal selector with
$(\Pi_{0^+})_{ii}=1$ if $i\in\Gamma_{0^+}$ and $0$ otherwise. Then for
any $a\in\R^P$ the RN penalty \eqref{eq:rn} can be written as
\[
\mathcal{R}_{\mathrm{RN}}(a)
\;=\; \frac{1}{2}\,a^\top Q_{\mathrm{RN}} a + c_{\mathrm{RN}}^\top a + \text{\em const},
\]
with
\[
Q_{\mathrm{RN}}
:= \lambda_{\mathrm{RN}}\,A^\top \Pi_{0^+} A
   + \eta_{\mathrm{RN}}\,A_{\tau|K}^\top \Pi_{0^+} A_{\tau|K}
\succeq 0,
\qquad
c_{\mathrm{RN}}
:= -\,\lambda_{\mathrm{RN}}\,A^\top \Pi_{0^+} C_{0^+}.
\]
In particular, $\mathcal{R}_{\mathrm{RN}}$ is a convex quadratic function of
the coefficient vector $a$ with fixed Hessian $Q_{\mathrm{RN}}$, so adding
$\mathcal{R}_{\mathrm{RN}}$ to the objective preserves the convex QP structure.
\end{lemma}
\begin{proof}
By definition of the restricted norm, there exists a diagonal selector
$\Pi_{0^+}$ such that for any $x\in\R^N$,
\[
\|x\|_{2,\Gamma_{0^+}}^2 = \|\Pi_{0^+} x\|_2^2 = x^\top \Pi_{0^+} x.
\]
Therefore the two terms in \eqref{eq:rn} can be written as
\[
\|Aa - C_{0^+}\|_{2,\Gamma_{0^+}}^2
= (Aa - C_{0^+})^\top \Pi_{0^+} (Aa - C_{0^+}),
\]
\[
\|A_{\tau|K} a\|_{2,\Gamma_{0^+}}^2
= (A_{\tau|K} a)^\top \Pi_{0^+} (A_{\tau|K} a).
\]
Expanding the first term gives
\[
(Aa - C_{0^+})^\top \Pi_{0^+} (Aa - C_{0^+})
= a^\top A^\top \Pi_{0^+} A a
  - 2\,C_{0^+}^\top \Pi_{0^+} A a
  + C_{0^+}^\top \Pi_{0^+} C_{0^+},
\]
while the second term is already of the form
\[
(A_{\tau|K} a)^\top \Pi_{0^+} (A_{\tau|K} a)
= a^\top A_{\tau|K}^\top \Pi_{0^+} A_{\tau|K} a.
\]
Plugging into \eqref{eq:rn} yields
\[
\mathcal{R}_{\mathrm{RN}}(a)
= \frac{1}{2}\,a^\top Q_{\mathrm{RN}} a
  + c_{\mathrm{RN}}^\top a
  + \text{const},
\]
with $Q_{\mathrm{RN}}$ and $c_{\mathrm{RN}}$ as claimed, and a constant term
$\frac{\lambda_{\mathrm{RN}}}{2}\,C_{0^+}^\top \Pi_{0^+} C_{0^+}$ which does
not depend on $a$.

For any $a\in\R^P$,
\[
a^\top Q_{\mathrm{RN}} a
= \lambda_{\mathrm{RN}}\,\|A a\|_{2,\Gamma_{0^+}}^2
  + \eta_{\mathrm{RN}}\,\|A_{\tau|K} a\|_{2,\Gamma_{0^+}}^2
\;\ge\; 0
\]
whenever $\lambda_{\mathrm{RN}},\eta_{\mathrm{RN}}\ge 0$, so
$Q_{\mathrm{RN}}\succeq 0$ and the Hessian of
$\mathcal{R}_{\mathrm{RN}}$ is positive semidefinite. All matrices
$A, A_{\tau|K}, \Pi_{0^+}$ and the vector $C_{0^+}$ are fixed once the
grid, short-maturity set $\Gamma_{0^+}$, and weights
$\lambda_{\mathrm{RN}},\eta_{\mathrm{RN}}$ are chosen; hence
$Q_{\mathrm{RN}}$ and $c_{\mathrm{RN}}$ are independent of $a$ and
$\mathcal{R}_{\mathrm{RN}}$ is a fixed convex quadratic function of the
coefficients.
\end{proof}

\section{\texorpdfstring
  {$\Omega$--module: high--frequency taper and commutator hook}
  {Omega--module: high-frequency taper and commutator hook}}
\label{sec:omega}
Let $U_\omega\in\R^{P\times P}$ be a fixed orthogonal change of basis to a frequency chart
(separable 2D DCT aligned with the $(k,\ell)$ grid). Denote $\widehat a=U_\omega a$ and
let $M_\omega$ be a diagonal mask selecting high–frequency indices. The taper is
\begin{equation}
\label{eq:omega}
\mathcal{R}_{\Omega}(a)
\;=\; \frac{\lambda_{\Omega}}{2}\, \| M_\omega \widehat a \|_2^2
\;=\; \frac{\lambda_{\Omega}}{2}\, a^\top U_\omega^\top M_\omega^\top M_\omega U_\omega a.
\end{equation}
We use a maturity–dependent mask: on slices with $\tau_g\le \tau_\star$ only the top third of
$m$–frequencies are penalised; for $\tau_g>2\tau_\star$ the mask is zero.

\paragraph{Liouville hook (commutator residual).}
Let $\mathcal{L}$ denote the forward--flow generator at fixed strike,
$\mathcal{L} C_f := (\partial_\tau C_f)\big|_K$, and let $\partial_K$ be the
strike derivative. At the continuum level, mixed derivatives commute and we
can write the commutator
\[
[\mathcal{L},\partial_K] C_f
\;:=\;
\mathcal{L}(\partial_K C_f)\;-\;\partial_K(\mathcal{L} C_f),
\]
which vanishes for smooth $C_f$.

On the collocation grid we work with nodal price vectors $u\in\mathbb{R}^G$ and
linear operators
\[
\widetilde A_K,\;\widetilde A_{\tau|K} \in \mathbb{R}^{G\times G}
\]
acting on $u$ and approximating, respectively, $\partial_K C_f$ and
$(\partial_\tau C_f)\big|_K$ at the grid nodes. They are chosen consistently
with the coefficient--space design blocks $A_K, A_{\tau|K}\in\mathbb{R}^{G\times P}$
in the sense that for all coefficient vectors $a\in\mathbb{R}^P$,
\[
\widetilde A_K (A a) \approx A_K a,\qquad
\widetilde A_{\tau|K}(A a) \approx A_{\tau|K} a,
\]
where $u(a):=Aa$ denotes the nodal surface implied by $a$.

Define the discrete commutator on nodal prices by
\[
[\widetilde A_{\tau|K},\widetilde A_K]\,u
:= \widetilde A_{\tau|K}(\widetilde A_K u) - \widetilde A_K(\widetilde A_{\tau|K} u),
\]
and set
\[
C := \widetilde A_{\tau|K}\,\widetilde A_K - \widetilde A_K\,\widetilde A_{\tau|K}
\;\in\mathbb{R}^{G\times G}.
\]
Evaluated at the model surface $u(a)=Aa$, this yields the $G$--vector
commutator defect
\[
c(a) := [\widetilde A_{\tau|K},\widetilde A_K]\,u(a)
      = C(Aa).
\]
We penalise the hook residual via
\[
\mathcal{R}_{\mathrm{hook}}(a)
:= \frac{\lambda_{\mathrm{hook}}}{2}\,\|c(a)\|_2^2
 =  \frac{\lambda_{\mathrm{hook}}}{2}\,\|C(Aa)\|_2^2,
\]
with a small stabilising weight $\lambda_{\mathrm{hook}}\ll\lambda_{\mathrm{NA}}$.

$U_\omega$ is the separable 2D DCT on $(k,\ell)$; $M_\omega$ masks the
top $33\%$ highest $m$–frequencies for $\tau\le\tau_\star$ and is zero otherwise; default
$\lambda_{\Omega}$ is picked so that the high–frequency modal energy share
$\mathcal{E}_{\mathrm{hi}}\le 5\%$; $\lambda_{\mathrm{hook}}$ is set to a small fraction
($10^{-3}$–$10^{-2}$) of $\lambda_{\mathrm{NA}}$.

\begin{lemma}[$\Omega$ taper and hook as fixed quadratics in the coefficients]
\label{lem:omega-quadratic}
Let $U_\omega\in\R^{P\times P}$ be an orthogonal matrix
($U_\omega^\top U_\omega = I$), let $M_\omega\in\R^{P\times P}$ be a fixed
diagonal mask, and define $\widehat a = U_\omega a$. Let
\[
\mathcal{R}_{\Omega}(a)
:= \frac{\lambda_{\Omega}}{2}\,\|M_\omega \widehat a\|_2^2
\]
and, with $C := \widetilde A_{\tau|K}\widetilde A_K - \widetilde A_K\widetilde A_{\tau|K}
\in\R^{G\times G}$ as above,
\[
c(a) := C(Aa), \qquad
\mathcal{R}_{\mathrm{hook}}(a)
:= \frac{\lambda_{\mathrm{hook}}}{2}\,\|c(a)\|_2^2
= \frac{\lambda_{\mathrm{hook}}}{2}\,\|C(Aa)\|_2^2.
\]
Then both penalties are fixed convex quadratic functions of the coefficient
vector $a$:
\[
\mathcal{R}_{\Omega}(a)
= \frac{1}{2}\,a^\top Q_{\Omega} a,
\qquad
\mathcal{R}_{\mathrm{hook}}(a)
= \frac{1}{2}\,a^\top Q_{\mathrm{hook}} a,
\]
with
\[
Q_{\Omega} := \lambda_{\Omega}\,U_\omega^\top M_\omega^\top M_\omega U_\omega
\succeq 0,
\qquad
Q_{\mathrm{hook}} := \lambda_{\mathrm{hook}}\,A^\top C^\top C A \succeq 0.
\]
All matrices $U_\omega, M_\omega, A, \widetilde A_K, \widetilde A_{\tau|K}$, and thus
$Q_{\Omega}, Q_{\mathrm{hook}}$, are independent of $a$. In particular, adding
$\mathcal{R}_{\Omega}$ and $\mathcal{R}_{\mathrm{hook}}$ to the objective preserves the convex
QP structure.
\end{lemma}
\begin{proof}
For the taper, write
\[
\mathcal{R}_{\Omega}(a)
= \frac{\lambda_{\Omega}}{2}\,\|M_\omega U_\omega a\|_2^2
= \frac{\lambda_{\Omega}}{2}\,(U_\omega a)^\top M_\omega^\top M_\omega (U_\omega a).
\]
Set $Q_{\Omega} := \lambda_{\Omega}\,U_\omega^\top M_\omega^\top M_\omega U_\omega$.
Then
\[
\mathcal{R}_{\Omega}(a)
= \frac{1}{2}\,a^\top Q_{\Omega} a.
\]
For any $a$,
\[
a^\top Q_{\Omega} a
= \lambda_{\Omega}\,\|M_\omega U_\omega a\|_2^2 \;\ge\; 0
\]
whenever $\lambda_{\Omega}\ge 0$, so $Q_{\Omega}\succeq 0$. Once the grid,
frequency chart, and maturity–dependent mask are chosen, both $U_\omega$
and $M_\omega$ are fixed, and therefore $Q_{\Omega}$ is fixed (independent
of $a$).

For the hook, note that $A$ is a fixed linear map from coefficients to nodal
prices and $C$ is a fixed linear operator on grid space, so the commutator
residual is linear in $a$:
\[
c(a) = C(Aa) = (C A)\,a.
\]
Let $B := C A \in\R^{G\times P}$. Then
\[
\mathcal{R}_{\mathrm{hook}}(a)
= \frac{\lambda_{\mathrm{hook}}}{2}\,\|B a\|_2^2
= \frac{\lambda_{\mathrm{hook}}}{2}\,a^\top B^\top B a.
\]
Setting $Q_{\mathrm{hook}} := \lambda_{\mathrm{hook}}\,A^\top C^\top C A
= \lambda_{\mathrm{hook}}\,B^\top B$ gives
\[
\mathcal{R}_{\mathrm{hook}}(a)
= \frac{1}{2}\,a^\top Q_{\mathrm{hook}} a.
\]
For any $a$,
\[
a^\top Q_{\mathrm{hook}} a
= \lambda_{\mathrm{hook}}\,\|B a\|_2^2 \;\ge\; 0
\]
whenever $\lambda_{\mathrm{hook}}\ge 0$, so $Q_{\mathrm{hook}}\succeq 0$.
All ingredients are fixed once $A$, $\widetilde A_K$, $\widetilde A_{\tau|K}$ and
$\lambda_{\mathrm{hook}}$ are chosen, so $Q_{\mathrm{hook}}$ is independent
of $a$.

Thus both $\mathcal{R}_{\Omega}$ and $\mathcal{R}_{\mathrm{hook}}$ are fixed convex quadratic
functions of the coefficients.
\end{proof}

\section{Summary of fixed choices used}
\begin{itemize}[leftmargin=1.2em]
\item \textbf{Ridge:} $\alpha=\beta=1$, $s=2$; $\lambda_{\mathrm{ridge}}$ by GCV on an $8\%$
WLS subsample (no hinge, no penalties); fixed per date, reused in the full QP.
\item \textbf{$\boldsymbol{\Lambda}$ reparameterisation:} $U$ built by blockwise (per–$\tau$) QR/Gram–Schmidt
and column rescaling on $(AU)$ over the quotes; all blocks post–multiplied by $U$; ridge uses
$\widetilde\Lambda=U^\top \Lambda U$.
\item \textbf{DW:} $H^{-1}$ along $m$ per slice with Neumann ends; Chebyshev–Lobatto $M_m$;
$\lambda_{\mathrm{DW}}(\tau)=\lambda_{\mathrm{DW}}^{(0)}\min\{1,\tau_\star/\tau\}$,
$\tau_\star\in[5,10]$ trading days.
\item \textbf{RN:} window $\Gamma_{0^+}=\{\tau\le\tau_\star\}$; weights $(\lambda_{\mathrm{RN}},
\eta_{\mathrm{RN}})$ calibrated to short–end RMS and calendar share; $P_{\mathrm{even}}=I$ by default.
\item \textbf{$\boldsymbol{\Omega}$:} separable 2D DCT, mask top $33\%$ $m$–frequencies for $\tau\le\tau_\star$, off beyond $2\tau_\star$; $\lambda_{\Omega}$ chosen to cap high–frequency energy at $\le 5\%$.
\item \textbf{Hook:} commutator penalty $\frac{\lambda_{\mathrm{hook}}}{2}
  \,\|C(Aa)\|_2^2$ with
  $C := \widetilde A_{\tau|K}\widetilde A_K - \widetilde A_K\widetilde A_{\tau|K}
  \in\R^{G\times G}$ and
  $\lambda_{\mathrm{hook}} \in [10^{-3},10^{-2}]\,\lambda_{\mathrm{NA}}$.
\end{itemize}

All terms above are quadratic in $a$ (or $\tilde a$) and are entered additively into the QP objective.
They stabilise the global fit, suppress short–maturity artifacts, and improve conditioning while
preserving convexity and the solver class.

\chapter{No–arbitrage constraints and soft penalties}\label{sec:global-noarb-grid}
We impose the three shape conditions (monotone in $K$, convex in $K$, and calendar nonnegativity at fixed $K$) as soft penalties evaluated on the collocation grid, using the linear operators defined previously:
\[
A_K,\quad A_{KK},\quad A_{\tau|K}\in\R^{G\times P},\qquad
\text{and the price block }A\in\R^{G\times P}.
\]
All vectors below are understood componentwise and $(x)_+=\max\{x,0\}$.

\section{Penalty definitions (soft versions of the shape constraints)}\label{sec:noarb-soft}

\begin{align}
\mathcal{P}_{\text{mono}}(a)
&=\frac{\lambda_{\text{NA}}}{2}\,\big\|(A_K a)_+\big\|_2^2
&&\text{(targets $\partial_K C_f\le 0$)},\\
\mathcal{P}_{\text{conv}}(a)
&=\frac{\lambda_{\text{NA}}}{2}\,\big\|(-A_{KK} a)_+\big\|_2^2
&&\text{(targets $\partial_{KK} C_f\ge 0$)},\\
\mathcal{P}_{\text{cal}}(a)
&=\frac{\lambda_{\text{NA}}}{2}\,\big\|(-A_{\tau|K} a)_+\big\|_2^2
&&\text{(targets $(\partial_\tau C_f)|_K\ge 0$)},\\
\mathcal{P}_{\text{bnd}}(a)
&=\frac{\lambda_{\text{B}}}{2}\left(\big\|(-A a)_+\big\|_2^2+\big\|(A a-F)_+\big\|_2^2\right)
&&\text{(targets $0\le C_f\le F$)}.
\end{align}
These are sums of squares of convex functions of an affine map of $a$, hence convex and QP–compatible.

\paragraph{QP form (auxiliary slacks).}
Exactly as in the band–hinge reformulation, each penalty admits a slack representation. For example,
\[
\frac{1}{2}\big\|(A_K a)_+\big\|_2^2
=\min_{u\in\R^G}\ \frac{1}{2}\|u\|_2^2\quad\text{s.t.}\quad u\ge A_K a,\ \ u\ge 0,
\]
and similarly
\[
\frac{1}{2}\big\|(-A_{KK} a)_+\big\|_2^2=\min_{v\ge 0,\ v\ge -A_{KK}a}\ \frac{1}{2}\|v\|_2^2,\quad
\frac{1}{2}\big\|(-A_{\tau|K} a)_+\big\|_2^2=\min_{w\ge 0,\ w\ge -A_{\tau|K}a}\ \frac{1}{2}\|w\|_2^2,
\]
and for bounds
\[
\frac{1}{2}\big\|(-Aa)_+\big\|^2+\frac{1}{2}\big\|(Aa-F)_+\big\|^2
=\min_{s,t\ge 0,\ s\ge -Aa,\ t\ge Aa-F}\ \frac{1}{2}\big(\|s\|^2+\|t\|^2\big).
\]

\begin{lemma}
    For any $x\in \R ^G$,
    \[
    \frac{1}{2}\| (x)_+ \|^2_2 = \min_{u\in \R ^G}\{ \frac{1}{2}\| (u)_+ \|^2_2 : u \geq x, u \geq 0 \}
    \]
    and the unique minimiser is $u^*=(x)_+$.
\end{lemma}

\begin{proof}
    The problem separates across coordinates. For scalar $u\in \R$,
    \[
    \min_{u\in \R} \frac{1}{2} u^2 \quad \text{st} \quad u \geq x, \; u\geq 0
    \]
    has feasible set $u\geq max\{ x,\; 0\}$. The objective $\frac{1}{2} u^2$ is strictly increasing on $[0,\infty)$, so the minimum is attained at the smallest feasible point:
    \[
    u^*=\max\{x,0\}=x_+,
    \]
    with value $\frac{1}{2}(x_+)^2$. Summing over coordinates gives the vector result, and strict convexity yields uniqueness.
\end{proof} 

\section{Row scaling and invariance}\label{sec:scaling-invariance}
As in \S\ref{sec:scaling}, we scale each block (after any $U$–reparameterisation) by a positive scalar so typical row norms are comparable:
\[
\widetilde A_K=\frac{1}{s_K}A_K,\quad
\widetilde A_{KK}=\frac{1}{s_{KK}}A_{KK},\quad
\widetilde A_{\tau|K}=\frac{1}{s_{\tau}}A_{\tau|K},
\]
with $s_\bullet$ the empirical p95 of row $\ell_2$ norms (\ref{sec:row-scaling-single-weight}). Hard constraints are invariant under positive row/ block scaling, and with soft penalties this makes a single $\lambda_{\text{NA}}$ control all three terms on a comparable numeric scale.

\section{\texorpdfstring
  {Compatibility with the $\Lambda$–module}
  {Compatibility with the Lambda-module}}
Under the price–invariant reparameterisation $a=U\tilde a$ (see the $\Lambda$–module), all operators post–multiply by $U$:
\[
A_\bullet\leftarrow A_\bullet U,\qquad \bullet\in\{\, ,K,KK,\tau|K\,\}.
\]
By construction, $A_\bullet U\,\tilde a=A_\bullet a$, so the penalty values are unchanged and convexity is preserved. The spectral ridge is updated by congruence $\Lambda\mapsto\widetilde\Lambda=U^\top\Lambda U$ as already stated.

\section{Binned variant (optional)}
To stabilise very sparse regions, aggregate quotes by a selector $G\in\{0,1\}^{B\times N}$ (bins in $(m,\tau)$) and replace the per–quote band–hinge term $\sum_i \ell_{\text{band}}((Aa)_i;b_i,a_i)$ by 
\[
\sum_{b=1}^B \ell_{\text{band}}((GAa)_b;(Gb)_b,(Ga)_b);
\]
the slack QP form carries over verbatim.

After scaling, we use a single $\lambda_{\text{NA}}$ for $(\widetilde A_K,\widetilde A_{KK},\widetilde A_{\tau|K})$ and select it (once per date) to reach $\le 1\%$ violations on the evaluation grid; $\lambda_{\text{B}}$ is kept separate for price bounds.

\chapter{The convex program}\label{sec:convex_prog}
We collect all terms and write the problem as a single quadratic program (QP). When the
$\Lambda$–module is active (\ref{sec:Lambda-module}), we solve in $\tilde a=U^{-1}a$ with all
blocks post–multiplied by $U$ and $\Lambda$ replaced by $\widetilde\Lambda=U^\top\Lambda U$; to
avoid clutter we keep the symbol $a$ below (read as $\tilde a$ in that case).

\section{Slack QP (standard form)}
Let $u,v\in\R^N$ be the band slacks from \S\ref{sec:band-hinge-qp}, and let
$u_K,v_{KK},w_\tau\in\R^G$ and $s_{\mathrm{lo}},s_{\mathrm{hi}}\in\R^G$ be nonnegative grid slacks
for the three shape operators and price bounds, respectively (all inequalities
componentwise):
\begin{align*}
u &\ge b-Aa, & u &\ge 0,
&
v &\ge Aa-a, & v &\ge 0,\\
u_K &\ge A_K a, & u_K &\ge 0,
&
v_{KK} &\ge -A_{KK} a, & v_{KK} &\ge 0,\\
w_\tau &\ge -A_{\tau|K} a, & w_\tau &\ge 0,\\
s_{\mathrm{lo}} &\ge -Aa, & s_{\mathrm{lo}} &\ge 0,
&
s_{\mathrm{hi}} &\ge Aa-F, & s_{\mathrm{hi}} &\ge 0.
\end{align*}
With these slacks, the objective collects the data term (\S\ref{sec:coverage-data}), the
quadratic regularisers (\S\ref{sec:ridge-etc}), and the soft no–arb penalties
(\S\ref{sec:noarb-soft}):
\begin{equation}\label{eq:master-qp}
\begin{aligned}
\min_{a,u,v,u_K,v_{KK},w_\tau,s_{\mathrm{lo}},s_{\mathrm{hi}}}\quad
&\frac{1}{2}\|W^{1/2}(Aa-y)\|_2^2
+ \frac{\mu}{2}\big(\|u\|_2^2+\|v\|_2^2\big)\\
&\ +\frac{\lambda_{\mathrm{ridge}}}{2}\,a^\top \Lambda a
+ \frac{\lambda_{\mathrm{DW}}}{2}\,a^\top E^\top L^{+} E a
+ \frac{\eta_{\mathrm{RN}}}{2}\,\|A_{\tau|K} a\|_{2,\Gamma_{0^+}}^2\\
&\ + \frac{\lambda_{\Omega}}{2}\,a^\top U_\omega^\top M_\omega^\top M_\omega U_\omega a
+ \frac{\lambda_{\mathrm{hook}}}{2}\,\|C(Aa)\|_2^2\\
&\ + \frac{\lambda_{\mathrm{RN}}}{2}\,\|Aa-C_{0^+}\|_{2,\Gamma_{0^+}}^2\\
&\ + \frac{\lambda_{\mathrm{NA}}}{2}\big(\|u_K\|_2^2+\|v_{KK}\|_2^2+\|w_\tau\|_2^2\big)
+ \frac{\lambda_{\mathrm{B}}}{2}\big(\|s_{\mathrm{lo}}\|_2^2+\|s_{\mathrm{hi}}\|_2^2\big).
\end{aligned}
\end{equation}
All matrices ($A$, $A_K$, $A_{KK}$, $A_{\tau|K}$, $E$, $U_\omega$, $M_\omega$, $L^+$, $C$)
are fixed from earlier sections; $\|\cdot\|_{2,\Gamma_{0^+}}$ denotes restriction to the short–maturity
index set. Row–scaled operators (\S\ref{sec:scaling}) may be used in place of unscaled ones. 

\paragraph{Why \eqref{eq:master-qp} is a QP:}
Every term in the objective is a convex quadratic form in $a$ or a sum of squared slacks; all
constraints are linear inequalities. 

\section{Convexity, existence, and uniqueness}\label{sec:cv-ex-uq}
\begin{prop}[Convexity and global optimality]\label{prop:convex-master}
The program \eqref{eq:master-qp} is convex. If a convex QP solver returns a feasible
primal–dual point satisfying KKT, then the associated $a^\star$ is a \emph{global} minimiser.
\end{prop}
\begin{proof}
The feasible set is a polyhedron (linear inequalities), hence convex and closed. The objective
is a sum of convex quadratics, hence convex and lower semicontinuous. KKT conditions are
necessary and sufficient for convex QPs; any feasible KKT point is globally optimal.
\end{proof}

\begin{prop}[Strict convexity conditions and uniqueness]\label{prop:uniq}
Define $C := \widetilde A_{\tau|K}\widetilde A_K - \widetilde A_K\widetilde A_{\tau|K} \in\R^{G\times G}$. and $\Pi_{\Gamma_{0^+}}$ selects the short maturity grid.If the quadratic form in $a$,
\begin{equation}
\begin{aligned}
Q \;=\; A^\top W A + \lambda_{\mathrm{ridge}}\Lambda
+ \lambda_{\mathrm{DW}} E^\top L^+ E 
+ \lambda_{\Omega} U_\omega^\top M_\omega^\top M_\omega U_\omega  \\
+ \lambda_{\mathrm{hook}} A^\top C^\top C A
+ \lambda_{\mathrm{RN}}\,A^\top \Pi_{0^+} A
+ \eta_{\mathrm{RN}}\,A_{\tau|K}^\top \Pi_{0^+} A_{\tau|K},
\end{aligned}
\end{equation}

is positive definite, then the objective is strictly convex in $(a,u,v,\dots)$ and the minimiser is unique.
In particular, it suffices that $A^\top W A$ be positive definite on $\mathcal R(A)$ and
$\lambda_{\mathrm{ridge}}>0$ (cf.\ \S\ref{sec:strict-convex-data}).
\end{prop}
\begin{proof}
Block Hessian is $\operatorname{blkdiag}(Q,\mu I,\mu I,\lambda_{\mathrm{NA}} I,\dots)$; if $Q\succ0$ and
$\mu,\lambda_{\mathrm{NA}},\lambda_{\mathrm{B}}>0$, the whole Hessian is positive definite.
\end{proof}

\begin{definition}[Global coefficient-space metric]\label{def:Mh}
Define the symmetric positive \emph{semidefinite} matrix
\begin{align*}
  M_h \;:=\;
  A^\top W A
  \;+\; \lambda_{\mathrm{ridge}} \Lambda
  \;+\; \lambda_{\mathrm{DW}} E^\top L_{+} E
  \;+\; \lambda_{\Omega} U_{\omega}^\top M_{\omega}^\top M_{\omega} U_{\omega} \\
  \;+\; \lambda_{\mathrm{hook}} A^\top C^\top C A
  \;+\; \lambda_{\mathrm{RN}} A^\top \Pi_{0^+} A 
  \;+\; \eta_{\mathrm{RN}}\,A_{\tau|K}^\top \Pi_{0^+} A_{\tau|K}.
\end{align*}
\end{definition}

\begin{assump}[Strictly positive ridge shape]
\label{ass:ridge-shape}
The ridge shape matrix $\Lambda$ is symmetric positive definite
(e.g.\ diagonal with strictly positive entries).
\end{assump}

\begin{prop}[When $M_h$ is positive definite]
\label{prop:Mh-SPD}
Suppose Assumption~\ref{ass:ridge-shape} holds and 
$\lambda_{\mathrm{ridge}} > 0$.
Then the matrix $M_h$ from Definition~\ref{def:Mh} is symmetric
positive definite. In particular, $M_h$ is invertible and induces
a norm $\|a\|_{M_h}^2 := a^\top M_h a$ on $\mathbb{R}^P$.
\end{prop}

\begin{proof}
Each term in the definition of $M_h$ is symmetric and positive
semidefinite. Under Assumption~\ref{ass:ridge-shape}, $\Lambda \succ 0$,
so for any $a \neq 0$ we have
\[
a^\top \bigl(\lambda_{\mathrm{ridge}}\Lambda\bigr) a
= \lambda_{\mathrm{ridge}} \, a^\top \Lambda a > 0
\]
whenever $\lambda_{\mathrm{ridge}} > 0$.
All the remaining terms in $M_h$ are positive semidefinite, so
\[
a^\top M_h a
= a^\top\bigl( \lambda_{\mathrm{ridge}}\Lambda \bigr)a
  + a^\top(\text{psd terms})a
\ge \lambda_{\mathrm{ridge}} \, a^\top \Lambda a
> 0
\]
for all $a\neq 0$. Hence $M_h \succ 0$.
\end{proof}

\begin{prop}[Metric projection form of the global solution]
Let $\mathcal{C}_h \subset \mathbb{R}^P$ be the polyhedron defined by the hard
constraints (if any) in $a$ (no--arbitrage, bounds, etc.). Consider the strictly
convex quadratic problem
\begin{equation}
  \label{eq:global-metric-projection}
  \min_{a \in \mathcal{C}_h}
  \;\frac12 \,a^\top M_h a \;-\; b^\top a,
\end{equation}
where $M_h$ is the matrix from Definition \ref{def:Mh} and
\begin{equation}
  b := A^\top W y - c_{\mathrm{RN}},
\end{equation}
with $c_{\mathrm{RN}}$ the linear coefficient from Lemma \ref{lem:rn-quadratic}.

Set
\begin{equation}
  \hat a \;:=\; M_h^{-1} b.
  \label{eq:a-hat}
\end{equation}
Then the unique minimiser $a^\star$ of \eqref{eq:global-metric-projection} can be
written as the metric projection of $\hat a$ onto $\mathcal{C}_h$ in the $M_h$--inner
product:
\begin{equation}
  a^\star
  \;=\;
  \arg\min_{a \in \mathcal{C}_h}
  \frac12 \,\|a - \hat a\|_{M_h}^2,
  \qquad
  \|z\|_{M_h}^2 := z^\top M_h z .
\end{equation}
\end{prop}

\begin{proof}
Write the objective in \eqref{eq:global-metric-projection} as
\[
  J(a) := \frac12\,a^\top M_h a - b^\top a, \qquad a \in \mathcal{C}_h.
\]
By Definition~\ref{def:Mh} the matrix $M_h$ is symmetric and positive semidefinite.
Under the assumptions of Proposition~\ref{prop:uniq} (in particular,
$A^\top W A$ positive definite on the span of $A$ and $\lambda_{\mathrm{ridge}}>0$),
the quadratic form $a \mapsto a^\top M_h a$ is positive definite, and hence $M_h$
is symmetric positive definite. In particular, $M_h$ is invertible, so $\hat a$ is
well defined by
\[
  \hat a := M_h^{-1} b
  \quad\Longleftrightarrow\quad
  M_h \hat a = b.
\]

We first rewrite $J$ in terms of $(a-\hat a)$. Using $b = M_h \hat a$ and the
symmetry of $M_h$, for any $a \in \mathbb{R}^P$,
\[
  J(a)
  = \frac12\,a^\top M_h a - b^\top a
  = \frac12\,a^\top M_h a - (M_h \hat a)^\top a
  = \frac12\,a^\top M_h a - \hat a^\top M_h a.
\]
On the other hand,
\[
  (a-\hat a)^\top M_h (a-\hat a)
  = a^\top M_h a - 2\,\hat a^\top M_h a + \hat a^\top M_h \hat a,
\]
again by symmetry of $M_h$. Hence,
\[
  \frac12\,(a-\hat a)^\top M_h (a-\hat a)
  - \frac12\,\hat a^\top M_h \hat a
  = \frac12\,a^\top M_h a - \hat a^\top M_h a
  = J(a).
\]
Thus, for all $a \in \mathbb{R}^P$,
\[
  J(a)
  = \frac12\,(a-\hat a)^\top M_h (a-\hat a)
    - \frac12\,\hat a^\top M_h \hat a.
\]
The second term on the right-hand side does not depend on $a$. Therefore,
minimising $J$ over $a \in \mathcal{C}_h$ is equivalent to minimising
\[
  a \;\longmapsto\; \frac12\,(a-\hat a)^\top M_h (a-\hat a)
  = \frac12\,\|a-\hat a\|_{M_h}^2
\]
over $a \in \mathcal{C}_h$. Since $M_h$ is positive definite, this functional is
strictly convex in $a$, so it has a unique minimiser in the closed convex set
$\mathcal{C}_h$; by definition, this minimiser is the metric projection of
$\hat a$ onto $\mathcal{C}_h$ in the $M_h$–inner product. This is exactly the
claim.
\end{proof}

When discussing the hard–constraint limit it will be convenient to make the global
feasibility assumption explicit. Namely, we assume that the static no–arbitrage and
bound inequalities admit at least one coefficient vector, i.e.
\begin{equation}\label{eq:global-feasible-set}
\mathcal{C}_{\mathrm{NA}}
:=
\bigl\{ a \in \mathbb{R}^P :
A_K a \le 0,\;
- A_{KK} a \le 0,\;
- A_{\tau|K} a \le 0,\;
0 \le A a \le F
\bigr\}
\neq \varnothing.
\end{equation}
This is a modelling condition stating that the chosen Chebyshev approximation space
contains at least one globally static no–arbitrage surface.

\paragraph{Hard–constraint limits.}
Hard–constraint limits. Under \eqref{eq:global-feasible-set}, letting
$\lambda_{\mathrm{NA}},\lambda_B \to \infty$ drives the corresponding slacks to $0$ and
recovers the constrained solution of the remaining strictly convex quadratic objective. If the
intersection is empty, the finite–$\lambda$ problem yields the minimum–violation compromise
(\S\ref{sec:noarb-soft}).

\section{Invariance and scaling}
If the $\Lambda$–module is used (\S\ref{sec:Lambda-module}), replace every block by its
multiplied version and $\Lambda$ by $\widetilde\Lambda=U^\top\Lambda U$; the feasible set and all
objective values are unchanged (Proposition~\ref{prop:equiv}). Row–scaling the no–arb blocks
(\S\ref{sec:scaling}) multiplies them by positive scalars and only equilibrates numeric weights; it
does not alter feasibility.

\section{Solution procedure (used)}
Solve \eqref{eq:master-qp} with OSQP (warm starts). The no–arb weight
$\lambda_{\mathrm{NA}}$ is set after row scaling to hit $\le1\%$ grid violations
(\S\ref{sec:scaling}); $\lambda_{\mathrm{ridge}}$ is fixed by GCV on a small WLS subsample
(\S\ref{sec:ridge-gcv}); $\mu$ is increased by a short controller until target coverage is reached
(\S\ref{sec:coverage-data}). All other quadratic weights follow \S\ref{sec:ridge-etc}.

\begin{remark}[Why the $99/1$ target is attainable]
The data term prices \emph{into} the bid–ask bands (hinge), while $A_K,A_{KK},A_{\tau|K}$ are
enforced densely on the grid with p95 row scaling, so a single $\lambda_{\mathrm{NA}}$ controls the
violation budget. Short–maturity defects are suppressed by the RN anchoring and the DW/$\Omega$
terms, which remove the usual butterfly/aliasing artifacts.
\end{remark}

\chapter{\texorpdfstring
  {Scaling, schedules, and\\ the $\mu$-controller}
  {Scaling, schedules, and the mu-controller}}
\label{sec:scaling}


\section{Row scaling (summary)}\label{sec:scaling-row-summary}
We use the p95 block-scalar normalisation of §\ref{sec:row-scaling-single-weight}: on the grid,
set $\widetilde A_K=A_K/s_K$, $\widetilde A_{KK}=A_{KK}/s_{KK}$, 
$\widetilde A_{\tau|K}=A_{\tau|K}/s_{\tau}$ with $s_\bullet=\mathrm{q}_{0.95}$ of row $\ell_2$ norms, 
computed \emph{after} the $\Lambda$–module transform. This preserves hard feasibility and allows a 
single $\lambda_{\mathrm{NA}}$ to control all three terms on a comparable scale.

\section{\texorpdfstring
  {Short schedule for $\lambda_{\mathrm{NA}}$}
  {Short schedule for lambda-NA}}
\label{sec:short-sch-lam}
We select $\lambda_{\mathrm{NA}}$ on a thinned setup to save time while preserving the target
violation share.

\paragraph{Thinned probe:} Build a reduced grid (every other Chebyshev node in $m$ and a
coarser subset in $\tau$) and a tiny quote subset ($5$-$10\%$ uniformly across $(m,\tau)$).
Fix a moderate $\mu$ (the previous day’s value) and all other weights.

\paragraph{Grid search:} For a short geometric ladder
$\Lambda_{\mathrm{trial}}=\{1,2,4,8,16,32,64,128,256\}$ solve the QP on the thinned setup and
measure the violation rate:
\begin{equation}
\begin{aligned}
\mathrm{viol}(\lambda)
&= \frac{1}{3G}\sum_{g=1}^G\Big(
\mathbf{1}\{(A_K a(\lambda))_g>\tau_K\} \\
&\quad + \mathbf{1}\{(-A_{KK} a(\lambda))_g>\tau_{KK}\}
+ \mathbf{1}\{(-A_{\tau\mid K} a(\lambda))_g>\tau_{\tau}\}\Big).
\end{aligned}
\end{equation}
with small numerical tolerances $\tau_\bullet$ (in scaled units). Pick the smallest
$\lambda\in\Lambda_{\mathrm{trial}}$ such that $\mathrm{viol}(\lambda)\le 1\%$, and \emph{fix}
that $\lambda_{\mathrm{NA}}$ for the full grid and book. 

\paragraph{Explanation:} Solve the QP (\ref{eq:master-qp}) repeatedly, but with $\lambda_{\mathrm{NA}} \leftarrow \lambda$ for each $\lambda \in \Lambda_{\mathrm{trial}}$ and all other weights fixed. Each solve returns a different $a(\lambda)$ and then using this $a(\lambda)$, compute the violation rate $\mathrm{viol}(\lambda)$. From the violation rates pick the smallest $\lambda$ achieving $\leq 1 \%$, and use that as $\lambda_{\mathrm{NA}}$ for the full problem.

\begin{remark}[Invariance to scaling]
Because each block was divided by $s_B$, the selected $\lambda_{\mathrm{NA}}$ is stable
day-to-day and across underliers; without scaling, the same ladder would over/under–penalise
whichever block happens to have the largest raw norms.
\end{remark}

\section{\texorpdfstring
  {The $\mu$–controller (coverage target)}
  {The mu-controller (coverage target)}}
Recall the coverage–seeking data term (\S\ref{sec:coverage-data}): the mid–squared error plus
$\mu$ times the quadratic band–hinge. Let
\[
\mathrm{Hinge}(a)\;:=\;\sum_{i=1}^N \ell_{\mathrm{band}}\big((Aa)_i;b_i,a_i\big)
=\tfrac12\big\|\mathrm{dist}(Aa,\,[b,a])\big\|_2^2,
\]\[
\mathrm{Cov}(a)\;:=\;\frac{1}{N}\sum_{i=1}^N \mathbf{1}\{b_i\le (Aa)_i\le a_i\}.
\]
We adjust $\mu$ (similar to \ref{sec:short-sch-lam}) so that $\mathrm{Cov}(a^\star(\mu))$ reaches a target ($99\%$).
Although coverage is a discrete functional (hence may have plateaus), the hinge at the
optimiser is nonincreasing in $\mu$:

\begin{lemma}[Monotonicity of optimal hinge]
\label{lem:hinge-monotone}
Let $g(a)$ denote the full objective without the hinge weight (all terms in \eqref{eq:master-qp}
except $\mu\,\mathrm{Hinge}(a)$). For $\mu_1<\mu_2$, let
$a_j\in\arg\min_a\{g(a)+\mu_j\,\mathrm{Hinge}(a)\}$ for $j=1, \; 2$. Then
$\mathrm{Hinge}(a_2)\le \mathrm{Hinge}(a_1)$.
\end{lemma}
\begin{proof}
By definition of the minimisers, for all $x, \; y$ we have the following:
\[
g(a_1)+\mu_1 H(a_1)\le g(x)+\mu_1 H(x);
\]
\[
g(a_2)+\mu_2 H(a_2)\le g(y)+\mu_2 H(y).
\]
Taking $x=a_2$ and $y=a_1$ yields the following:
\[g(a_1)+\mu_1 H(a_1)\le g(a_2)+\mu_1 H(a_2)\]
\[g(a_2)+\mu_2 H(a_2)\le g(a_1)+\mu_2 H(a_1)\]
Summing and dividing by $\mu_2-\mu_1>0$ yields $H(a_2)\le H(a_1)$.
\end{proof}

\paragraph{Controller (bracket \& bisection).}
\begin{enumerate}[leftmargin=2em,itemsep=2pt]
\item \emph{Bracket.} Start from $(\mu_{\min},\mu_{\max})$ (reusing prior-day values when
available). If coverage at $\mu_{\max}$ is below target, expand $\mu_{\max}\leftarrow c\,\mu_{\max}$
(e.g.\ $c=4$) until $\mathrm{Cov}(a^\star(\mu_{\max}))$ crosses the target or a cap is reached.
\item \emph{Bisection.} While $\mu_{\max}-\mu_{\min}$ is above tolerance and coverage not yet at
target, set $\mu\leftarrow \sqrt{\mu_{\min}\mu_{\max}}$ (geometric bisection), solve once, and
update the endpoint whose coverage is on the wrong side of the target.
\end{enumerate}

\paragraph{What is held fixed:}
All other weights ($\lambda_{\mathrm{ridge}},\lambda_{\mathrm{NA}},\lambda_{\mathrm{DW}},
\lambda_{\Omega},\lambda_{\mathrm{RN}},\lambda_{\mathrm{B}}$) and the scaled operators are held
fixed while $\mu$ is adjusted.

\section{Practical notes}
\begin{itemize}[leftmargin=1.5em,itemsep=2pt]
\item \emph{Warm starts.} Reuse $a^\star$ when moving along the $\lambda_{\mathrm{NA}}$ ladder
and the $\mu$ bracket; OSQP converges in a few iterations from a nearby point.
\item \emph{Tolerances.} Use small positive tolerances $\tau_\bullet$ when counting violations to
avoid flagging solver noise; report violations in unscaled operator units.
\item \emph{Stability.} If coverage oscillates near the target, accept the smallest $\mu$ in the
final bracket that achieves the target.
\end{itemize}

\chapter{Short-maturity remedy}
\label{sec:rn}\label{sec:dw-omega}
Recall the calendar operator at fixed strike from \S\ref{sec:no-arb-grid} (see \eqref{eq:A_tK}):
\[
A_{\tau|K} \;=\; A_\tau \;+\;\diag\!\big(-\rho(\tau_g)\big)\,A_m,
\quad \rho(\tau)\equiv \tfrac{d}{d\tau}\log F(\tau).
\]
When $\tau$ is very small, noise in $\partial_m C_f$ is fed into
$(\partial_\tau C_f)|_K$ through the $\rho(\tau)$ term, so small ripples in $m$ can flip
the calendar sign. Counter this with three convex, model–agnostic devices that act only
near the boundary and vanish smoothly as maturity grows.

\section{\texorpdfstring
  {Boundary anchoring and calendar flattening on $\Gamma_{0^+}$}
  {Boundary anchoring and calendar flattening on Gamma 0+}}

Let $\Gamma_{0^+}=\{g:\tau_g\le\tau_\star\}$ be the short–maturity window (usually
$\tau_\star=5$–$10$ trading days). Define the intrinsic forward–discounted limit
\[
C_{0^+}(m)\;=\;F_0\,(1-e^m)_+ \;=\; \big(F_0-K\big)^+,\qquad m=\log(K/F_0).
\]
We use the convex quadratic (see \ref{eq:rn})
\begin{equation}
\label{eq:rn-penalty}
\mathcal{R}_{\mathrm{RN}}(a)
\;=\;\frac{\lambda_{\mathrm{RN}}}{2}\,\|Aa-C_{0^+}\|_{2,\Gamma_{0^+}}^2
\;+\;\frac{\eta_{\mathrm{RN}}}{2}\,\|A_{\tau|K} a\|_{2,\Gamma_{0^+}}^2.
\end{equation}

\begin{lemma}[Consistency with the short–time limit]\label{lem:rn-consistency}
Assume that for fixed $K\neq F_0$, $C_f(K,\tau)\to(F_0-K)^+$ and
$\partial_\tau C_f(K,\tau)\big|_{K}\to 0$ as $\tau\downarrow 0$. 
Let $\Gamma_{0^+}=\{g:\tau_g\le\tau_\star\}$ and write
\[
R_1(a):=\|Aa-C_{0^+}\|_{2,\Gamma_{0^+}}^2,\qquad
R_2(a):=\|A_{\tau|K}a\|_{2,\Gamma_{0^+}}^2,
\]
so that the RN penalty in \eqref{eq:rn-penalty} is $\tfrac{\lambda_{\mathrm{RN}}}{2}R_1(a)+\tfrac{\eta_{\mathrm{RN}}}{2}R_2(a)$.
Let $\mathcal F\subset\mathbb R^P$ be a closed convex feasible set and let $g:\mathcal F\to\mathbb R$ 
collect all other (convex) terms of the objective in \eqref{eq:master-qp}.
Assume the boundary conditions are attainable on $\Gamma_{0^+}$, i.e.
\[
\mathcal C:=\{a\in\mathcal F:\ R_1(a)=0,\ R_2(a)=0\}\neq\varnothing.
\]
For $\lambda,\eta>0$ define
\[
J_{\lambda,\eta}(a):=g(a)+\frac{\lambda}{2}R_1(a)+\frac{\eta}{2}R_2(a),
\quad
a_{\lambda,\eta}\in\arg\min_{a\in\mathcal F} J_{\lambda,\eta}(a).
\]
Then, as $\min\{\lambda,\eta\}\to\infty$,
\[
R_1(a_{\lambda,\eta})\to 0
\qquad\text{and}\qquad
R_2(a_{\lambda,\eta})\to 0,
\]
so $Aa_{\lambda,\eta}\to C_{0^+}$ and $A_{\tau|K}a_{\lambda,\eta}\to 0$ on $\Gamma_{0^+}$.
Moreover, every cluster point $a^\star$ of $\{a_{\lambda,\eta}\}$ solves the equality–constrained problem
\[
\min\{\,g(a):\ a\in\mathcal C\,\}.
\]
\end{lemma}

\begin{proof}
\underline{\textit{Step 1 (residuals vanish).}}
Pick $a^0\in\mathcal C$ (exists by assumption, discussed later), so $R_1(a^0)=R_2(a^0)=0$.
By optimality of $a_{\lambda,\eta}$,
\[
J_{\lambda,\eta}(a_{\lambda,\eta})
\;\le\;J_{\lambda,\eta}(a^0)
\;=\;g(a^0).
\]
Hence, for all $\lambda,\eta>0$,
\[
g(a_{\lambda,\eta})+\frac{\lambda}{2}R_1(a_{\lambda,\eta})
+\frac{\eta}{2}R_2(a_{\lambda,\eta})
\;\le\; g(a^0).
\]
Since $g$ is bounded below on $\mathcal F$ (true in our QP, e.g.\ by the ridge term; (Remark \ref{rem:g-lower-bound})), there exists $m>-\infty$
with $g(a)\ge m$ for all $a\in\mathcal F$. Therefore
\[
\frac{\lambda}{2}R_1(a_{\lambda,\eta})+\frac{\eta}{2}R_2(a_{\lambda,\eta})
\;\le\; g(a^0)-g(a_{\lambda,\eta})
\;\le\; g(a^0)-m \;=:\; C < \infty.
\]
Let $\min\{\lambda,\eta\}\to\infty$. The left-hand side is a sum of nonnegative terms with coefficients
diverging to $+\infty$, so necessarily
\[
R_1(a_{\lambda,\eta})\to 0
\quad\text{and}\quad
R_2(a_{\lambda,\eta})\to 0.
\]
By linearity of $A$ and $A_{\tau|K}$, this yields $Aa_{\lambda,\eta}\to C_{0^+}$ and
$A_{\tau|K}a_{\lambda,\eta}\to 0$ on $\Gamma_{0^+}$.

\smallskip
\underline{\textit{Step 2 (limit points solve the constrained problem).}}
Because $g$ is convex and (by the ridge) coercive on $\mathcal F$, the sequence $\{a_{\lambda,\eta}\}$ is bounded; thus it has cluster points. Let $a_{\lambda,\eta}\to a^\star$ along some subsequence. From Step~1 and continuity of the linear maps, $a^\star\in\mathcal C$.

For any $a\in\mathcal C$, we have $J_{\lambda,\eta}(a)=g(a)$, hence
\[
J_{\lambda,\eta}(a_{\lambda,\eta})\le J_{\lambda,\eta}(a)=g(a)
\quad\Rightarrow\quad
g(a_{\lambda,\eta})\le g(a)\qquad\forall\,\lambda,\eta.
\]
Taking $\limsup$ and using lower semicontinuity of $g$,
\[
g(a^\star)\ \le\ \liminf g(a_{\lambda,\eta})\ \le\ \limsup g(a_{\lambda,\eta})\ \le\ g(a)\quad\forall\,a\in\mathcal C,
\]
so $g(a^\star)=\min_{x\in\mathcal C} g(x)$. This completes the proof.
\end{proof}

\begin{remark}
If exact feasibility on $\Gamma_{0^+}$ is relaxed (e.g.\ excluding an ATM tube $|m|\le c\sqrt{\tau}$),
interpret $R_1,R_2$ with that restriction; the same argument applies. If $\mathcal C=\varnothing$, then
$R_1(a_{\lambda,\eta})$ and $R_2(a_{\lambda,\eta})$ converge to their joint infimum and $a_{\lambda,\eta}$
approaches the minimum–violation compromise.
\end{remark}

\begin{remark}[Lower bound for $g$]\label{rem:g-lower-bound}
In \eqref{eq:master-qp} the function $g(a)$ (all terms except the RN penalty) is a sum of
positive–semidefinite quadratics and squared norms with positive weights:
$\tfrac12\|W^{1/2}(Aa-y)\|_2^2$, $\tfrac{\lambda_{\mathrm{ridge}}}{2}a^\top\Lambda a$,
$\tfrac{\lambda_{\mathrm{DW}}}{2}a^\top E^\top L^+ E a$, 
$\tfrac{\lambda_{\Omega}}{2}a^\top U_\omega^\top M_\omega^\top M_\omega U_\omega a$, etc.
After dropping the constant $\tfrac12 y^\top W y$ from the LS term, we have $g(a)\ge 0$
for all $a\in\mathcal F$. Hence $m:=\inf_{a\in\mathcal F} g(a)\ge 0$ is finite, which is the bound
used in the proof of Lemma~\ref{lem:rn-consistency}.
\end{remark}

\begin{remark}[Attainability is an assumption]\label{rem:rn-attainability}
Lemma~\ref{lem:rn-consistency} assumes the boundary conditions are attainable on $\Gamma_{0^+}$,
i.e.\ $\mathcal C:=\{a\in\mathcal F:\ R_1(a)=0,\ R_2(a)=0\}\neq\varnothing$.
This is not automatic; it depends on the basis and the feasible set $\mathcal F$.
\end{remark}

\begin{prop}[Sufficient discrete condition for attainability]\label{prop:rn-attainability}
Let $A_{\Gamma}\in\R^{G_0\times P}$ and $A_{\tau|K,\Gamma}\in\R^{G_0\times P}$ denote
$A$ and $A_{\tau|K}$ restricted to the rows $g\in\Gamma_{0^+}$ (with $G_0:=|\Gamma_{0^+}|$).
Stack the constraints into
\[
S \;:=\; 
\begin{bmatrix}
A_{\Gamma}\\[2pt] A_{\tau|K,\Gamma}
\end{bmatrix} \in \R^{(2G_0)\times P},
\qquad
c \;:=\;
\begin{bmatrix}
C_{0^+}\\[2pt] 0
\end{bmatrix}\in\R^{2G_0}.
\]
If $\mathrm{rank}(S)=2G_0$ (full row rank) and $\mathcal F=\R^P$ (or $\mathcal F$ is any convex set
that contains a solution of $Sa=c$), then $\mathcal C\neq\varnothing$.
In particular, when $P\ge 2G_0$ and $S$ has full row rank, the system $Sa=c$ is solvable and any solution $a$
satisfies $R_1(a)=R_2(a)=0$.
\end{prop}

\begin{proof}
If $\mathrm{rank}(S)=2G_0\le P$, then $\mathrm{Range}(S)=\R^{2G_0}$, so for any right-hand side $c$
there exists $a\in\R^P$ with $Sa=c$. Such an $a$ obeys $A_{\Gamma}a=C_{0^+}$ and
$A_{\tau|K,\Gamma}a=0$, hence $R_1(a)=R_2(a)=0$; if $\mathcal F$ contains one such $a$, then $a\in\mathcal C$.
\end{proof}

\begin{cor}[Practical sufficient conditions]\label{cor:rn-attainability}
Attainability holds if (i) the coefficient space is rich enough so that $P\ge 2G_0$
and the slice-restricted design matrices have independent rows (so $S$ has full row rank), and
(ii) the feasible set $\mathcal F$ does not exclude these solutions.
For Chebyshev tensor bases in $(m,\tau)$, increasing degrees $(K,L)$ makes $S$ generically full row rank
on a fixed grid $\Gamma_{0^+}$; moreover, the boundary value $C_{0^+}\in[0,F_0]$ and the condition
$A_{\tau|K,\Gamma}a=0$ are compatible with usual hard constraints (monotonicity/convexity and bounds).
\end{cor}

\begin{remark}[If attainability fails]\label{rem:rn-approx}
If $\mathcal C=\varnothing$ (e.g.\ degrees too low, or additional constraints forbid exact matching),
then the conclusions of Lemma~\ref{lem:rn-consistency} hold in the approximate sense:
$R_1(a_{\lambda,\eta})\to R_1^\star$ and $R_2(a_{\lambda,\eta})\to R_2^\star$ where
$R_1^\star+R_2^\star=\inf_{a\in\mathcal F}\{R_1(a)+R_2(a)\}$; the minimisers converge to the
minimum-violation compromise on $\Gamma_{0^+}$.
\end{remark}

\paragraph{ATM tube (optional).}
To avoid over–penalizing the thin region where the $O(\sqrt{\tau})$ time value concentrates,
one may exclude $|m|\le c\sqrt{\tau}$ from $\Gamma_{0^+}$ (small $c$) or down–weight those rows.

\paragraph{Parity projection (optional).}
As a linear preprocessing that preserves QP structure, set $a\leftarrow P_{\mathrm{even}}a$ with
$(P_{\mathrm{even}})_{(k,\ell),(k,\ell)}=1$ for even $k$ and $0$ for odd $k$ when evaluating the
first term in \eqref{eq:rn-penalty}; this removes odd–in–$m$ glitches near $\tau\downarrow 0$.

\section{\texorpdfstring
  {Frequency truncation: $\Omega$ taper near the boundary}
  {Frequency truncation: Omega taper near the boundary}}

Use the spectral mask from \S\ref{sec:omega}, but only on short maturities. Let
$U_\omega$ be the fixed frequency chart (e.g.\ separable 2D DCT) and define a maturity–dependent
mask $M_\omega(\tau_g)$ that zeros the highest $m$–frequencies on slices with
$\tau_g\le\tau_\star$ and ramps to $0$ by $2\tau_\star$:
\[
\mathcal{R}_{\Omega}(a)
\;=\;\frac{1}{2}\sum_{g:\,\tau_g\ \text{grid}}
\lambda_{\Omega}(\tau_g)\,\|M_\omega(\tau_g)\,U_\omega a\|_2^2,
\]
\[
\lambda_{\Omega}(\tau)\!=\!\lambda_{\Omega}^{(0)}\!\times\!
\begin{cases}
1,& \tau\le\tau_\star,\\
2-\tau/\tau_\star,& \tau_\star<\tau\le 2\tau_\star,\\
0,& \tau>2\tau_\star.
\end{cases}
\]
This convex quadratic suppresses only the high–$k$ content where butterfly ripples originate;
low modes (ATM/term structure) are left intact by construction.

\section{\texorpdfstring
  {Transport damping: $H^{-1}$ smoothing of density at short maturities}
  {Transport damping: H-1 smoothing of density at short maturities}}
Let $\rho=\partial_{KK}C_f$ and recall the discrete $H^{-1}$ seminorm from \S\ref{sec:dw-mod}:
$\|f\|_{H^{-1}}^2=f^\top L^{+} f$ with $L^{+}\succeq 0$ fixed. We weight it more at short $\tau$:
\[
\mathcal{R}_{\mathrm{DW}}(a)
\;=\;\frac{1}{2}\sum_{g:\,\tau_g\ \text{grid}}
\lambda_{\mathrm{DW}}(\tau_g)\,\|(E a)_g\|_{H^{-1}}^2,
\qquad
\lambda_{\mathrm{DW}}(\tau)=\lambda_{\mathrm{DW}}^{(0)}\min\{1,\tau_\star/\tau\},
\]
where $Ea$ stacks $A_{KK}a$ slice–wise. This biases the optimiser toward \emph{spreading} density
rather than oscillating it in short strips, eliminating spurious negative lobes.

\section{Convexity and invariance}

Each addend in this chapter is a sum of squares of affine functions of $a$ (or a congruence of a
fixed SPD quadratic), hence convex and QP–compatible. Under the $\Lambda$–module
(\S\ref{sec:Lambda-module}), post–multiply all blocks by $U$ and replace $\Lambda$ by
$\widetilde\Lambda=U^\top\Lambda U$; penalty values on $(Aa, A_{\tau|K}a, A_{KK}a)$ are unchanged.

\section{Practical choices and interaction with scaling}

\begin{itemize}[leftmargin=1.25em]
\item \textbf{Window:} $\tau_\star=5$–$10$ trading days. Apply \eqref{eq:rn-penalty} only on
$\Gamma_{0^+}$; optionally exclude an ATM tube $|m|\le c\sqrt{\tau}$.
\item \textbf{Weights:} Choose $\lambda_{\mathrm{RN}}$ so that the RMS of $Aa-C_{0^+}$ on
$\Gamma_{0^+}$ matches the median band width there; choose $\eta_{\mathrm{RN}}$ to bring short–end
calendar violations under the global $1\%$ budget when combined with the scaled
no–arb penalties (\S\ref{sec:scaling}). Use the short–maturity ramps
$\lambda_{\Omega}(\tau)$ and $\lambda_{\mathrm{DW}}(\tau)$ above.
\item \textbf{Row scaling:} Apply the same p95 row scaling (\S\ref{sec:scaling}) to
$A_{\tau|K}$ when used inside \eqref{eq:rn-penalty}; this keeps the knob $\eta_{\mathrm{RN}}$
comparable to $\lambda_{\mathrm{NA}}$.
\end{itemize}

\chapter{Diagnostics and Implementation}\label{sec:diag_and_imp}

\section{Structure monitors}
We carry a set of diagnostics that do not enter the optimisation but certify stability.

\paragraph{MON1 (symplectic/volume/reversibility).}
If an auxiliary Hamiltonian stepper is used (e.g.\ to generate $U$ or kicks),
we compute the discrete symplectic defect and the map‐determinant on random probes.
Both are monitor scalars and must remain below pre‐set tolerances.

\paragraph{MON2 (RN residual/commuting defect).}
Report \emph{RN residual}
$\norm{Aa - C_{0^+}}_{\Gamma_{0^+}}$,
and \emph{calendar commutator}
$\norm{C(Aa)}_2$ on the grid, where
$C := \widetilde A_{\tau|K}\widetilde A_K - \widetilde A_K \widetilde A_{\tau|K}
\in\R^{G\times G}$ is the grid-space commutator from
Section~\ref{sec:omega}.

\paragraph{MON3 (aliasing).}
Let $\mathcal{E}_{\text{hi}}$ be the share of modal energy in the upper third of
$(k,\ell)$ indices. Large values predict convexity noise; the $\Omega$ penalty is tuned to
cap $\mathcal{E}_{\text{hi}}$.

\paragraph{Q (Egorov bridge).}
Under short‐time linearisation, classical transport of observables commutes with
quantum propagation up to $\mathcal{O}(\tau)$ (Egorov’s theorem).
We monitor the deviation of $C_f$ pushed through the (linearised) forward drift
versus the surface rebuilt at $\tau+\delta\tau$; the resulting defect is reported
as absolute/relative scalars. (Purely diagnostic; no constraints are added.)

\section{Implementation notes}
\begin{itemize}[leftmargin=1.25em]
\item \textbf{Basis sizes:} $K\!\in[28,40]$, $L\!\in[22,32]$; grid $G\!\approx\! (2K){\times}(2L)$.
\item \textbf{Scaling:} row p95 to 1; convexity block optionally scaled by a small boost ($\times 2$--$5$).
\item \textbf{Penalties:} $\lambda_{\text{NA}}$ via short probe; $\lambda_{\text{ridge}}$ by GCV on a random 8\% subsample.
\item \textbf{Short--maturity:} set $\tau_\star\approx$ 5--10 trading days; ramp $\lambda_{\text{RN}},\lambda_{\Omega},\lambda_{\text{DW}}$ to zero after $2\tau_\star$.
\item \textbf{Solver:} OSQP with $\varepsilon_{\text{abs}}=\varepsilon_{\text{rel}}=4.5\times 10^{-7}$, adaptive $\rho$ (default $0.1$), polishing enabled, and a 900-second time cap; other options stay at their defaults.
\end{itemize}

\chapter[Hamiltonian Fog Post-Fit in Price Space]%
  {Hamiltonian Fog Post-Fit in\\ Price Space}
\label{sec:postfit-hamiltonian-fog}

The global Chebyshev–QP fit constructed in the previous Chapters already enforces
static no–arbitrage on a dense collocation grid and is tuned to achieve
approximately $99\%$ within–band coverage with at most $1\%$ grid violations.
However, some trading dates and regions of the $(m,\tau)$–plane
remain unsatisfactory even after the main fit and
the short–maturity remedy.

The main QP already delivers a strong baseline surface on most dates. For \emph{calm} years such as 2019 and 2022-23, a single choice of $\mu$ and $\lambda_{NA}$ is enough to reach the target. 
In these regimes, the badness field $\omega$ is close to zero everywhere, so any local post-fit is minimal.

However, in 2020-21 quotes are noisier, and cross-sectional inconsistencies between strikes and maturities are more common. On these stressed dates, the baseline QP is forced into local compromises (either coverage or no-arbitrage deteriorates). Empirically, the misfit is concentrated in small regions.

In this chapter, we describe a second \emph{local} post-fit layer that takes, for
each trading date \(t\), a forward-discounted option baseline surface
\(C_f^0(m,\tau)\) and returns a corrected nodal surface \(u_t^\star\) on a
structured \((m,\tau)\)-grid. The post-fit acts only on regions where the baseline
fit is locally problematic (poor band coverage and/or fragile static
no-arbitrage).

\section{Baseline grid surface and quotes}\label{sec:baseline-grid}

Fix a trading date $t$ and let
\[
\mathcal{G} := \{(m_i,\tau_j) : i = 1,\dots,n_m,\ j = 1,\dots,n_\tau\}
\]
be a structured working grid in log-moneyness $m$ and maturity $\tau$. Denote its
cardinality by
\[
G := |\mathcal{G}| = n_m n_\tau,
\]
fixing any one-to-one enumeration of $\mathcal{G}$ by indices
$g \in \{1,\dots,G\} \leftrightarrow (i(g),j(g)) \in \mathcal{G}$.

Let $C_f^0(m,\tau)$ be the baseline forward-discounted call surface obtained from
the main QP fit on date $t$. The corresponding nodal values on $\mathcal{G}$ are given by
\[
u^0_{i,j} := C_f^0(m_i,\tau_j), \qquad (i,j)\in\mathcal{G}.
\]
Collecting them into a vector, we obtain $u^0 \in \mathbb{R}^{G}$.

For the same date $t$, consider a set of cleaned forward-discounted quote bands
\[
\{(m_q,\tau_q,b_q,a_q)\}_{q=1}^{Q}, \qquad 0 \leq b_q \leq a_q,
\]
where $(m_q,\tau_q)$ denotes the quote location in $(m,\tau)$ and
$[b_q,a_q]$ is the corresponding bid–ask interval in forward-discounted units.

Let $S \in \mathbb{R}^{Q\times G}$ be the (fixed) bilinear interpolation operator
that maps nodal values on $\mathcal{G}$ to model prices at the quote locations.
For any nodal field $u \in \mathbb{R}^{G}$ we write
\[
C_q(u) := (S u)_q, \qquad q = 1,\dots,Q,
\]
so $C_q(u)$ is the model forward-discounted call price at $(m_q,\tau_q)$ implied
by the nodal surface $u$.

\section{Badness map and patch decomposition}

We now detect where the baseline surface is locally problematic.

\begin{definition}[Baseline band misfit]\label{def:baseline-band-misfit}
For each quote $q \in \{1,\dots,Q\}$, the baseline band violation is defined as
\[
d_q(u^0)
:=
\operatorname{dist}\big((S u^0)_q,\ [b_q,a_q]\big)
=
\max\{b_q - (S u^0)_q,\ 0,\ (S u^0)_q - a_q\}\geq 0.
\]
Collecting all quote-wise misfits into a vector
\[
d(u^0) := \big(d_q(u^0)\big)_{q=1}^Q \in \mathbb{R}^Q_{\ge 0},
\]
regard $d(u^0)$ as the baseline distance-to-band profile at the true quote locations.
\end{definition}

We transport these quote-level misfits to the working grid $\mathcal{G}$ via a fixed
linear operator
\[
R^{\mathrm{band}} \in \mathbb{R}^{G\times Q},
\]
(for example by locally averaging nearby quotes around each grid node). The purpose is to observe these violations at the grid nodes rather than at the quote locations. We write
\[
w^{\mathrm{band}} := R^{\mathrm{band}} d(u^0) \in \mathbb{R}^G,\qquad
w^{\mathrm{band}}_{i,j} \ge 0,\ (i,j)\in\mathcal{G}.
\]
In practice the entries of $R^{\mathrm{band}}$ are chosen nonnegative and supported
only on quotes $(m_q,\tau_q)$ lying in a small neighbourhood of $(m_i,\tau_j)$, so
that $w^{\mathrm{band}}_{i,j}$ is a local aggregated band-misfit around that node.

Next, build a single scalar at each grid node to see how badly the baseline surface $u^0$ violates static no-arbitrage at that grid node. Let $w^{\mathrm{noarb}}\in\mathbb{R}^G_{\ge 0}$ be any nonnegative
\emph{static no-arbitrage defect field} obtained by aggregating local violations of
discrete bounds, strike monotonicity, strike convexity, and calendar monotonicity
on the grid when evaluated at $u^0$. Concretely, one may define at each node
$(i,j)$:

\begin{definition}[Bound violation $v^{\mathrm{bnd}}_{i,j}$]
Bounds are $0\leq u^0_{i,j} \leq F_{i,j}$. At each node, define the distance to the interval $[0, F_{i,j}]$
\[
v^{\mathrm{bnd}}_{i,j} = \max \big\{ -u^0_{i,j},\; 0,\; u^0_{i,j}-F_{i,j} \big\}.
\]
Therefore,
\begin{itemize}
  \item If $u^0_{i,j}\in [0, F_{i,j}]$, then $v^{\mathrm{bnd}}_{i,j}=0$.
  \item If $u^0_{i,j}<0$, then $v^{\mathrm{bnd}}_{i,j}=-u^0_{i,j}$.
  \item If $u^0_{i,j}>F_{i,j}$, then $v^{\mathrm{bnd}}_{i,j}=u^0_{i,j}-F_{i,j}$.
\end{itemize}
\end{definition}

\begin{definition}[Strike monotonicity violation $v^{\mathrm{mono}}_{i,j}$]
Static no-arbitrage requires call prices to be non-increasing in strike at fixed maturity. On the grid,
this means that along each maturity slice $j$ we should have
\[
u^0_{i+1,j} - u^0_{i,j} \le 0,\qquad i=1,\dots,n_m-1.
\]
Define the forward slope in $m$ between nodes $i$ and $i+1$ by
\[
\Delta^{\mathrm{mono}}_{i+\frac12,j}
:=
u^0_{i+1,j} - u^0_{i,j},
\qquad i=1,\dots,n_m-1.
\]
If $\Delta^{\mathrm{mono}}_{i+\frac12,j}\le 0$, there is no monotonicity issue on that edge; if
$\Delta^{\mathrm{mono}}_{i+\frac12,j}>0$, the price goes up in strike there (violation).
Define the edge violation
\[
d^{\mathrm{mono}}_{i+\frac12,j}
:=
\max\big\{\Delta^{\mathrm{mono}}_{i+\frac12,j},\,0\big\}.
\]
We then attach a node-based violation by aggregating the incident edges. For interior nodes
$2\le i\le n_m-1$,
\[
v^{\mathrm{mono}}_{i,j}
:=
\max\big\{
d^{\mathrm{mono}}_{i-\frac12,j},
d^{\mathrm{mono}}_{i+\frac12,j}
\big\},
\]
and at the boundaries we use the single available edge,
\[
v^{\mathrm{mono}}_{1,j} := d^{\mathrm{mono}}_{1+\frac12,j},\qquad
v^{\mathrm{mono}}_{n_m,j} := d^{\mathrm{mono}}_{n_m-\frac12,j}.
\]
Thus,
\begin{itemize}
  \item If all nearby slopes around $(i,j)$ are non-positive, then $v^{\mathrm{mono}}_{i,j}=0$.
  \item If some local slope is upward, $v^{\mathrm{mono}}_{i,j}$ records the largest upward jump
        touching that node.
\end{itemize}
\end{definition}

\begin{definition}[Strike convexity violation $v^{\mathrm{conv}}_{i,j}$]
Static no-arbitrage also requires convexity in strike. On a uniform $m$-grid, this is encoded by
non-negative second differences
\[
\Delta^2_{i,j} := u^0_{i+1,j} - 2u^0_{i,j} + u^0_{i-1,j} \ge 0,
\qquad i=2,\dots,n_m-1.
\]
Define the convexity defect at the central node by
\[
v^{\mathrm{conv}}_{i,j}
:=
\begin{cases}
\max\big\{-\Delta^2_{i,j},\,0\big\}, & 2\le i\le n_m-1,\\[3pt]
0, & i=1 \text{ or } i=n_m \text{ (no centred stencil).}
\end{cases}
\]
Therefore,
\begin{itemize}
  \item If $\Delta^2_{i,j}\ge 0$, there is no convexity issue and $v^{\mathrm{conv}}_{i,j}=0$.
  \item If $\Delta^2_{i,j}<0$, the profile is locally concave in strike (violation) and
        $v^{\mathrm{conv}}_{i,j}=-\Delta^2_{i,j}$ measures how badly convexity is violated.
\end{itemize}
On a non-uniform $m$-grid one can replace $\Delta^2_{i,j}$ by the standard three-point
second-derivative formula with unequal spacings; for the purposes of constructing a badness
indicator $w^{\mathrm{noarb}}$ the simple second difference is typically sufficient.
\end{definition}

\begin{definition}[Calendar violation $v^{\mathrm{cal}}_{i,j}$ via $A_{\tau|K}$]
Let $a^\star$ be the baseline coefficient vector and let $A_{\tau|K}$ be the
calendar operator at fixed strike.
Evaluate
\[
h := A_{\tau|K} a^\star \in \mathbb{R}^G,
\]
and reshape $h$ on the grid as $h_{i,j}$.
Define the node-wise calendar defect by
\[
v^{\mathrm{cal}}_{i,j} := \max\{-h_{i,j},\,0\}.
\]
Then $v^{\mathrm{cal}}_{i,j} = 0$ whenever $(\partial_\tau C_f)(K,\tau)\big|_{K} \ge 0$
at $(m_i,\tau_j)$, and $v^{\mathrm{cal}}_{i,j}$ measures the local size of negative
calendar slopes at fixed strike.
\end{definition}

With the bound violation $v^{\mathrm{bnd}}_{i,j}$, strike monotonicity violation $v^{\mathrm{mono}}_{i,j}$, strike convexity violation $v^{\mathrm{conv}}_{i,j}$ and calendar violation $v^{\mathrm{cal}}_{i,j}$ define the static no-arbitrage defect field $w^{\mathrm{noarb}}$.

\begin{definition}[Static no-arbitrage defect field $w^{\mathrm{noarb}}$]
Given the node-wise violations
\[
v^{\mathrm{bnd}}_{i,j},\quad
v^{\mathrm{mono}}_{i,j},\quad
v^{\mathrm{conv}}_{i,j},\quad
v^{\mathrm{cal}}_{i,j}\ \ge 0,
\]
defined respectively for bounds, strike monotonicity, strike convexity, and calendar
constraints at $(m_i,\tau_j)$, the \emph{static no-arbitrage defect field} is
\[
w^{\mathrm{noarb}}_{i,j}
:=
\max\big\{
v^{\mathrm{bnd}}_{i,j},
v^{\mathrm{mono}}_{i,j},
v^{\mathrm{conv}}_{i,j},
v^{\mathrm{cal}}_{i,j}
\big\},
\qquad (i,j)\in\mathcal{G}.
\]
Thus $w^{\mathrm{noarb}}_{i,j}\ge 0$ for all $(i,j)$, and
$w^{\mathrm{noarb}}_{i,j}=0$ whenever all discrete no-arbitrage inequalities
(bounds, strike monotonicity, strike convexity, and calendar monotonicity) hold
without violation in a neighbourhood of $(m_i,\tau_j)$.

Any alternative construction of a field $w^{\mathrm{noarb}}\in\mathbb{R}^G_{\ge 0}$
with the same qualitative properties
\begin{itemize}
  \item $w^{\mathrm{noarb}}_{i,j}\ge 0$ for all $(i,j)$, and
  \item $w^{\mathrm{noarb}}_{i,j}=0$ whenever all discrete no-arbitrage
        inequalities are satisfied (with margin) near $(m_i,\tau_j)$,
\end{itemize}
is equally admissible for the purposes of the badness map construction below.
\end{definition}

The band-misfit field $w^{\mathrm{band}}$ and the static no-arbitrage defect field
$w^{\mathrm{noarb}}$ provide two complementary scalar diagnostics on the grid
$\mathcal{G}$: the former reflects how hard it is for the baseline surface to
respect the bid-ask bands, while the latter reflects how fragile the static
shape constraints are in a neighborhood of each node. For the purposes of
patch detection we now compress these two pieces of information into a single
scalar badness field on $\mathcal{G}$, allowing for a tunable trade-off
between band fit and no-arbitrage robustness.

\begin{definition}[Raw and smoothed badness field]
Fix positive scalars $\alpha_{\mathrm{band}},\alpha_{\mathrm{noarb}}>0$. The
\emph{raw badness field} on $\mathcal{G}$ is
\[
\tilde w_{i,j}
:= \alpha_{\mathrm{band}}\,w^{\mathrm{band}}_{i,j}
 + \alpha_{\mathrm{noarb}}\,w^{\mathrm{noarb}}_{i,j},
\qquad (i,j)\in\mathcal{G}.
\]
Let $K_{\sigma}$ be a fixed separable Gaussian kernel on the grid,
$K_{\sigma}(i,j) = k_{\sigma}^{(m)}(i)\,k_{\sigma}^{(\tau)}(j)$, and let
$*$ denote discrete convolution on $\mathcal{G}$:
\[
(K_{\sigma} * \tilde w)_{i,j}
:= \sum_{(i',j')\in\mathcal{G}} K_{\sigma}(i-i',j-j')\,\tilde w_{i',j'}.
\]
Define the (componentwise) clipping operator
\[
\operatorname{Clip}_{[0,1]}(x)_{i,j}
:= \min\{1,\max\{0,x_{i,j}\}\}.
\]
The \emph{smoothed badness field} is then
\[
w := \operatorname{Clip}_{[0,1]}(K_{\sigma} * \tilde w)\ \in\ [0,1]^G.
\]
\end{definition}

The fixed separable Gaussian kernel on the grid $K_{\sigma}$ is a bell-shaped weight function centered at 0 and decaying as you move away. Separable means that it can factor into a product of a 1D kernel in $m$ and a 1D kernel inn $\tau$. 
The convolution $(K_{\sigma} * \tilde w)_{i,j}$ is a weighted average of the raw badness $\tilde w_{i,j}$ in a neighborhood of $(i,j)$ with weights given by the Gaussian kernel evaluated at offsets $(i-i',j-j')$.

\begin{example}[Single Spike]
Consider a 1D grid with indices $i=1,\dots,7$ and a raw badness vector
\[
\tilde w = (0,\ 0,\ 0,\ 1,\ 0,\ 0,\ 0),
\]
so there is a single spike of badness at $i=4$. Take a simple discrete kernel
\[
K = \Big(\tfrac14,\ \tfrac12,\ \tfrac14\Big),
\]
interpreted as $K(-1)=\tfrac14$, $K(0)=\tfrac12$, $K(1)=\tfrac14$, and $K(k)=0$
for $|k|>1$. The convolution $(K * \tilde w)_i$ with zero-padding at the
boundaries is
\[
(K * \tilde w)_i
= \tfrac14\,\tilde w_{i-1} + \tfrac12\,\tilde w_i + \tfrac14\,\tilde w_{i+1},
\qquad i=1,\dots,7.
\]
A direct computation gives
\[
(K * \tilde w)
= (0,\ 0,\ 0.25,\ 0.5,\ 0.25,\ 0,\ 0).
\]
Thus the original spike of height $1$ at $i=4$ is smoothed into a smaller peak
of height $0.5$ at $i=4$ with nonzero neighbours of height $0.25$ at $i=3$ and
$i=5$: the mass has been spread out and diluted.
\end{example}

\begin{example}[Cluster]
Now consider a cluster of three bad nodes
\[
\tilde w = (0,\ 0,\ 1,\ 1,\ 1,\ 0,\ 0),
\]
again on indices $i=1,\dots,7$, with the same kernel
\[
K = \Big(\tfrac14,\ \tfrac12,\ \tfrac14\Big).
\]
Using the same convolution formula
\[
(K * \tilde w)_i
= \tfrac14\,\tilde w_{i-1} + \tfrac12\,\tilde w_i + \tfrac14\,\tilde w_{i+1},
\qquad i=1,\dots,7,
\]
we obtain
\[
\begin{aligned}
(K * \tilde w)_2 &= \tfrac14\cdot 0 + \tfrac12\cdot 0 + \tfrac14\cdot 1 = 0.25,\\
(K * \tilde w)_3 &= \tfrac14\cdot 0 + \tfrac12\cdot 1 + \tfrac14\cdot 1 = 0.75,\\
(K * \tilde w)_4 &= \tfrac14\cdot 1 + \tfrac12\cdot 1 + \tfrac14\cdot 1 = 1.00,\\
(K * \tilde w)_5 &= \tfrac14\cdot 1 + \tfrac12\cdot 1 + \tfrac14\cdot 0 = 0.75,\\
(K * \tilde w)_6 &= \tfrac14\cdot 1 + \tfrac12\cdot 0 + \tfrac14\cdot 0 = 0.25,
\end{aligned}
\]
and $(K * \tilde w)_1 = (K * \tilde w)_7 = 0$. Hence
\[
(K * \tilde w)
= (0,\ 0.25,\ 0.75,\ 1.0,\ 0.75,\ 0.25,\ 0).
\]
In this case a contiguous cluster of bad nodes remains a single coherent bump:
the central node retains height $1$ and its neighbours are only slightly
reduced to $0.75$, while a halo of smaller values $0.25$ appears around the
cluster. This illustrates how Gaussian smoothing preserves genuine regions of
badness while softening their edges and suppressing isolated spikes.

\end{example}

Therefore, a single noisy spike in $\tilde w$ will spread out to neighbors, clusters of large $\tilde w$ will be smoothed into a broader hot region rather than an isolated piece and a smooth, spatially coherent badness image is created.

The clipping operator is applied to obtain a dimensionless, bounded heatmap which is comparable across dates and assets. One can therefore chose a threshold $\theta \in (0,1)$ and obtain patches from the connected components of $\{\omega_{i,j}>\theta \}$. Moreover, the smoothing step ensures that patches correspond to regions rather than isolated single nodes.

Graphically, $w$ can be viewed as a heatmap on the $(m,\tau)$-plane:
nodes with $w_{i,j}\approx 0$ are locally well-behaved (good band coverage and
robust static no-arbitrage), while nodes with $w_{i,j}$ close to $1$ lie in fragile or hard-to-fit regions.

\begin{figure}[htbp]
  \centering
  \begin{tikzpicture}[scale=0.7]

    \draw[->] (0,0) -- (6.2,0) node[right] {$m$};
    \draw[->] (0,0) -- (0,4.2) node[above] {$\tau$};

    \foreach \i in {0,...,5} {
      \foreach \j in {0,...,3} {
        \fill[gray!10] (\i,\j) rectangle ++(1,1);
        \draw[gray!40, line width=0.2pt] (\i,\j) rectangle ++(1,1);
      }
    }

    \foreach \i/\j/\c in {
      2/1/gray!60,
      3/1/gray!70,
      4/1/gray!60,
      2/2/gray!70,
      3/2/gray!90,
      4/2/gray!70,
      3/3/gray!60
    }{
      \fill[\c] (\i,\j) rectangle ++(1,1);
      \draw[gray!40, line width=0.2pt] (\i,\j) rectangle ++(1,1);
    }

    \draw[red, very thick, rounded corners]
      (2,1) rectangle (5,3);
    \node[red] at (3.5,3.2) {$\Omega_p$};

    \node[anchor=west] at (6.4,2.8) {\small darker = larger $w_{i,j}$};
    \node[anchor=west] at (6.4,2.3) {\small lighter = smaller $w_{i,j}$};

  \end{tikzpicture}
  \caption{Schematic badness map $w_{i,j}$ on the $(m,\tau)$ grid $\mathcal{G}$.
  Darker cells indicate regions where the smoothed badness field is large; a
  connected high-badness region is shown as a patch $\Omega_p$.}
\end{figure}
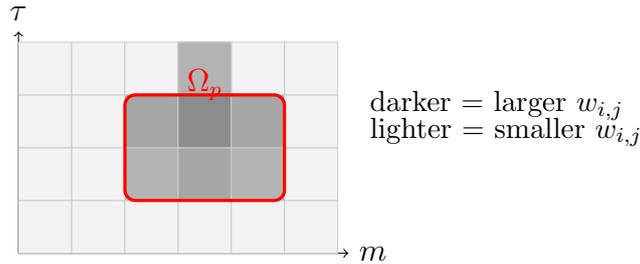

We now threshold $w$ and decompose the high-badness region into connected
components.

\begin{definition}[Active set and patches]
Fix a threshold $\theta\in(0,1)$. The \emph{active set} is
\[
A := \{(i,j)\in\mathcal{G} : w_{i,j} > \theta\}.
\]
Equip $\mathcal{G}$ with a nearest-neighbour graph structure, either
\emph{4-neighbour} (edges between $(i,j)$ and $(i\pm1,j)$, $(i,j\pm1)$) or
\emph{8-neighbour} (4-neighbour plus diagonals $(i\pm1,j\pm1)$). Two nodes of
$A$ are said to be connected if they are joined by a path of neighbours in this
graph. The connected components of $A$ are then
\[
A = \bigsqcup_{p\in\mathcal{P}} \Omega_p,
\]
where each $\Omega_p\subset \mathcal{G}$ is a maximal connected subset of $A$
(with respect to the chosen neighbourhood) and is called a \emph{patch}.
\end{definition}

Nodes in $\bigcup_{p}\Omega_p$ belong to locally difficult regions and are candidates
for post-fit adjustment; nodes in
$\mathcal{G}\setminus\bigcup_{p}\Omega_p$ are left unchanged by the post-fit. The
construction is entirely local and depends only on the baseline misfit and static
defect fields at the given date, not on the calendar regime; both calm and stressed
dates are treated identically.

From now on we fix a single patch $\Omega\subset\mathcal{G}$ and describe the
patch-level post-fit problem.

\section{Discrete 3D fog on \texorpdfstring{\((m,\tau,u)\)}{(m,τ,u)}}\label{sec:3D-fog}

On a fixed patch $\Omega\subset\mathcal{G}$ we will represent local pricing
uncertainty and potential quote noise by a discretised probability density over
the three-dimensional space $(m,\tau,u)$, where $u$ denotes forward-discounted
call price.

\subsection{3D lattice and fog variables}

\begin{definition}[3D fog lattice]
Let $\Omega\subset\mathcal{G}$ be a patch with cardinality $N_\Omega := |\Omega|$.
Fix a finite, strictly increasing sequence of price levels
\[
U := \{u_k\}_{k=1}^{n_u} \subset \mathbb{R},
\qquad
u_1 < u_2 < \dots < u_{n_u},
\]
spanning a relevant price range (for example from $0$ up to a suitable multiple
of the local forward). The associated three-dimensional lattice of
$(m,\tau,u)$-nodes on $\Omega$ is
\[
\mathcal{L}_\Omega
:=
\{(i,j,k) : (i,j)\in\Omega,\ k=1,\dots,n_u\}.
\]
\end{definition}

\begin{definition}[Fog variables and normalisation]
For each $(i,j,k)\in\mathcal{L}_\Omega$ introduce a nonnegative \emph{fog
variable} $\pi_{i,j,k}\ge 0$. Collecting all such variables into a single vector
\[
\pi := (\pi_{i,j,k})_{(i,j,k)\in\mathcal{L}_\Omega}
\in \mathbb{R}^{N_\Omega n_u},
\]
we impose the global normalisation
\[
\sum_{(i,j)\in\Omega} \sum_{k=1}^{n_u} \pi_{i,j,k} = 1.
\]
The associated feasible set of fog configurations on $\Omega$ is the simplex
\[
\mathcal{C}_\pi(\Omega)
:=
\Big\{
\pi \in \mathbb{R}^{N_\Omega n_u}_{\ge 0}
:
\sum_{(i,j)\in\Omega} \sum_{k=1}^{n_u} \pi_{i,j,k} = 1
\Big\}.
\]
\end{definition}

Thus $\pi$ is a discrete probability measure on the finite set $\mathcal{L}_\Omega$:
each $\pi_{i,j,k}$ is the probability mass (or fog mass) attached to the node
$(m_i,\tau_j,u_k)$, and $\mathcal{C}_\pi(\Omega)$ is a standard $(N_\Omega n_u-1)$-dimensional
simplex.

At fixed $(i,j)\in\Omega$, the vertical profile
\[
\{\pi_{i,j,k}\}_{k=1}^{n_u}
\]
describes the distribution of fog across price levels $\{u_k\}$ at that grid
location. In particular, we may interpret:
\begin{itemize}
  \item $\pi_{i,j,k}$ large for some $k$ as assigning high plausibility to local
        prices near $u_k$ at $(m_i,\tau_j)$;
  \item a spread-out vertical profile as expressing substantial local uncertainty
        over $u$;
  \item a concentrated profile as expressing locally precise information.
\end{itemize}

\subsection{2D noise marginal on the patch}

\begin{definition}[2D marginal noise density]
Given $\pi\in\mathcal{C}_\pi(\Omega)$, define the \emph{2D marginal} of the fog
on $\Omega$ by
\[
n_{i,j} := \sum_{k=1}^{n_u} \pi_{i,j,k},
\qquad (i,j)\in\Omega.
\]
\end{definition}

The marginal $n := (n_{i,j})_{(i,j)\in\Omega}$ satisfies
\[
n_{i,j}\ge 0,\qquad
\sum_{(i,j)\in\Omega} n_{i,j} = 1,
\]
and can be interpreted as a probability distribution on $\Omega$:
\[
n_{i,j} = \mathbb{P}\{(m,\tau)\text{ lies at node }(m_i,\tau_j)\}.
\]
Equivalently, $n_{i,j}$ is the \emph{total fog mass} sitting above $(m_i,\tau_j)$
when one integrates out the price dimension $u$.

Where $n_{i,j}$ is relatively large, the fog is thick in the $(m,\tau)$-plane
and the local order book in that region is regarded as noisy or unreliable;
where $n_{i,j}$ is small, very little fog mass resides and the local book is
regarded as comparatively clean.

When $n_{i,j}>0$, it is often convenient to speak of the \emph{conditional
vertical distribution} at $(i,j)$,
\[
p_{i,j,k}
:=
\frac{\pi_{i,j,k}}{n_{i,j}},
\qquad k=1,\dots,n_u,
\]
which is a discrete probability distribution on $\{u_k\}$ satisfying
$\sum_{k=1}^{n_u} p_{i,j,k} = 1$. In terms of this factorisation,
\[
\pi_{i,j,k} = n_{i,j}\,p_{i,j,k},
\]
so that $n$ encodes where fog mass is located on the patch, and $p$ encodes
how that fog is distributed vertically in price at each node.

\subsection{Schematic geometry}

The geometry of the fog on a patch can be visualised as a stack of vertical
columns above the $(m,\tau)$ nodes in $\Omega$, each column sampled at levels
$\{u_k\}$:

\begin{figure}[H]
  \centering
  \tdplotsetmaincoords{65}{120}
  \begin{tikzpicture}[scale=1.0,tdplot_main_coords]

    \draw[->] (0,0,0) -- (4.5,0,0) node[below right] {$m$};
    \draw[->] (0,0,0) -- (0,4.0,0) node[left] {$\tau$};
    \draw[->] (0,0,0) -- (0,0,3.0) node[above] {$u$};

    \foreach \i in {0,...,3} {
      \foreach \j in {0,...,2} {
        \draw[gray!40] (\i,\j,0) -- ++(1,0,0) -- ++(0,1,0) -- ++(-1,0,0) -- cycle;
      }
    }

    \newcommand{\fogcolumn}[3]{%
      \pgfmathsetmacro{\ix}{#1}
      \pgfmathsetmacro{\jy}{#2}
      \pgfmathsetmacro{\hh}{#3}
      \fill[gray!30] (\ix,\jy,0) -- ++(1,0,0) -- ++(0,0,\hh) -- ++(-1,0,0) -- cycle;
      \fill[gray!45] (\ix+1,\jy,0) -- ++(0,1,0) -- ++(0,0,\hh) -- ++(0,-1,0) -- cycle;
      \fill[gray!60] (\ix,\jy+1,0) -- ++(1,0,0) -- ++(0,0,\hh) -- ++(-1,0,0) -- cycle;
      \fill[gray!75] (\ix,\jy,\hh) -- ++(1,0,0) -- ++(0,1,0) -- ++(-1,0,0) -- cycle;
      \draw[gray!70] (\ix,\jy,0) -- ++(1,0,0) -- ++(0,1,0) -- ++(-1,0,0) -- cycle;
      \draw[gray!70] (\ix,\jy,\hh) -- ++(1,0,0) -- ++(0,1,0) -- ++(-1,0,0) -- cycle;
      \draw[gray!70] (\ix,\jy,0) -- ++(0,0,\hh);
      \draw[gray!70] (\ix+1,\jy,0) -- ++(0,0,\hh);
      \draw[gray!70] (\ix,\jy+1,0) -- ++(0,0,\hh);
      \draw[gray!70] (\ix+1,\jy+1,0) -- ++(0,0,\hh);
    }

    \fogcolumn{1}{1}{0.9}  
    \fogcolumn{2}{1}{1.2}  
    \fogcolumn{1}{2}{1.8}  
    \fogcolumn{2}{2}{2.2}  

    \draw[red, dashed, thick]
      (1,1,0) -- (3,1,0) -- (3,3,0) -- (1,3,0) -- cycle;
    \node[red] at (2,3.3,0) {$\Omega$};

  \end{tikzpicture}
  \caption{Schematic 3D fog on a patch $\Omega$. The $(m,\tau)$-plane is spanned
  by the horizontal $m$-axis and the in-plane $\tau$-axis; the $u$-axis is vertical.
  Each node $(m_i,\tau_j)\in\Omega$ carries a vertical column of fog mass discretised
  at price levels $\{u_k\}$. The variables $\pi_{i,j,k}$ encode the mass at
  $(m_i,\tau_j,u_k)$; their 2D marginal $n_{i,j}$ is the total mass in the column
  above $(m_i,\tau_j)$.}
\end{figure}
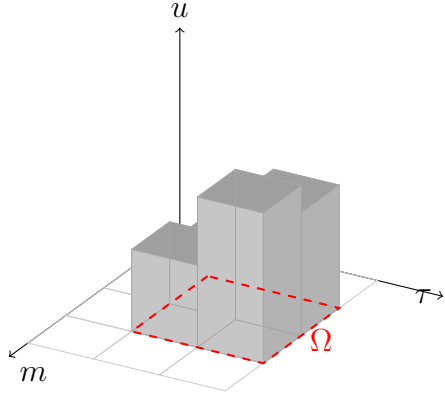

In the optimisation below, the fog configuration $\pi\in\mathcal{C}_\pi(\Omega)$
will be coupled to the nodal price field on $\Omega$ via band-based potentials
and a Hamiltonian energy, with $n_{i,j}$ controlling how strongly local
bid-ask information is enforced at each node.

\section{Patch-level price field and static no-arbitrage}\label{sec:patch-assembly}

Fix a patch $\Omega\subset\mathcal{G}$ with cardinality $N_\Omega:=|\Omega|$ and
baseline nodal surface $u^0\in\mathbb{R}^G$ as in Section~\ref{sec:baseline-grid}.
On $\Omega$ we will allow nodal prices to move, while all off-patch values are
kept fixed at their baseline levels. Static no-arbitrage is imposed on the
\emph{assembled} full-grid surface.

\subsection{Interior price variables and assembly map}

\begin{definition}[Interior price variables and assembled surface]
Choose any one-to-one enumeration of the patch
\[
\Omega = \{(i_\ell,j_\ell)\}_{\ell=1}^{N_\Omega} \subset \mathcal{G}.
\]
The \emph{interior price vector} (unknown) on $\Omega$ is
\[
u_I := (u_{i_\ell,j_\ell})_{\ell=1}^{N_\Omega}\in\mathbb{R}^{N_\Omega},
\]
whose entries correspond to the nodal prices on $\Omega$. Define the \emph{assembly map}
\[
\mathcal{A}_\Omega : \mathbb{R}^{N_\Omega} \to \mathbb{R}^{G},\qquad
u_I \mapsto u(u_I),
\]
by
\[
u(u_I)_{i,j}
:=
\begin{cases}
u_{i,j}, & (i,j)\in\Omega,\\[3pt]
u^0_{i,j}, & (i,j)\notin\Omega.
\end{cases}
\]
\end{definition}

Equivalently, if we write $u\in\mathbb{R}^G$ in the same enumeration as $\mathcal{G}$
and let $P_\Omega\in\mathbb{R}^{G\times N_\Omega}$ be the binary matrix that injects
$u_I$ into the coordinates corresponding to $\Omega$ (and zeros elsewhere), then
the assembly map can be written as
\begin{equation}
\label{eq:assembly-affine}
u(u_I) = P_\Omega u_I + u^{0,\mathrm{off}},
\end{equation}
where $u^{0,\mathrm{off}}\in\mathbb{R}^G$ coincides with $u^0$ on
$\mathcal{G}\setminus\Omega$ and is zero on $\Omega$. Thus $\mathcal{A}_\Omega$ is an
affine map, with linear part $P_\Omega$.

\subsection{Global discrete static no-arbitrage on the grid}

We recall that in Chapters~\ref{sec:noarb-discrete}–\ref{sec:global-noarb-grid} the
discrete static no-arbitrage conditions on the nodal grid $\mathcal{G}$ were encoded
as a finite system of linear inequalities in the full nodal vector
$u\in\mathbb{R}^G$. Concretely, there exists an index set
$\mathcal{I} = \mathcal{I}_{\mathrm{bnd}}\cup\mathcal{I}_{\mathrm{mono}}
\cup\mathcal{I}_{\mathrm{conv}}\cup\mathcal{I}_{\mathrm{cal}}$ and, for each
$\alpha\in\mathcal{I}$, a row vector $\ell_\alpha^\top\in\mathbb{R}^{1\times G}$
and a scalar $r_\alpha\in\mathbb{R}$ such that:
\begin{itemize}
  \item for $\alpha\in\mathcal{I}_{\mathrm{bnd}}$, the inequality
        $\ell_\alpha^\top u \le r_\alpha$ encodes a bound constraint
        $0\le u_{i,j} \le F_{i,j}$ at some node $(i,j)\in\mathcal{G}$;
  \item for $\alpha\in\mathcal{I}_{\mathrm{mono}}$, the inequality encodes a
        discrete strike-monotonicity condition $\partial_K u\le 0$ on a maturity
        slice (e.g.\ $u_{i+1,j}-u_{i,j}\le 0$);
  \item for $\alpha\in\mathcal{I}_{\mathrm{conv}}$, the inequality encodes a
        discrete strike-convexity condition $\partial_{KK} u\ge 0$ on a slice
        (e.g.\ a local second-difference inequality);
  \item for $\alpha\in\mathcal{I}_{\mathrm{cal}}$, the inequality encodes a
        discrete calendar condition $(\partial_\tau u)|_K \ge 0$ at fixed strike,
        implemented via the fixed-strike calendar operator $A_{\tau|K}$ from
        Section~\ref{sec:calendar-operator}.
\end{itemize}
Collecting these, the global static no-arbitrage feasible set on the full grid is
\begin{equation}
\label{eq:Cglobal}
\mathcal{C}_{\mathrm{glob}}
:=
\Big\{
u\in\mathbb{R}^G :
\ell_\alpha^\top u \le r_\alpha,\quad \forall \alpha\in\mathcal{I}
\Big\}.
\end{equation}
This is precisely the intersection of finitely many closed half-spaces in
$\mathbb{R}^G$.

\begin{definition}[Global no-arbitrage operators]\label{def:global-noarb-operators}
Let $\mathcal{G} = \{(i(g),j(g)) : g=1,\dots,G\}$ be an enumeration of the nodal
grid and let $e^{(g)}\in\mathbb{R}^G$ denote the $g$-th standard basis vector.
Write $g(i,j)$ for the index such that $(i(g(i,j)),j(g(i,j)))=(i,j)$.

We define index sets and pairs $(\ell_\alpha,r_\alpha)$ as follows, where $e^{(g)}$ is the $g$-th standard basis vector in $\mathbb{R}^G$ (vector with $1$ in position $g$ and $0$ elsewhere):
\begin{itemize}
  \item \emph{(Bounds)}
    For each $(i,j)\in\mathcal{G}$ with $g=g(i,j)$ define lower and upper indices
    $\alpha=(i,j,\mathrm{lo})$, $\alpha'=(i,j,\mathrm{up})$ and
    \[
      \ell_{(i,j,\mathrm{lo})}^\top := - (e^{(g)})^\top,\quad r_{(i,j,\mathrm{lo})} := 0,\qquad
      \ell_{(i,j,\mathrm{up})}^\top := (e^{(g)})^\top,\quad r_{(i,j,\mathrm{up})} := F_{i,j}.
    \]
    Collect all such indices into $\mathcal{I}_{\mathrm{bnd}}$.
  \item \emph{(Monotonicity)}    
    For each maturity $j$ and $i=1,\dots,n_m-1$, with
    $g_1=g(i,j)$, $g_2=g(i+1,j)$, define $\alpha=(i,j)\in\mathcal{I}_{\mathrm{mono}}$ and
    \[
      \ell_{(i,j)}^\top := (e^{(g_2)}-e^{(g_1)})^\top,\qquad r_{(i,j)} := 0.
    \]
  \item \emph{(Convexity)}    
    For each maturity $j$ and $i=2,\dots,n_m-1$, with
    $g_- = g(i-1,j)$, $g_0 = g(i,j)$, $g_+ = g(i+1,j)$, define $\alpha=(i,j)\in\mathcal{I}_{\mathrm{conv}}$ and
    \[
      \ell_{(i,j)}^\top := (-e^{(g_+)} + 2 e^{(g_0)} - e^{(g_-)})^\top,\qquad r_{(i,j)} := 0.
    \]
  \item \emph{(Calendar)}    
    Let $A_{\tau|K}\in\mathbb{R}^{G\times G}$ be the fixed-strike calendar operator.
    For each $g\in\{1,\dots,G\}$ define $\alpha=g\in\mathcal{I}_{\mathrm{cal}}$ and
    \[
      \ell_g^\top := - (A_{\tau|K})_{g,\cdot},\qquad r_g := 0.
    \]
\end{itemize}
Set $\mathcal{I}:=\mathcal{I}_{\mathrm{bnd}}\cup\mathcal{I}_{\mathrm{mono}}
\cup\mathcal{I}_{\mathrm{conv}}\cup\mathcal{I}_{\mathrm{cal}}$.
\end{definition}

This definition can be seen as follows:

\paragraph{Bounds.}
We first explain how to encode the pointwise bounds
\[
0 \le u_{i,j} \le F_{i,j}, \qquad (i,j)\in\mathcal{G},
\]
as linear inequalities of the form $\ell_\alpha^\top u \le r_\alpha$.

Fix a node $(i,j)\in\mathcal{G}$ and let $g = g(i,j)$ be its index in the flattened
nodal vector $u\in\mathbb{R}^G$, so that $u_g = u_{i,j}$. For this node we introduce
\emph{two} constraint indices:
\[
\alpha = (i,j,\mathrm{lo}),\qquad \alpha' = (i,j,\mathrm{up}).
\]

\emph{Lower bound.} For $\alpha = (i,j,\mathrm{lo})$ we define
\[
\ell_{(i,j,\mathrm{lo})}^\top := - (e^{(g)})^\top,\qquad r_{(i,j,\mathrm{lo})} := 0,
\]
where $e^{(g)}\in\mathbb{R}^G$ is the $g$-th standard basis vector. Then
\[
\ell_{(i,j,\mathrm{lo})}^\top u
= - e^{(g)\top} u
= - u_g
= - u_{i,j},
\]
so the inequality $\ell_{(i,j,\mathrm{lo})}^\top u \le r_{(i,j,\mathrm{lo})}$ reads
\[
- u_{i,j} \le 0
\quad\Longleftrightarrow\quad
u_{i,j} \ge 0.
\]

\emph{Upper bound.} For $\alpha' = (i,j,\mathrm{up})$ we set
\[
\ell_{(i,j,\mathrm{up})}^\top := (e^{(g)})^\top,\qquad r_{(i,j,\mathrm{up})} := F_{i,j}.
\]
Then
\[
\ell_{(i,j,\mathrm{up})}^\top u
= e^{(g)\top} u
= u_g
= u_{i,j},
\]
so $\ell_{(i,j,\mathrm{up})}^\top u \le r_{(i,j,\mathrm{up})}$ is exactly
\[
u_{i,j} \le F_{i,j}.
\]

Thus for each node $(i,j)$ we obtain the two bounds $0\le u_{i,j}\le F_{i,j}$ as the
pair of inequalities
\[
\ell_{(i,j,\mathrm{lo})}^\top u \le r_{(i,j,\mathrm{lo})},\qquad
\ell_{(i,j,\mathrm{up})}^\top u \le r_{(i,j,\mathrm{up})}.
\]
All such indices $(i,j,\mathrm{lo})$ and $(i,j,\mathrm{up})$ are collected in
$\mathcal{I}_{\mathrm{bnd}}$.

\medskip
\paragraph{Monotonicity.}
We now encode discrete strike monotonicity, namely
\[
u_{i+1,j} - u_{i,j} \le 0,\qquad i=1,\dots,n_m-1,\quad j=1,\dots,n_\tau.
\]

Fix a maturity $j$ and $i\in\{1,\dots,n_m-1\}$. Let
\[
g_1 := g(i,j),\qquad g_2 := g(i+1,j)
\]
be the indices of the adjacent nodes $(i,j)$ and $(i+1,j)$ in the flattened
vector $u$. We introduce a single index $\alpha=(i,j)\in\mathcal{I}_{\mathrm{mono}}$
and define
\[
\ell_{(i,j)}^\top := (e^{(g_2)} - e^{(g_1)})^\top,\qquad r_{(i,j)} := 0.
\]
Then
\[
\ell_{(i,j)}^\top u
= e^{(g_2)\top} u - e^{(g_1)\top} u
= u_{g_2} - u_{g_1}
= u_{i+1,j} - u_{i,j}.
\]
Thus the inequality $\ell_{(i,j)}^\top u \le r_{(i,j)}$ is precisely
\[
u_{i+1,j} - u_{i,j} \le 0,
\]
the desired monotonicity condition. Each adjacent pair of strikes at fixed $j$
contributes one such index $(i,j)$ to $\mathcal{I}_{\mathrm{mono}}$.

\medskip
\paragraph{Convexity.}
Discrete strike convexity requires
\[
u_{i+1,j} - 2 u_{i,j} + u_{i-1,j} \ge 0,\qquad i=2,\dots,n_m-1,\quad j=1,\dots,n_\tau.
\]
Equivalently,
\[
- u_{i+1,j} + 2 u_{i,j} - u_{i-1,j} \le 0.
\]

Fix a maturity $j$ and an interior strike index $i\in\{2,\dots,n_m-1\}$. Let
\[
g_- := g(i-1,j),\qquad g_0 := g(i,j),\qquad g_+ := g(i+1,j).
\]
We introduce an index $\alpha=(i,j)\in\mathcal{I}_{\mathrm{conv}}$ and set
\[
\ell_{(i,j)}^\top := (-e^{(g_+)} + 2 e^{(g_0)} - e^{(g_-)})^\top,\qquad
r_{(i,j)} := 0.
\]
Then
\[
\ell_{(i,j)}^\top u
= -u_{g_+} + 2 u_{g_0} - u_{g_-}
= -u_{i+1,j} + 2 u_{i,j} - u_{i-1,j},
\]
so the inequality $\ell_{(i,j)}^\top u \le r_{(i,j)}$ is exactly
\[
- u_{i+1,j} + 2 u_{i,j} - u_{i-1,j} \le 0
\quad\Longleftrightarrow\quad
u_{i+1,j} - 2 u_{i,j} + u_{i-1,j} \ge 0.
\]
All such indices $(i,j)$ form the convexity index set $\mathcal{I}_{\mathrm{conv}}$.

\medskip
\paragraph{Calendar (fixed strike).}
Finally, we encode the discrete calendar condition at fixed strike via the operator
$A_{\tau|K}\in\mathbb{R}^{G\times G}$. By construction, $(A_{\tau|K}u)_g$ is the
discrete approximation of $(\partial_\tau C_f)(K,\tau)|_K$ at the grid node
indexed by $g$. The condition
\[
(A_{\tau|K}u)_g \ge 0,\qquad g=1,\dots,G,
\]
can be written as
\[
-(A_{\tau|K}u)_g \le 0.
\]

For each grid index $g\in\{1,\dots,G\}$ we take $\alpha = g\in\mathcal{I}_{\mathrm{cal}}$
and define
\[
\ell_g^\top := - (A_{\tau|K})_{g,\cdot},\qquad r_g := 0,
\]
i.e.\ $\ell_g^\top$ is the negative of the $g$-th row of $A_{\tau|K}$. Then
\[
\ell_g^\top u = - (A_{\tau|K}u)_g,
\]
so the inequality $\ell_g^\top u \le r_g$ is precisely
\[
-(A_{\tau|K}u)_g \le 0
\quad\Longleftrightarrow\quad
(A_{\tau|K}u)_g \ge 0.
\]
Thus each row $g$ of $A_{\tau|K}$ generates one calendar inequality, and all these
indices $g$ belong to $\mathcal{I}_{\mathrm{cal}}$.

\begin{definition}[Primitive discrete static no-arbitrage on $\mathcal{G}$]
\label{def:primitive-noarb}
A nodal surface 
\[u=(u_{i,j})_{(i,j)\in\mathcal{G}}\in\mathbb{R}^G\]
is said to satisfy \emph{primitive discrete static no-arbitrage} on $\mathcal{G}$ if:

\begin{itemize}
  \item \emph{(BND)} (Bounds) For all $(i,j)\in\mathcal{G}$,
  \[
    0 \le u_{i,j} \le F_{i,j}.
  \]

  \item \emph{(MONO)} (Strike monotonicity) For each maturity $j$ and each
        $i=1,\dots,n_m-1$,
  \[
    u_{i+1,j} - u_{i,j} \le 0.
  \]

  \item \emph{(CONV)} (Strike convexity) For each maturity $j$ and
        $i=2,\dots,n_m-1$,
  \[
    u_{i+1,j} - 2u_{i,j} + u_{i-1,j} \ge 0.
  \]

  \item \emph{(CAL)} (Calendar at fixed strike) Let $A_{\tau|K}\in\mathbb{R}^{G\times G}$
        be the fixed-strike calendar operator from
        Section~\ref{sec:calendar-operator}. Then
  \[
    (A_{\tau|K} u)_g \ge 0\qquad\text{for all grid indices }g=1,\dots,G.
  \]
\end{itemize}
\end{definition}

\begin{prop}[Equivalence of operator encoding and primitive no-arbitrage]
\label{prop:noarb-encoding-equivalence}
Let $(\ell_\alpha,r_\alpha)_{\alpha\in\mathcal{I}}$ be defined as in
Definition~\ref{def:global-noarb-operators}, with index sets
$\mathcal{I}_{\mathrm{bnd}}$, $\mathcal{I}_{\mathrm{mono}}$,
$\mathcal{I}_{\mathrm{conv}}$, $\mathcal{I}_{\mathrm{cal}}$ corresponding to
bounds, strike monotonicity, strike convexity, and calendar constraints respectively,
and $\mathcal{I}
= \mathcal{I}_{\mathrm{bnd}}\cup\mathcal{I}_{\mathrm{mono}}
\cup\mathcal{I}_{\mathrm{conv}}\cup\mathcal{I}_{\mathrm{cal}}$.

For $u\in\mathbb{R}^G$, the following are equivalent:
\begin{enumerate}[label=(\roman*)]
  \item $u$ satisfies all linear inequalities
  \[
    \ell_\alpha^\top u \le r_\alpha,\qquad \forall \alpha\in\mathcal{I}.
  \]
  \item $u$ satisfies primitive discrete static no-arbitrage on $\mathcal{G}$ in the
        sense of Definition~\ref{def:primitive-noarb}.
\end{enumerate}
In particular, the global static no-arbitrage feasible set
\[
\mathcal{C}_{\mathrm{glob}}
=
\big\{u\in\mathbb{R}^G : \ell_\alpha^\top u \le r_\alpha\ \forall\alpha\in\mathcal{I}\big\}
\]
coincides with the set of all nodal surfaces that satisfy the primitive discrete
no-arbitrage conditions (BND), (MONO), (CONV) and (CAL).
\end{prop}

\begin{proof}
We prove $(i)\Leftrightarrow (ii)$ by decomposing $\mathcal{I}$ into its four
subsets.

\medskip
\noindent\emph{(ii) $\Rightarrow$ (i).} Suppose $u$ satisfies the primitive
conditions (BND), (MONO), (CONV), (CAL).

\smallskip
\emph{Bounds.} Fix $(i,j)\in\mathcal{G}$ and let $g=g(i,j)$ be its index. By
Definition~\ref{def:global-noarb-operators}, the two bound indices
$\alpha=(i,j,\mathrm{lo})$, $\alpha'=(i,j,\mathrm{up})$ satisfy
\[
\ell_{(i,j,\mathrm{lo})}^\top = - (e^{(g)})^\top,\quad
r_{(i,j,\mathrm{lo})} = 0,\qquad
\ell_{(i,j,\mathrm{up})}^\top = (e^{(g)})^\top,\quad
r_{(i,j,\mathrm{up})} = F_{i,j},
\]
where $e^{(g)}$ is the $g$-th standard basis vector in $\mathbb{R}^G$.
Then
\[
\ell_{(i,j,\mathrm{lo})}^\top u = -u_{i,j},\qquad
\ell_{(i,j,\mathrm{up})}^\top u = u_{i,j}.
\]
The bound condition $0\le u_{i,j}\le F_{i,j}$ implies
\[
-u_{i,j} \le 0 = r_{(i,j,\mathrm{lo})},\qquad
u_{i,j} \le F_{i,j} = r_{(i,j,\mathrm{up})},
\]
hence $\ell_\alpha^\top u \le r_\alpha$ for all
$\alpha\in\mathcal{I}_{\mathrm{bnd}}$.

\smallskip
\emph{Strike monotonicity.} Fix a maturity $j$ and $i=1,\dots,n_m-1$, and let
$g_1=g(i,j)$, $g_2=g(i+1,j)$. By
Definition~\ref{def:global-noarb-operators}, for $\alpha=(i,j)\in\mathcal{I}_{\mathrm{mono}}$
we have
\[
\ell_{(i,j)}^\top := (e^{(g_2)} - e^{(g_1)})^\top,\qquad r_{(i,j)} := 0,
\]
so that
\[
\ell_{(i,j)}^\top u = u_{i+1,j} - u_{i,j}.
\]
The primitive monotonicity condition (MONO) states that $u_{i+1,j} - u_{i,j}\le 0$,
hence $\ell_{(i,j)}^\top u \le r_{(i,j)}$ for all $(i,j)$, i.e.\ for all
$\alpha\in\mathcal{I}_{\mathrm{mono}}$.

\smallskip
\emph{Strike convexity.} Fix a maturity $j$ and $i=2,\dots,n_m-1$, and let
$g_- = g(i-1,j)$, $g_0=g(i,j)$, $g_+=g(i+1,j)$. For $\alpha=(i,j)\in\mathcal{I}_{\mathrm{conv}}$,
Definition~\ref{def:global-noarb-operators} gives
\[
\ell_{(i,j)}^\top := (- e^{(g_+)} + 2e^{(g_0)} - e^{(g_-)})^\top,\qquad
r_{(i,j)} := 0.
\]
Then
\[
\ell_{(i,j)}^\top u
= -u_{i+1,j} + 2u_{i,j} - u_{i-1,j}.
\]
The primitive convexity condition (CONV) is
$u_{i+1,j}-2u_{i,j}+u_{i-1,j}\ge 0$, or equivalently
$-u_{i+1,j}+2u_{i,j}-u_{i-1,j}\le 0$, hence $\ell_{(i,j)}^\top u \le r_{(i,j)}$
for all $(i,j)$, i.e.\ all $\alpha\in\mathcal{I}_{\mathrm{conv}}$.

\smallskip
\emph{Calendar.} Let $A_{\tau|K}\in\mathbb{R}^{G\times G}$ be the fixed-strike
calendar operator. For each grid index $g\in\{1,\dots,G\}$, the calendar index
$\alpha=g\in\mathcal{I}_{\mathrm{cal}}$ is defined by
\[
\ell_g^\top := - (A_{\tau|K})_{g,\cdot},\qquad r_g := 0,
\]
so that
\[
\ell_g^\top u = - (A_{\tau|K} u)_g.
\]
The primitive calendar condition (CAL) states $(A_{\tau|K}u)_g \ge 0$ for all
$g$, which is equivalent to $- (A_{\tau|K}u)_g \le 0 = r_g$, hence
$\ell_g^\top u \le r_g$ for all $\alpha\in\mathcal{I}_{\mathrm{cal}}$.

Combining the four families, we see that (ii) implies $\ell_\alpha^\top u \le r_\alpha$
for all $\alpha\in\mathcal{I}$, i.e.\ (i) holds.

\medskip
\noindent\emph{(i) $\Rightarrow$ (ii).} Conversely, suppose $u$ satisfies
$\ell_\alpha^\top u \le r_\alpha$ for all $\alpha\in\mathcal{I}$. We show that the
primitive conditions (BND), (MONO), (CONV), (CAL) hold.

\smallskip
\emph{Bounds.} Fix $(i,j)\in\mathcal{G}$ and $g=g(i,j)$. For $\alpha=(i,j,\mathrm{lo})$
and $\alpha'=(i,j,\mathrm{up})$ we have, by definition,
\[
\ell_{(i,j,\mathrm{lo})}^\top u = -u_{i,j} \le r_{(i,j,\mathrm{lo})} = 0,\qquad
\ell_{(i,j,\mathrm{up})}^\top u = u_{i,j} \le r_{(i,j,\mathrm{up})} = F_{i,j}.
\]
Thus $-u_{i,j}\le 0$ and $u_{i,j}\le F_{i,j}$, i.e.\ $0\le u_{i,j}\le F_{i,j}$
for all $(i,j)$, which is (BND).

\smallskip
\emph{Strike monotonicity.} For each $j$ and $i=1,\dots,n_m-1$, the inequality
for $\alpha=(i,j)\in\mathcal{I}_{\mathrm{mono}}$ reads
\[
\ell_{(i,j)}^\top u
= u_{i+1,j} - u_{i,j} \le 0,
\]
which is exactly the discrete monotonicity condition $u_{i+1,j}-u_{i,j}\le 0$
for all such $(i,j)$; this is (MONO).

\smallskip
\emph{Strike convexity.} For each $j$ and $i=2,\dots,n_m-1$, the inequality
for $\alpha=(i,j)\in\mathcal{I}_{\mathrm{conv}}$ is
\[
\ell_{(i,j)}^\top u
= -u_{i+1,j} + 2u_{i,j} - u_{i-1,j} \le 0.
\]
Rearranging gives $u_{i+1,j}-2u_{i,j}+u_{i-1,j}\ge 0$, which is (CONV).

\smallskip
\emph{Calendar.} Finally, for each grid index $g\in\{1,\dots,G\}$ we have
\[
\ell_g^\top u = -(A_{\tau|K} u)_g \le 0,
\]
or equivalently $(A_{\tau|K} u)_g \ge 0$. Thus (CAL) holds at every grid index.

\smallskip
Therefore all four primitive conditions (BND), (MONO), (CONV), (CAL) hold, and
$u$ satisfies primitive discrete static no-arbitrage. This proves (ii).

\medskip
We have shown $(ii)\Rightarrow(i)$ and $(i)\Rightarrow(ii)$, so the two statements
are equivalent. The characterisation of $\mathcal{C}_{\mathrm{glob}}$ as the set
of all primitively no-arbitrage nodal surfaces follows immediately from the
definition of $\mathcal{C}_{\mathrm{glob}}$.
\end{proof}

\subsection{Patch-level feasible set and its geometry}

We now enforce global static no-arbitrage on the \emph{assembled} surface $u(u_I)$
obtained from a patch interior vector $u_I$.

\begin{definition}[No-arbitrage feasible set on a patch]
\label{def:CuOmega}
The \emph{patch-level no-arbitrage feasible set} is
\[
\mathcal{C}_u(\Omega)
:=
\big\{
u_I\in\mathbb{R}^{N_\Omega} : u(u_I)\in\mathcal{C}_{\mathrm{glob}}
\big\}.
\]
Equivalently, using \eqref{eq:Cglobal} and the assembly map
\eqref{eq:assembly-affine},
\[
\mathcal{C}_u(\Omega)
=
\Big\{
u_I\in\mathbb{R}^{N_\Omega} :
\ell_\alpha^\top u(u_I) \le r_\alpha,\quad \forall \alpha\in\mathcal{I}
\Big\}.
\]
\end{definition}

We now characterise $\mathcal{C}_u(\Omega)$ as a polyhedron and prove its
basic geometric properties.

\begin{prop}[Polyhedral structure of $\mathcal{C}_u(\Omega)$]
\label{prop:Cu-polyhedron}
The set $\mathcal{C}_u(\Omega)\subset\mathbb{R}^{N_\Omega}$ can be written as the
finite intersection of affine half-spaces
\[
\mathcal{C}_u(\Omega)
=
\bigcap_{\alpha\in\mathcal{I}} H_\alpha,
\]
where each $H_\alpha$ is of the form
\[
H_\alpha
:=
\{u_I\in\mathbb{R}^{N_\Omega} : a_\alpha^\top u_I \le b_\alpha\}
\]
for some $a_\alpha\in\mathbb{R}^{N_\Omega}$ and $b_\alpha\in\mathbb{R}$. In particular,
$\mathcal{C}_u(\Omega)$ is a (possibly empty) closed convex polyhedron in
$\mathbb{R}^{N_\Omega}$.
\end{prop}

\begin{proof}
Fix $\alpha\in\mathcal{I}$. For $u_I\in\mathbb{R}^{N_\Omega}$, the corresponding
assembled surface is $u(u_I) = P_\Omega u_I + u^{0,\mathrm{off}}$ by
\eqref{eq:assembly-affine}. The $\alpha$-th global no-arbitrage inequality reads
\[
\ell_\alpha^\top u(u_I) \le r_\alpha.
\]
Substituting the affine expression for $u(u_I)$, we obtain
\[
\ell_\alpha^\top(P_\Omega u_I + u^{0,\mathrm{off}}) \le r_\alpha,
\]
which can be rearranged as
\[
(\ell_\alpha^\top P_\Omega)\, u_I \le r_\alpha - \ell_\alpha^\top u^{0,\mathrm{off}}.
\]
Now, define
\[
a_\alpha := P_\Omega^\top \ell_\alpha \in\mathbb{R}^{N_\Omega},\qquad
b_\alpha := r_\alpha - \ell_\alpha^\top u^{0,\mathrm{off}} \in\mathbb{R}.
\]
Note that $\ell_\alpha^\top P_\Omega = a_\alpha^\top$ by construction. Then the
$\alpha$-th constraint is equivalent to
\[
a_\alpha^\top u_I \le b_\alpha.
\]
Therefore, the set of $u_I$ satisfying the $\alpha$-th global no-arbitrage inequality
is the half-space
\[
H_\alpha := \{u_I\in\mathbb{R}^{N_\Omega} : a_\alpha^\top u_I \le b_\alpha\}.
\]

Because this construction holds for every $\alpha\in\mathcal{I}$ (namely; to be patch-feasible $u_I$ has to satisfy every constraint, so belong to every $H_\alpha$), we have
\[
\mathcal{C}_u(\Omega)
=
\{u_I\in\mathbb{R}^{N_\Omega} : a_\alpha^\top u_I \le b_\alpha\ \forall\alpha\in\mathcal{I}\}
=
\bigcap_{\alpha\in\mathcal{I}} H_\alpha.
\]

Recall that a half-space in $\R^n$ is any set that can be written as $\{ x\in  \R^n: a^\top x\leq b\}$ or $\{ x\in  \R^n: a^\top x\geq b\}$ for some fixed nonzero vector $a$ and scalar $b$. Therefore; by definition, each $H_\alpha$ is a half-space in $\mathbb{R}^{N_\Omega}$.

A set $C$ is convex if for all $x,y\in C$ and $\lambda \in [0,1]$,
\[
\lambda x + (1-\lambda) y \in C.
\]
Now, take $H_\alpha := \{u_I\in\mathbb{R}^{N_\Omega} : a_\alpha^\top u_I \le b_\alpha\}$ and let $u_I^{(1)}, u_I^{(2)}\in H_\alpha$. We therefore have by definition of $H_\alpha$ that $a_\alpha^\top u_I^{(1)}\leq b_\alpha$ and $a_\alpha^\top u_I^{(2)}\leq b_\alpha$. Let $\lambda \in [0,1]$ and consider $u_I^{(\lambda)}:= \lambda u_I^{(1)} + (1-\lambda) u_I^{(2)}$. Computing:
\[
a_\alpha^\top u_I^{(\lambda)}=a_\alpha^\top(\lambda u_I^{(1)} + (1-\lambda) u_I^{(2)}) = a_\alpha^\top\lambda u_I^{(1)} + a_\alpha^\top (1-\lambda) u_I^{(2)}.
\]
Since $a_\alpha^\top u_I^{(1)}\leq b_\alpha$ and $a_\alpha^\top u_I^{(2)}\leq b_\alpha$, we have
\[
a_\alpha^\top\lambda u_I^{(1)} + a_\alpha^\top (1-\lambda) u_I^{(2)}\leq \lambda b_\alpha +(1-\lambda)b_\alpha = b_\alpha.
\]
Therefore, $a_\alpha^\top u_I^{(\lambda)}\leq b_\alpha$. By definition, $u_I^{(\lambda)}\in H_\alpha$; which is the exact definition of convexity and thus $H_\alpha$ is convex.

We can see that each $H_\alpha$ is closed is via sequences.
Let $(u_I^{(n)})_{n\in\mathbb{N}}$ be a sequence in $H_\alpha$ converging to
some $u_I^\star\in\mathbb{R}^{N_\Omega}$. By definition of $H_\alpha$ we have
\[
a_\alpha^\top u_I^{(n)} \le b_\alpha,\qquad \forall n\in\mathbb{N}.
\]
The map $u_I \mapsto a_\alpha^\top u_I$ is linear and hence continuous on
$\mathbb{R}^{N_\Omega}$, so passing to the limit $n\to\infty$ yields
\[
a_\alpha^\top u_I^\star
= \lim_{n\to\infty} a_\alpha^\top u_I^{(n)}
\le b_\alpha.
\]
Thus $u_I^\star$ also satisfies $a_\alpha^\top u_I^\star \le b_\alpha$, i.e.\
$u_I^\star\in H_\alpha$. Therefore $H_\alpha$ contains the limit of every
convergent sequence of its elements, and is closed.

Hence, each $H_\alpha$ is a closed half-space in $\mathbb{R}^{N_\Omega}$ and is convex.
The intersection of any family of convex sets is convex, and the intersection of
any family of closed sets is closed. Since $\mathcal{I}$ is finite, this intersection
defines a closed convex polyhedron. The polyhedron may be empty or nonempty,
depending on the data; we analyse feasibility separately.
\end{proof}

\begin{definition}[Patch feasibility]
\label{def:patch-feasible}
A patch $\Omega\subset\mathcal{G}$ is said to be \emph{feasible} if
$\mathcal{C}_u(\Omega)\neq\emptyset$.
\end{definition}

The assumption that all patches used in the post-fit are feasible is mild and
consistent with the construction of $\Omega$ from the baseline surface. Indeed,
if the baseline nodal surface $u^0$ is globally statically no-arbitrage on
$\mathcal{G}$, then $u^0\in\mathcal{C}_{\mathrm{glob}}$. If, in addition, the
patch decomposition is chosen so that every static no-arbitrage stencil that
intersects the interior of $\Omega$ is either fully contained in $\Omega$ or
fully contained in $\mathcal{G}\setminus\Omega$, then the restriction of $u^0$
to $\Omega$ defines an interior vector $u_I^{0,\Omega}$ satisfying all
constraints $a_\alpha^\top u_I \le b_\alpha$ and hence
$u_I^{0,\Omega}\in\mathcal{C}_u(\Omega)$. In practice, feasibility can be
verified numerically by solving a simple linear feasibility problem for each
patch; in the theoretical development of this chapter we take it as an explicit
assumption that patches are chosen so that $\mathcal{C}_u(\Omega)$ is nonempty.

\section{Hamiltonian energy on the fog}

We now endow the discrete fog $\pi$ on the lattice $\mathcal{L}_\Omega$ with a
quadratic Hamiltonian energy. The Hamiltonian has two components:
a \emph{kinetic} (Dirichlet) term that penalises roughness of $\pi$ across
neighboring lattice sites in $(m,\tau,u)$, and a \emph{potential} term that
penalises fog mass lying far outside local bid--ask tubes or basic price ranges.

Throughout this section we fix a patch $\Omega\subset\mathcal{G}$ with
$N_\Omega = |\Omega|$ and a set of price levels $\{u_k\}_{k=1}^{n_u}$, and we
work on the lattice
\[
\mathcal{L}_\Omega := \{(i,j,k) : (i,j)\in\Omega,\ k=1,\dots,n_u\},
\]
as in Section~\ref{sec:3D-fog}.

\subsection{Discrete 3D graph and difference operators}

We begin by making the discrete graph structure of $\mathcal{L}_\Omega$ explicit
and constructing the associated difference operators.

\begin{definition}[Adjacency graph on $\mathcal{L}_\Omega$]
Let $\mathcal{L}_\Omega$ be the 3D lattice of nodes $(i,j,k)$ with
$(i,j)\in\Omega$ and $k\in\{1,\dots,n_u\}$. We define three families of
undirected edges:
\begin{itemize}
  \item \emph{$m$-edges} $E_m$: for every $(i,j,k)\in\mathcal{L}_\Omega$ such that
        $(i+1,j)\in\Omega$, we introduce an edge between $(i,j,k)$ and
        $(i+1,j,k)$;
  \item \emph{$\tau$-edges} $E_\tau$: for every $(i,j,k)\in\mathcal{L}_\Omega$
        such that $(i,j+1)\in\Omega$, we introduce an edge between $(i,j,k)$
        and $(i,j+1,k)$;
  \item \emph{$u$-edges} $E_u$: for every $(i,j,k)\in\mathcal{L}_\Omega$ such that
        $k+1\le n_u$, we introduce an edge between $(i,j,k)$ and $(i,j,k+1)$.
\end{itemize}
The full edge set is
\[
E := E_m \cup E_\tau \cup E_u.
\]
\end{definition}

Thus $E_m$ connects nearest neighbours in the $m$-direction (at fixed
$(\tau,u)$), $E_\tau$ connects nearest neighbours in the $\tau$-direction (at
fixed $(m,u)$), and $E_u$ connects nearest neighbours in the $u$-direction (at
fixed $(m,\tau)$). Since edges are only drawn between nodes that both belong to
$\mathcal{L}_\Omega$, this corresponds to homogeneous Neumann boundary
conditions at the boundary of the patch.

For each edge family we now define a discrete gradient operator as the signed
incidence matrix of the corresponding graph.

\begin{definition}[Discrete gradients along $m,\tau,u$]
Fix an arbitrary but fixed orientation of each edge in $E_m,E_\tau,E_u$:
for each edge $e=\{p,q\}\in E_m$, choose an ordering $(p\to q)$ (e.g.\ increasing
in $i$); similarly for $E_\tau$ (increasing in $j$) and $E_u$ (increasing in $k$).

Let $N_L := |\mathcal{L}_\Omega| = N_\Omega n_u$ be the number of lattice nodes,
and enumerate $\mathcal{L}_\Omega$ as
\[
\mathcal{L}_\Omega = \{\xi_\ell\}_{\ell=1}^{N_L},\qquad \xi_\ell=(i_\ell,j_\ell,k_\ell).
\]
We identify fog configurations $\pi$ with vectors in $\mathbb{R}^{N_L}$ via
$\pi_\ell := \pi_{i_\ell,j_\ell,k_\ell}$.

\begin{itemize}
  \item The \emph{$m$-gradient} $D_m:\mathbb{R}^{N_L}\to\mathbb{R}^{|E_m|}$ is
        defined as follows: index $E_m = \{e_r\}_{r=1}^{|E_m|}$ and for each
        edge $e_r=(p\to q)$, set
        \[
        (D_m \pi)_r := \pi_q - \pi_p.
        \]
        In matrix form, $D_m$ is the $|E_m|\times N_L$ matrix whose $r$-th row
        has entry $-1$ in the column corresponding to node $p$, entry $+1$ in
        the column corresponding to node $q$, and zeros elsewhere.
  \item The \emph{$\tau$-gradient} $D_\tau:\mathbb{R}^{N_L}\to\mathbb{R}^{|E_\tau|}$
        is defined analogously, with one row per edge in $E_\tau$, oriented in
        increasing $j$.
  \item The \emph{$u$-gradient} $D_u:\mathbb{R}^{N_L}\to\mathbb{R}^{|E_u|}$ is
        defined analogously, with one row per edge in $E_u$, oriented in
        increasing $k$.
\end{itemize}
\end{definition}

Thus $D_m\pi$ collects all forward differences of $\pi$ along $m$-edges, and similarly
for $D_\tau$ and $D_u$. We now weight these differences by nonnegative edge
weights.

\begin{definition}[Edge-weight matrices]
Let $w^m_r\ge 0$ be a nonnegative weight associated with the $r$-th $m$-edge in
$E_m$, and define the diagonal matrix
\[
W_m := \mathrm{diag}(w^m_1,\dots,w^m_{|E_m|})\in\mathbb{R}^{|E_m|\times |E_m|}.
\]
Similarly, let $w^\tau_r\ge 0$ and $w^u_r\ge 0$ be edge weights on $E_\tau$ and
$E_u$, and define diagonal matrices
$W_\tau\in\mathbb{R}^{|E_\tau|\times |E_\tau|}$ and
$W_u\in\mathbb{R}^{|E_u|\times |E_u|}$ with these weights on the diagonal.
\end{definition}

Typical choices include $w^m_r,w^\tau_r,w^u_r\equiv 1$ (unweighted differences), or
weights that depend on grid spacings and/or the 2D marginal $n_{i,j}$; the only
property needed here is nonnegativity.

For any such diagonal matrix $W\succeq 0$ and vector $x$, we write
\[
\|x\|_W^2 := x^\top W x
\]
for the weighted squared norm.

\subsection{Kinetic energy and graph Laplacian}

We now define the Dirichlet kinetic energy of the fog in terms of these
discrete gradients.

\begin{definition}[Kinetic energy of the fog]
\label{def:kinetic-energy}
Let $\kappa_m,\kappa_\tau,\kappa_u\ge 0$ be fixed nonnegative parameters. The
\emph{kinetic energy} (Dirichlet energy) of a fog configuration
$\pi\in\mathbb{R}^{N_L}$ is
\[
\mathcal{E}_{\mathrm{kin}}(\pi)
:=
\frac{\kappa_m}{2}\,\|D_m\pi\|_{W_m}^2
+ \frac{\kappa_\tau}{2}\,\|D_\tau\pi\|_{W_\tau}^2
+ \frac{\kappa_u}{2}\,\|D_u\pi\|_{W_u}^2.
\]
Explicitly,
\[
\mathcal{E}_{\mathrm{kin}}(\pi)
=
\frac{\kappa_m}{2}\,(D_m\pi)^\top W_m (D_m\pi)
+ \frac{\kappa_\tau}{2}\,(D_\tau\pi)^\top W_\tau (D_\tau\pi)
+ \frac{\kappa_u}{2}\,(D_u\pi)^\top W_u (D_u\pi).
\]
\end{definition}

The Dirichlet energy is a nonnegative quadratic form in $\pi$, and it can be
written in the standard graph-Laplacian form.

\begin{prop}[Matrix form and positive semidefiniteness of $L_\pi$]
\label{prop:Lpi-psd}
Define
\[
L_m := D_m^\top W_m D_m,\qquad
L_\tau := D_\tau^\top W_\tau D_\tau,\qquad
L_u := D_u^\top W_u D_u
\]
and
\[
L_\pi := \kappa_m L_m + \kappa_\tau L_\tau + \kappa_u L_u.
\]
Then:
\begin{enumerate}[label=(\roman*)]
  \item $L_m,L_\tau,L_u$ and $L_\pi$ are symmetric positive semidefinite matrices
        in $\mathbb{R}^{N_L\times N_L}$;
  \item for all $\pi\in\mathbb{R}^{N_L}$,
  \[
  \mathcal{E}_{\mathrm{kin}}(\pi)
  = \frac{1}{2}\,\pi^\top L_\pi \pi.
  \]
\end{enumerate}
\end{prop}

\begin{proof}
(i) For any matrix $D$ and diagonal matrix $W\succeq 0$, the matrix
$L := D^\top W D$ is symmetric and positive semidefinite:
\[
L^\top = (D^\top W D)^\top = D^\top W^\top D = D^\top W D = L,
\]
and for any $x$,
\[
x^\top L x = x^\top D^\top W D x = (D x)^\top W (D x) \ge 0,
\]
since $W$ has nonnegative diagonal entries. Applying this with
$(D,W)=(D_m,W_m)$, $(D_\tau,W_\tau)$ and $(D_u,W_u)$ yields the claimed properties
for $L_m,L_\tau,L_u$. A nonnegative linear combination of symmetric positive
semidefinite matrices is again symmetric positive semidefinite, so $L_\pi$ is
symmetric positive semidefinite.

(ii) By definition of $L_m,L_\tau,L_u$,
\[
(D_m\pi)^\top W_m (D_m\pi) = \pi^\top L_m \pi,\quad
(D_\tau\pi)^\top W_\tau (D_\tau\pi) = \pi^\top L_\tau \pi,
\]
\[
(D_u\pi)^\top W_u (D_u\pi) = \pi^\top L_u \pi.\]
Therefore
\[
\mathcal{E}_{\mathrm{kin}}(\pi)
=
\frac{1}{2}
\bigl(
\kappa_m\,\pi^\top L_m \pi
+ \kappa_\tau\,\pi^\top L_\tau \pi
+ \kappa_u\,\pi^\top L_u \pi
\bigr)
=
\frac{1}{2}\,\pi^\top L_\pi \pi,
\]
which proves the claim.
\end{proof}

Consequently, $\mathcal{E}_{\mathrm{kin}}$ is a convex quadratic functional on
$\mathbb{R}^{N_L}$, with flat directions corresponding to fog configurations that
are constant along connected components of the underlying graph (if all
$\kappa_m,\kappa_\tau,\kappa_u>0$ and the graph is connected with Neumann
boundary, the constant vector lies in the kernel of $L_\pi$).

\subsection{Potential energy and band-aware penalisation}

We now introduce a nonnegative potential field $V$ on $\mathcal{L}_\Omega$ that
penalises fog mass far from the local bid-ask tubes and from basic price bounds.

\begin{definition}[Band and range potential]
\label{def:band-range-potential}
For each quote $q$ lying on the patch, let $(m_q,\tau_q)$ be its location and
$[b_q,a_q]$ its cleaned forward-discounted bid--ask band, and choose a
representative grid node $(i_q,j_q)\in\Omega$ (e.g.\ the nearest neighbour in
$\Omega$).

Fix parameters $\alpha_{\mathrm{band}}\ge 0$ and $\alpha_{\mathrm{range}}\ge 0$.
Define the \emph{band potential} as
\[
V^{\mathrm{band}}_{i,j,k}
:=
\begin{cases}
\alpha_{\mathrm{band}}\,
\operatorname{dist}(u_k,[b_q,a_q])^2, & \text{if }(i,j)=(i_q,j_q)\text{ for some quote }q,\\[3pt]
0, & \text{otherwise},
\end{cases}
\]
for all $(i,j,k)\in\mathcal{L}_\Omega$, where
$\operatorname{dist}(u,[b,a]) := \max\{b-u,0,u-a\}$ is the Euclidean distance
from $u$ to the interval $[b,a]$.

Define the \emph{range potential} by
\[
V^{\mathrm{range}}_{i,j,k}
:=
\alpha_{\mathrm{range}}\Bigl(
\mathbf{1}_{\{u_k<0\}} + \mathbf{1}_{\{u_k>F_{i,j}\}}
\Bigr),
\]
where $F_{i,j}$ is the forward at node $(i,j)$ and $\mathbf{1}_A$ is the
indicator of the event $A$.

Finally, set
\[
V_{i,j,k} := V^{\mathrm{band}}_{i,j,k} + V^{\mathrm{range}}_{i,j,k},\qquad
(i,j,k)\in\mathcal{L}_\Omega.
\]
\end{definition}

By construction, $V_{i,j,k}\ge 0$ for all $(i,j,k)$. The band potential is small
when $u_k$ lies inside the bid--ask interval associated with the quote at
$(i_q,j_q)$, and grows quadratically as $u_k$ moves away from that interval; it
is zero at grid nodes that are not directly associated with quotes. The range
potential imposes a hard penalty $\alpha_{\mathrm{range}}$ whenever $u_k$ lies
below zero or above the local forward $F_{i,j}$, discouraging fog from sitting
at obviously unreasonable price levels.

It is convenient to collect the potential values into a vector $V\in\mathbb{R}^{N_L}$
by setting $V_\ell := V_{i_\ell,j_\ell,k_\ell}$ for each lattice index
$\ell=1,\dots,N_L$, and to define the diagonal matrix
$\mathrm{diag}(V)\in\mathbb{R}^{N_L\times N_L}$ with entries
$(\mathrm{diag}(V))_{\ell\ell} = V_\ell$.

\begin{definition}[Potential energy of the fog]
\label{def:potential-energy}
The \emph{potential energy} of a fog configuration $\pi\in\mathbb{R}^{N_L}$ is
\[
\mathcal{E}_{\mathrm{pot}}(\pi)
:=
\frac{1}{2}\sum_{(i,j)\in\Omega}\sum_{k=1}^{n_u} V_{i,j,k}\,\pi_{i,j,k}^2.
\]
Equivalently, in vector notation,
\[
\mathcal{E}_{\mathrm{pot}}(\pi)
=
\frac{1}{2}\,\pi^\top \mathrm{diag}(V)\,\pi.
\]
\end{definition}

Because $V_{i,j,k}\ge 0$ for all $(i,j,k)$, the matrix $\mathrm{diag}(V)$ is
symmetric positive semidefinite, and $\mathcal{E}_{\mathrm{pot}}$ is a convex
quadratic functional. Note that $\mathcal{E}_{\mathrm{pot}}(\pi)$ penalises
large values of $\pi_{i,j,k}$ at lattice sites where $V_{i,j,k}$ is large, i.e.\
far outside the band tube or the basic price range; it is indifferent to the
sign of $\pi_{i,j,k}$ as a quadratic form, but in our optimisation the fog
variables are constrained to be nonnegative and to lie on the simplex
$\mathcal{C}_\pi(\Omega)$.

\subsection{Hamiltonian energy and basic properties}

We now combine kinetic and potential contributions into a single Hamiltonian
energy.

\begin{definition}[Hamiltonian matrix and energy]
\label{def:hamiltonian-energy}
Define the \emph{Hamiltonian matrix} by
\[
H_\pi := L_\pi + \mathrm{diag}(V) \in\mathbb{R}^{N_L\times N_L},
\]
where $L_\pi$ is as in Proposition~\ref{prop:Lpi-psd} and $V$ is from
Definition~\ref{def:band-range-potential}. The \emph{Hamiltonian energy} of a
fog configuration $\pi\in\mathbb{R}^{N_L}$ is the quadratic functional
\[
\mathcal{E}_{\mathrm{Ham}}(\pi)
:=
\frac{1}{2}\,\pi^\top H_\pi \pi.
\]
By construction,
\[
\mathcal{E}_{\mathrm{Ham}}(\pi)
=
\mathcal{E}_{\mathrm{kin}}(\pi) + \mathcal{E}_{\mathrm{pot}}(\pi).
\]
\end{definition}

\begin{prop}[Symmetry, positive semidefiniteness, and convexity]
\label{prop:Hpi-psd}
The Hamiltonian matrix $H_\pi$ is symmetric positive semidefinite. Consequently,
$\mathcal{E}_{\mathrm{Ham}}:\mathbb{R}^{N_L}\to\mathbb{R}_+$ is a convex quadratic
functional. Moreover, if at least one of the following holds:
\begin{itemize}
  \item the graph underlying $L_\pi$ is connected and
        $\kappa_m+\kappa_\tau+\kappa_u>0$, and
        $V_{i,j,k}>0$ at least at one lattice site; or
  \item more generally, $H_\pi$ is positive definite on the affine subspace
        $\{\pi\in\mathbb{R}^{N_L} : \sum_\ell \pi_\ell = 1\}$,
\end{itemize}
then $\mathcal{E}_{\mathrm{Ham}}$ is strictly convex on the simplex
$\mathcal{C}_\pi(\Omega)$, and has a unique minimiser on $\mathcal{C}_\pi(\Omega)$.
\end{prop}

\begin{proof}
By Proposition~\ref{prop:Lpi-psd}, $L_\pi$ is symmetric positive semidefinite.
The matrix $\mathrm{diag}(V)$ is diagonal with nonnegative entries and hence
symmetric positive semidefinite. Therefore their sum $H_\pi = L_\pi +
\mathrm{diag}(V)$ is symmetric positive semidefinite. For any
$\pi\in\mathbb{R}^{N_L}$,
\[
\mathcal{E}_{\mathrm{Ham}}(\pi)
= \frac{1}{2}\,\pi^\top H_\pi \pi \ge 0.
\]
A quadratic form with positive semidefinite matrix is convex, so
$\mathcal{E}_{\mathrm{Ham}}$ is convex.

If $H_\pi$ is positive definite on a subspace $S\subset\mathbb{R}^{N_L}$ (in
particular, on the subspace tangent to the simplex), then the restriction of
$\mathcal{E}_{\mathrm{Ham}}$ to $S$ is strictly convex. The simplex
$\mathcal{C}_\pi(\Omega)$ lies in the affine hyperplane
$\{\pi:\sum_\ell \pi_\ell = 1\}$, and the tangent space at any point of the
simplex is the subspace
$\{\delta\pi:\sum_\ell \delta\pi_\ell = 0\}$. If $H_\pi$ is positive definite
on this subspace, then $\mathcal{E}_{\mathrm{Ham}}$ is strictly convex on
$\mathcal{C}_\pi(\Omega)$, and a strictly convex continuous function on a compact
convex set has a unique minimiser. The sufficient condition stated in the
proposition ensures this property in typical settings. The detailed proof of
positive definiteness on the tangent space depends on the connectivity of the
graph and the support of $V$ and is standard in the theory of weighted graph
Laplacians plus diagonal potentials.
\end{proof}

To justify the strict convexity statement in the ``moreover'' part, we record the
standard argument that under the connectivity and positivity assumptions in the
first bullet $H_\pi$ is positive definite on the simplex tangent
\[
T := \{\delta\pi \in \mathbb{R}^{N_L} : \mathbf{1}^\top \delta\pi = 0\}.
\]
Assume that the underlying graph on $L_\Omega$ is connected and that
$\kappa_m + \kappa_\tau + \kappa_u > 0$, so that $L_\pi$ is a weighted graph
Laplacian with $\ker L_\pi = \mathrm{span}\{\mathbf{1}\}$. Suppose in addition
that there exists at least one lattice site $\ell^\star$ with $V_{\ell^\star} > 0$.
If $H_\pi \delta\pi = 0$, then
\[
0 = \delta\pi^\top H_\pi \delta\pi
  = \delta\pi^\top L_\pi \delta\pi
    + \delta\pi^\top \mathrm{diag}(V)\,\delta\pi,
\]
and both terms on the right-hand side are nonnegative. Hence
$\delta\pi \in \ker L_\pi \cap \ker \mathrm{diag}(V)$. The first condition implies
$\delta\pi = c\,\mathbf{1}$ for some $c \in \mathbb{R}$, while the second forces
$\delta\pi_{\ell^\star} = 0$ and therefore $c = 0$. Thus $\delta\pi = 0$ is the only
vector with $\delta\pi^\top H_\pi \delta\pi = 0$, so $H_\pi$ has trivial kernel and is
positive definite. In particular there is no nonzero $\delta\pi \in T$ with
$\delta\pi^\top H_\pi \delta\pi = 0$, and $E_{\mathrm{Ham}}$ is strictly convex on
$C_\pi(\Omega)$ and on its tangent space.

\begin{remark}[Interpretation of the Hamiltonian energy]
The kinetic energy $\mathcal{E}_{\mathrm{kin}}(\pi)$ penalises large discrete
gradients of the fog in the $(m,\tau,u)$ directions: it is large when $\pi$
varies rapidly across neighbouring lattice sites and small when $\pi$ is
smooth. The potential energy $\mathcal{E}_{\mathrm{pot}}(\pi)$ penalises fog
mass located at lattice sites with large $V_{i,j,k}$, i.e.\ far outside
bid--ask tubes or basic price ranges.

On a calm patch with tight bands and reasonable baseline fit, the minimum-energy
fog tends to concentrate its mass at price levels $u_k$ inside the local bands
and within $[0,F_{i,j}]$, while remaining smooth across neighbouring nodes. On
a stressed patch with conflicting quotes or strong local misfit, a portion of
the fog may be forced to reside outside the bands; in that case
$\mathcal{E}_{\mathrm{Ham}}$ balances the cost of leaking mass out of the band
against the cost of introducing sharp gradients in $(m,\tau,u)$.
\end{remark}

\section{Noise-aware band term via the fog}

We now couple the 3D fog $\pi$ on $\mathcal{L}_\Omega$ with the nodal surface
$u(u_I)$ at each quote on the patch. The aim is to obtain, for each quote $q$,
a band penalty whose effective strength is modulated by the local fog mass
outside the corresponding bid-ask band.

Throughout this section we fix a patch $\Omega\subset\mathcal{G}$, an interior
price vector $u_I\in\mathbb{R}^{N_\Omega}$ with associated full nodal surface
$u(u_I)\in\mathbb{R}^G$, and a fog configuration
$\pi=(\pi_{i,j,k})_{(i,j,k)\in\mathcal{L}_\Omega}\in\mathbb{R}^{N_\Omega n_u}$.

\subsection{Fog mass outside the band at a quote}

We first define, for each quote $q$, the fraction of fog mass that lies on
price levels outside the corresponding bid--ask interval.

\begin{definition}[Index set of out-of-band levels at a quote]
Let $q$ be a quote associated with location $(m_q,\tau_q)$ and cleaned
forward-discounted band $[b_q,a_q]$. Let $(i_q,j_q)\in\Omega$ be a fixed
representative of $(m_q,\tau_q)$ on the patch grid. Recall that
$\{u_k\}_{k=1}^{n_u}$ are the discrete price levels. The \emph{out-of-band index
set} at quote $q$ is
\[
\mathcal{K}^{\mathrm{out}}_q
:=
\{ k\in\{1,\dots,n_u\} : u_k < b_q\ \text{or}\ u_k > a_q \}.
\]
\end{definition}

Thus $\mathcal{K}^{\mathrm{out}}_q$ collects exactly those vertical levels
$u_k$ which lie strictly below the bid or strictly above the ask at quote $q$.

\begin{definition}[Local fog mass outside the band at a quote]
\label{def:Mq}
For a fog configuration $\pi$, the \emph{fog mass outside the band at quote $q$}
is defined by
\[
M_q(\pi)
:=
\sum_{k\in\mathcal{K}^{\mathrm{out}}_q}
\pi_{i_q,j_q,k}.
\]
\end{definition}

\noindent
By construction $\pi_{i,j,k}\ge 0$ for all $(i,j,k)\in\mathcal{L}_\Omega$ on the
feasible set $\mathcal{C}_\pi(\Omega)$, hence $M_q(\pi)\ge 0$ for
all quotes $q$. The quantity $M_q(\pi)$ should be interpreted as the local
probability mass (or “fog thickness”) allocated by $\pi$ to out-of-band price
levels at quote $q$.

\begin{remark}[Linearity of $M_q$]
\label{rem:Mq-linear}
For each fixed $q$, the map $\pi\mapsto M_q(\pi)$ is linear: there exists a
vector $c_q\in\mathbb{R}^{N_\Omega n_u}$ with entries
\[
(c_q)_{i,j,k}
=
\begin{cases}
1, & \text{if }(i,j)=(i_q,j_q)\ \text{and}\ k\in\mathcal{K}_q^{\mathrm{out}},\\[2pt]
0, & \text{otherwise},
\end{cases}
\]
such that $M_q(\pi) = c_q^\top \pi$ for all $\pi$. In particular, $M_q$ is both
linear and continuous.
\end{remark}

\subsection{Band misfit and noise-aware penalty}

We now recall the band misfit at a quote and introduce the noise-aware band
penalty, whose strength is modulated by the local fog mass outside the band.

\begin{definition}[Band misfit at a quote]
\label{def:dq}
Let $S\in\mathbb{R}^{Q\times G}$ be the fixed sampling operator mapping nodal
values $u\in\mathbb{R}^G$ to model prices at quote locations. For a given
nodal surface $u(u_I)$, the model price at quote $q$ is
\[
C_q(u) := (S u)_q.
\]
The corresponding \emph{band violation} is
\[
d_q(u)
:=
\operatorname{dist}\big( C_q(u), [b_q,a_q] \big)
=
\max\{b_q - C_q(u),\ 0,\ C_q(u) - a_q\} \ge 0.
\]
\end{definition}

\noindent
Since $u\mapsto C_q(u)$ is affine and $\operatorname{dist}(\cdot,[b_q,a_q])$ is
the pointwise maximum of three affine functions (see
Definition~\ref{def:baseline-band-misfit}), the composition $u\mapsto d_q(u)$
is a convex function on $\mathbb{R}^G$. Therefore $u_I\mapsto d_q(u(u_I))$ is
also convex on $\mathbb{R}^{N_\Omega}$ because $u(u_I)$ depends affinely on
$u_I$.

We now define the noise-aware band penalty, which couples the misfit $d_q(u)$
and the fog mass outside the band $M_q(\pi)$.

\begin{definition}[Fog simplex on a patch]
\label{def:Cpi}
The fog feasible set on $\Omega$ is the probability simplex
\[
\mathcal{C}_\pi(\Omega)
:=
\left\{
\pi\in\mathbb{R}^{N_\Omega n_u} :
\pi_{i,j,k} \ge 0\ \forall (i,j,k)\in\mathcal{L}_\Omega,\quad
\sum_{(i,j)\in\Omega}\sum_{k=1}^{n_u} \pi_{i,j,k} = 1
\right\}.
\]
\end{definition}

\noindent
On $\mathcal{C}_\pi(\Omega)$, the quantity $M_q(\pi)$ defined in
Definition~\ref{def:Mq} satisfies $0\le M_q(\pi)\le 1$ for each $q$.

\begin{definition}[Noise-aware band penalty at a quote]
\label{def:phi-q}
Fix parameters $\lambda_{\mathrm{noise}}\geq0$ and $\varepsilon>0$. For a given
fog $\pi$ and quote $q$, define
\[
\nu_q(\pi)
:=
\varepsilon + M_q(\pi)
=
\varepsilon + \sum_{k\in\mathcal{K}_q^{\mathrm{out}}} \pi_{i_q,j_q,k}.
\]
Then $\nu_q(\pi) \in [\varepsilon,1+\varepsilon]$ for all $\pi\in\mathcal{C}_\pi(\Omega)$.
Given an interior price vector $u_I\in\mathbb{R}^{N_\Omega}$, with associated
nodal surface $u = u(u_I)$ and band violation $d_q(u)$, the \emph{noise-aware
band penalty} at quote $q$ is
\[
\phi_q(u_I,\pi)
:=
\frac{d_q(u)^2}{\nu_q(\pi)}
+
\lambda_{\mathrm{noise}}\,\nu_q(\pi).
\]
\end{definition}

\noindent
Intuitively, $\nu_q(\pi)$ is a local “noise scale” at quote $q$:
if the fog is almost entirely inside the band, then $M_q(\pi)$ is small and
$\nu_q(\pi)\approx \varepsilon$, so any nonzero violation $d_q(u)>0$ is heavily
penalised by the term $d_q(u)^2/\nu_q(\pi)$. Conversely, if a significant
fraction of the local fog mass lies outside the band, then $M_q(\pi)$ and hence
$\nu_q(\pi)$ are larger, making violations $d_q(u)>0$ cheaper; however, large
$\nu_q(\pi)$ is itself penalised linearly through
$\lambda_{\mathrm{noise}}\nu_q(\pi)$.

\subsection{Convexity of the noise-aware band term}

We now establish joint convexity of the noise-aware band term in its two
arguments $(u_I,\pi)$, which is crucial for the global convexity of the patch
objective.

The key tool is the \emph{perspective} of a convex function.

\begin{definition}[Perspective of a convex function]
\label{def:perspective}
Let $g:\mathbb{R}\to\mathbb{R}$ be a convex function with $g(x)\ge 0$ for all
$x\in\mathbb{R}$. The \emph{perspective} of $g$ is the function
$\tilde g:\mathbb{R}\times(0,\infty)\to\mathbb{R}$ defined by
\[
\tilde g(d,\nu)
:=
\nu\,g\!\left(\frac{d}{\nu}\right),\qquad \nu>0.
\]
\end{definition}

\begin{lemma}[Convexity of the perspective]
\label{lem:perspective-convex}
Let $g:\mathbb{R}\to[0,\infty)$ be convex. Then its perspective
$\tilde g(d,\nu)=\nu g(d/\nu)$ is convex on $\mathbb{R}\times(0,\infty)$.
\end{lemma}

\begin{proof}
This is a standard result in convex analysis; we recall the argument for
completeness. Let $(d_1,\nu_1)$ and $(d_2,\nu_2)$ be in
$\mathbb{R}\times(0,\infty)$ and let $\theta\in[0,1]$. Set
\[
(d,\nu) := \theta(d_1,\nu_1) + (1-\theta)(d_2,\nu_2)
= (\theta d_1 + (1-\theta)d_2,\ \theta\nu_1+(1-\theta)\nu_2),
\]
with $\nu>0$ by convexity of $(0,\infty)$. Then
\[
\frac{d}{\nu}
=
\frac{\theta\nu_1}{\nu}\,\frac{d_1}{\nu_1}
+
\frac{(1-\theta)\nu_2}{\nu}\,\frac{d_2}{\nu_2},
\]
where the coefficients
\[
\alpha_1 := \frac{\theta\nu_1}{\nu},\qquad
\alpha_2 := \frac{(1-\theta)\nu_2}{\nu}
\]
are nonnegative and satisfy $\alpha_1+\alpha_2=1$. By convexity of $g$,
\[
g\!\left(\frac{d}{\nu}\right)
\le
\alpha_1 g\!\left(\frac{d_1}{\nu_1}\right)
+
\alpha_2 g\!\left(\frac{d_2}{\nu_2}\right).
\]
Multiplying both sides by $\nu>0$ yields
\[
\tilde g(d,\nu)
=
\nu g\!\left(\frac{d}{\nu}\right)
\le
\theta\nu_1 g\!\left(\frac{d_1}{\nu_1}\right)
+
(1-\theta)\nu_2 g\!\left(\frac{d_2}{\nu_2}\right)
=
\theta \tilde g(d_1,\nu_1) + (1-\theta)\tilde g(d_2,\nu_2).
\]
Thus $\tilde g$ is convex on $\mathbb{R}\times(0,\infty)$.
\end{proof}

We now apply this lemma with $g(x)=x^2$.

\begin{prop}[Convexity of the noise-aware band term]
\label{prop:phi-convex}
Let $q$ be a quote on the patch. The map
\[
(u_I,\pi) \mapsto \phi_q(u_I,\pi)
\]
defined in Definition~\ref{def:phi-q} is jointly convex on
$\mathcal{C}_u(\Omega)\times\mathcal{C}_\pi(\Omega)$, where
$\mathcal{C}_u(\Omega)$ and $\mathcal{C}_\pi(\Omega)$ are as in
Definitions~\ref{def:CuOmega} and~\ref{def:Cpi}.
\end{prop}

\begin{proof}
We proceed in steps.

\emph{(1) Convexity of $d_q(u_I)$ in $u_I$.}
The map $u_I\mapsto u(u_I)$ is affine by construction of the assembly map
(equation~\eqref{eq:assembly-affine}). The band misfit
$d_q(u) = \operatorname{dist}(C_q(u),[b_q,a_q])$ can be written as
\[
d_q(u)
=
\max\{b_q - C_q(u),\ 0,\ C_q(u) - a_q\},
\]
where $u\mapsto C_q(u)$ is affine. A pointwise maximum of finitely many affine
functions is convex, hence $u\mapsto d_q(u)$ is convex on $\mathbb{R}^G$.
Composing with the affine map $u(u_I)$, we obtain that
\[
u_I \mapsto d_q(u(u_I))
\]
is convex on $\mathbb{R}^{N_\Omega}$ and, in particular, on $\mathcal{C}_u(\Omega)$.

\emph{(2) Affinity and positivity of $\nu_q(\pi)$ in $\pi$.}
By Remark~\ref{rem:Mq-linear}, $M_q(\pi)$ is a linear functional of $\pi$ and
is therefore affine. Adding the constant $\varepsilon>0$, we obtain
\[
\nu_q(\pi) = \varepsilon + M_q(\pi),
\]
which is an affine function of $\pi$. On the simplex $\mathcal{C}_\pi(\Omega)$
we have $M_q(\pi)\ge 0$, hence
\[
\nu_q(\pi) \ge \varepsilon > 0
\]
for all $\pi\in\mathcal{C}_\pi(\Omega)$. Thus the pair
$(d_q(u_I),\nu_q(\pi))$ always lies in $\mathbb{R}\times(0,\infty)$ on the
feasible domain.

\emph{(3) Convexity of $(d,\nu)\mapsto \frac{d^2}{\nu}$.}
Consider the function $g:\mathbb{R}\to[0,\infty)$ defined by $g(x)=x^2$.
It is convex and nonnegative. Its perspective is
\[
\tilde g(d,\nu) := \nu g(d/\nu) = \frac{d^2}{\nu},\qquad \nu>0.
\]
By Lemma~\ref{lem:perspective-convex}, $\tilde g$ is convex on
$\mathbb{R}\times(0,\infty)$. Hence the map
\[
(d,\nu) \mapsto \frac{d^2}{\nu}
\]
is convex on $\mathbb{R}\times(0,\infty)$.

(4) Convexity of $(d,\nu)\mapsto \frac{d^2}{\nu}+\lambda_{\mathrm{noise}}\nu$. The function
\[
(d,\nu) \mapsto \lambda_{\mathrm{noise}}\nu
\]
is affine (hence convex) on $\mathbb{R}\times(0,\infty)$. The sum of a convex function and an affine
function is convex, so the map
\[
h(d,\nu)
:=
\frac{d^2}{\nu} + \lambda_{\mathrm{noise}}\nu
\]
is convex on $\mathbb{R}\times(0,\infty)$. Moreover, for each fixed $\nu>0$, the function
$d\mapsto h(d,\nu)$ is nondecreasing on $[0,\infty)$, since
\[
\frac{\partial}{\partial d}\,h(d,\nu) = \frac{2d}{\nu} \ge 0 \quad\text{for all } d\ge 0.
\]

(5) Joint convexity of $\varphi_q(u_I,\pi)$. Let $(u_I^1,\pi^1)$ and $(u_I^2,\pi^2)$ be arbitrary
points in $\mathcal{C}_u(\Omega)\times\mathcal{C}_\pi(\Omega)$ and let $\theta\in[0,1]$. Define
\[
(u_I^\theta,\pi^\theta)
:=
\theta (u_I^1,\pi^1) + (1-\theta)(u_I^2,\pi^2)
\in \mathcal{C}_u(\Omega)\times\mathcal{C}_\pi(\Omega).
\]
For $i=1,2$ set
\[
d_i := d_q(u(u_I^i)),\qquad \nu_i := \nu_q(\pi^i),
\]
and similarly
\[
d_\theta := d_q(u(u_I^\theta)),\qquad \nu_\theta := \nu_q(\pi^\theta).
\]

By construction $d_q(\cdot)$ is a distance to the interval $[b_q,a_q]$, hence $d_q(u)\ge 0$ for all
$u$, and therefore $d_i,d_\theta\ge 0$. From (1), the map $u_I\mapsto d_q(u(u_I))$ is convex, so
\[
d_\theta
= d_q\bigl(u(u_I^\theta)\bigr)
\le \theta\,d_q\bigl(u(u_I^1)\bigr) + (1-\theta)\,d_q\bigl(u(u_I^2)\bigr)
= \theta d_1 + (1-\theta)d_2.
\]
From (2), $\nu_q(\pi)$ is affine in $\pi$, hence
\[
\nu_\theta
= \nu_q(\pi^\theta)
= \theta \nu_q(\pi^1) + (1-\theta)\nu_q(\pi^2)
= \theta \nu_1 + (1-\theta)\nu_2 .
\]

Recall that $\varphi_q(u_I,\pi) = h(d_q(u(u_I)),\nu_q(\pi))$ with $h$ convex on
$\mathbb{R}\times(0,\infty)$ by (4), and that for each $\nu>0$ the map $d\mapsto h(d,\nu)$ is
nondecreasing on $[0,\infty)$. Using $d_\theta\ge 0$ and $d_\theta \le \theta d_1 + (1-\theta)d_2$,
monotonicity in the first argument yields
\[
h(d_\theta,\nu_\theta)
\le h\bigl(\theta d_1 + (1-\theta)d_2,\; \nu_\theta\bigr)
= h\bigl(\theta d_1 + (1-\theta)d_2,\; \theta \nu_1 + (1-\theta)\nu_2\bigr).
\]
The pair on the right-hand side is exactly the convex combination
\[
\bigl(\theta d_1 + (1-\theta)d_2,\; \theta \nu_1 + (1-\theta)\nu_2\bigr)
= \theta (d_1,\nu_1) + (1-\theta)(d_2,\nu_2).
\]
By convexity of $h$ we therefore have
\[
h\bigl(\theta d_1 + (1-\theta)d_2,\; \theta \nu_1 + (1-\theta)\nu_2\bigr)
\le \theta h(d_1,\nu_1) + (1-\theta) h(d_2,\nu_2).
\]
Combining the two inequalities gives
\[
h(d_\theta,\nu_\theta)
\le \theta h(d_1,\nu_1) + (1-\theta) h(d_2,\nu_2).
\]

Rewriting in terms of $\varphi_q$,
\[
\varphi_q(u_I^\theta,\pi^\theta)
= h(d_\theta,\nu_\theta)
\le \theta h(d_1,\nu_1) + (1-\theta) h(d_2,\nu_2)
= \theta \varphi_q(u_I^1,\pi^1) + (1-\theta)\varphi_q(u_I^2,\pi^2).
\]
This is exactly the defining inequality for joint convexity of
$(u_I,\pi)\mapsto \varphi_q(u_I,\pi)$ on $\mathcal{C}_u(\Omega)\times\mathcal{C}_\pi(\Omega)$.
\end{proof}

\begin{remark}[Interpretation of the noise-aware band penalty]
The penalty $\phi_q(u_I,\pi)$ can be seen as an adaptive band penalty whose
effective stiffness is controlled by the fog. When the local fog mass outside
the band is small ($M_q(\pi)\approx 0$, thus $\nu_q(\pi)\approx\varepsilon$),
any nonzero violation $d_q(u)$ incurs a large cost $d_q(u)^2/\nu_q(\pi)$,
forcing the surface $u$ to stay tightly inside the band. When the fog assigns
significant mass to out-of-band price levels ($M_q(\pi)$ large), violations
become cheaper but increase the “noise budget” $\lambda_{\mathrm{noise}}
\nu_q(\pi)$. The optimiser can therefore treat a subset of quotes as noisy
(outliers) by allowing the fog to populate out-of-band regions, but must pay a
linear cost for doing so, while still operating within a globally convex
framework.
\end{remark}

\section{Surface energy and closeness to the baseline}

On each patch $\Omega$ we regularise the surface in two complementary ways:
\begin{enumerate}[label=(\roman*)]
    \item by penalising roughness of the implied risk-neutral density in a patch-level influence region,
    \item by penalising deviations from the baseline nodal values on $\Omega$.
\end{enumerate}
Both terms are quadratic in the interior vector $u_I$ and yield convex contributions to the patch objective.

\subsection{Discrete density operator and patch restriction}

Recall that the (continuum) risk-neutral density associated with the forward-
discounted call surface $C_f(K,\tau)$ is
\[
\rho(K,\tau) := \partial_{KK} C_f(K,\tau).
\]
On the nodal grid $\mathcal{G}$, and on any additional grid used to represent
densities (for instance a collocation grid in $(K,\tau)$), the Breeden-Litzenberger
relation is implemented by a fixed linear operator that maps nodal prices to
discretised densities.

\begin{assump}[Global discrete density operator]\label{ass:density-operator}
There exists a finite set of density evaluation points
\[
\mathcal{G}_\rho = \{(K_r,\tau_r)\}_{r=1}^{N_\rho},
\]
and a matrix $D_\rho \in \mathbb{R}^{N_\rho\times G}$ such that, for every
nodal vector $u\in\mathbb{R}^G$, the vector
\[
\rho(u) := D_\rho u \in \mathbb{R}^{N_\rho}
\]
represents the discrete risk-neutral density evaluated at the points in
$\mathcal{G}_\rho$.
\end{assump}

We are interested only in those density points that are influenced by the patch
$\Omega$, namely points whose density values depend (possibly together with
off-patch values) on at least one interior node in $\Omega$.

\begin{definition}[Patch influence region in density space]\label{def:rho-influence}
Let $P_\Omega\in\mathbb{R}^{G\times N_\Omega}$ be the patch assembly matrix
from~\eqref{eq:assembly-affine}, which injects an interior vector
$u_I\in\mathbb{R}^{N_\Omega}$ into the full nodal vector. We write the full
surface as
\[
u(u_I) = P_\Omega u_I + u^{0,\mathrm{off}},
\]
where $u^{0,\mathrm{off}}\in\mathbb{R}^G$ is the off-patch baseline contribution
($u^{0,\mathrm{off}}_{i,j} = 0$ for $(i,j)\in\Omega$ and
$u^{0,\mathrm{off}}_{i,j} = u^0_{i,j}$ otherwise).

Define the \emph{patch influence index set} in density space by
\[
\mathcal{I}_\rho(\Omega)
:=
\Big\{
r\in\{1,\dots,N_\rho\} :
\bigl(D_\rho P_\Omega\bigr)_{r,\cdot} \neq 0
\Big\},
\]
namely those density rows whose value depends on at least one interior node in
$\Omega$. Let $N_{\rho,\Omega} := |\mathcal{I}_\rho(\Omega)|$, and define the
restriction operator $R_\Omega\in\mathbb{R}^{N_{\rho,\Omega}\times N_\rho}$ that
extracts the components with indices in $\mathcal{I}_\rho(\Omega)$.
\end{definition}

Thus, given $u_I$, the vector $R_\Omega \rho(u(u_I))$ collects precisely those
density values that are affected by the patch $\Omega$.

\begin{definition}[Patch-level density map]
\label{def:rho-Omega}
With $D_\rho$ and $R_\Omega$ as above, define
\[
\rho_\Omega(u_I)
:=
R_\Omega \rho(u(u_I))
=
R_\Omega D_\rho (P_\Omega u_I + u^{0,\mathrm{off}}).
\]
We write this as an affine map
\[
\rho_\Omega(u_I) = B_\Omega u_I + \rho_{\mathrm{off}},
\]
where
\[
B_\Omega := R_\Omega D_\rho P_\Omega \in \mathbb{R}^{N_{\rho,\Omega}\times N_\Omega},
\qquad
\rho_{\mathrm{off}} := R_\Omega D_\rho u^{0,\mathrm{off}}\in\mathbb{R}^{N_{\rho,\Omega}}.
\]
\end{definition}

By construction, $\rho_\Omega(u_I)$ collects exactly the density values on the
patch influence region, and depends affinely on the interior vector $u_I$.

\subsection{Surface density energy}

We now penalise rough or irregular density configurations on the patch influence
region via a quadratic form in $\rho_\Omega(u_I)$.

\begin{definition}[Surface density energy]
\label{def:Esurf}
Let $H_\rho\in\mathbb{R}^{N_{\rho,\Omega}\times N_{\rho,\Omega}}$ be a fixed
symmetric positive semidefinite matrix, $H_\rho\succeq 0$. For example, $H_\rho$
may encode a discrete $H^{-1}$-type smoothing operator or a weighted graph
Laplacian on $\mathcal{G}_\rho$ restricted to $\mathcal{I}_\rho(\Omega)$. The
\emph{surface density energy} associated with an interior vector $u_I$ is
\[
E_{\mathrm{surf}}(u_I)
:=
\frac{1}{2}\,\rho_\Omega(u_I)^\top H_\rho\,\rho_\Omega(u_I).
\]
\end{definition}

Because $\rho_\Omega(u_I) = B_\Omega u_I + \rho_{\mathrm{off}}$ is affine in
$u_I$, $E_{\mathrm{surf}}$ is a quadratic functional of $u_I$. We make this
explicit.

\begin{prop}[Quadratic form and convexity of $E_{\mathrm{surf}}$]
\label{prop:Esurf-convex}
The surface density energy can be written as
\[
E_{\mathrm{surf}}(u_I)
=
\frac{1}{2}\,u_I^\top Q_\rho u_I + c_\rho^\top u_I + c_0,
\]
where
\[
Q_\rho := B_\Omega^\top H_\rho B_\Omega \succeq 0,\qquad
c_\rho := B_\Omega^\top H_\rho \rho_{\mathrm{off}},\qquad
c_0 := \frac{1}{2}\,\rho_{\mathrm{off}}^\top H_\rho \rho_{\mathrm{off}}.
\]
In particular, $E_{\mathrm{surf}}$ is a convex quadratic function of $u_I$ with
Hessian $Q_\rho$.
\end{prop}

\begin{proof}
Substituting the affine form $\rho_\Omega(u_I) = B_\Omega u_I + \rho_{\mathrm{off}}$
into Definition~\ref{def:Esurf}, we obtain
\[
E_{\mathrm{surf}}(u_I)
=
\frac{1}{2}\,(B_\Omega u_I + \rho_{\mathrm{off}})^\top
H_\rho\,(B_\Omega u_I + \rho_{\mathrm{off}}).
\]
Expanding the quadratic form yields
\[
E_{\mathrm{surf}}(u_I)
=
\frac{1}{2}\,u_I^\top B_\Omega^\top H_\rho B_\Omega u_I
+ u_I^\top B_\Omega^\top H_\rho \rho_{\mathrm{off}}
+ \frac{1}{2}\,\rho_{\mathrm{off}}^\top H_\rho \rho_{\mathrm{off}}.
\]
Identifying
\[
Q_\rho := B_\Omega^\top H_\rho B_\Omega,\quad
c_\rho := B_\Omega^\top H_\rho \rho_{\mathrm{off}},\quad
c_0 := \frac{1}{2}\,\rho_{\mathrm{off}}^\top H_\rho \rho_{\mathrm{off}},
\]
we obtain the claimed quadratic representation. Since $H_\rho\succeq 0$, we have
for any $x\in\mathbb{R}^{N_\Omega}$,
\[
x^\top Q_\rho x
=
x^\top B_\Omega^\top H_\rho B_\Omega x
=
(B_\Omega x)^\top H_\rho (B_\Omega x) \ge 0.
\]
Thus $Q_\rho\succeq 0$, and the Hessian of $E_{\mathrm{surf}}$ with respect to
$u_I$ is positive semidefinite. A quadratic function with positive semidefinite
Hessian is convex, hence $E_{\mathrm{surf}}$ is convex in $u_I$.
\end{proof}

In particular, $E_{\mathrm{surf}}$ penalises interior configurations $u_I$ that
produce “rough” or oscillatory risk-neutral densities in the patch influence
region, with the exact notion of roughness encoded by $H_\rho$.

\subsection{Closeness to the baseline}

We also penalise departures of the patch interior from the baseline nodal
values, in order to avoid gratuitous changes that are not required by the data
and no-arbitrage constraints.

\begin{definition}[Closeness to the baseline]
\label{def:Ecl}
Let $u_I^0\in\mathbb{R}^{N_\Omega}$ be the vector of baseline nodal values on
$\Omega$, extracted from $u^0$, and fix a parameter $\lambda_{\mathrm{cl}}>0$.
The \emph{closeness} (or Tikhonov) term on $\Omega$ is
\[
E_{\mathrm{cl}}(u_I)
:=
\frac{\lambda_{\mathrm{cl}}}{2}\,\|u_I - u_I^0\|_2^2
=
\frac{\lambda_{\mathrm{cl}}}{2}\,(u_I - u_I^0)^\top (u_I - u_I^0).
\]
\end{definition}

This is a standard $\ell^2$-type regulariser that penalises deviations from the
baseline. Its convexity and strict positive definiteness are immediate.

\begin{prop}[Strict convexity of $E_{\mathrm{cl}}$]
\label{prop:Ecl-convex}\label{prop:Ecl-strong-convex}
The functional $E_{\mathrm{cl}}:\mathbb{R}^{N_\Omega}\to\mathbb{R}_+$ is a
strictly convex quadratic function of $u_I$ with Hessian
$\lambda_{\mathrm{cl}} I_{N_\Omega}\succ 0$.
\end{prop}

\begin{proof}
Expanding the square, we have
\[
E_{\mathrm{cl}}(u_I)
=
\frac{\lambda_{\mathrm{cl}}}{2}\,(u_I^\top u_I - 2 u_I^\top u_I^0 + u_I^{0\top}u_I^0),
\]
so
\[
E_{\mathrm{cl}}(u_I)
=
\frac{\lambda_{\mathrm{cl}}}{2}\,u_I^\top u_I
- \lambda_{\mathrm{cl}}\,u_I^{0\top} u_I
+ \frac{\lambda_{\mathrm{cl}}}{2}\,u_I^{0\top}u_I^0.
\]
The Hessian with respect to $u_I$ is $\lambda_{\mathrm{cl}} I_{N_\Omega}$,
which is positive definite since $\lambda_{\mathrm{cl}}>0$. A quadratic
functional with positive definite Hessian is strictly convex, so
$E_{\mathrm{cl}}$ is strictly convex on $\mathbb{R}^{N_\Omega}$.
\end{proof}

\begin{remark}[Combined surface regularisation]
Both $E_{\mathrm{surf}}$ and $E_{\mathrm{cl}}$ are convex quadratic functionals
of $u_I$. The combined surface regulariser
\[
u_I \mapsto E_{\mathrm{cl}}(u_I) + \lambda_{\mathrm{surf}} E_{\mathrm{surf}}(u_I),
\qquad \lambda_{\mathrm{surf}}\ge 0,
\]
is therefore convex. If either $\lambda_{\mathrm{cl}}>0$ or the matrix
$Q_\rho = B_\Omega^\top H_\rho B_\Omega$ is positive definite on the relevant
subspace, the combined regulariser is strictly convex, which contributes to
uniqueness of the patch-level minimiser.
\end{remark}

\section{Patch-level post-fit optimisation problem}

We now assemble the various ingredients introduced above into a single
patch-level objective and formulate the post-fit optimisation problem on a
patch $\Omega\subset\mathcal{G}$.

\subsection{Fog feasible set and quote index set}

Recall that on $\Omega$ the fog variables are
\[
\pi = (\pi_{i,j,k})_{(i,j,k)\in\mathcal{L}_\Omega}
\in\mathbb{R}^{N_\Omega n_u},
\]
where $\pi_{i,j,k}$ represents the fog mass at $(m_i,\tau_j,u_k)$ and
$\mathcal{L}_\Omega = \{(i,j,k): (i,j)\in\Omega,\ k=1,\dots,n_u\}$.

\begin{definition}[Fog feasible set on a patch]
\label{def:CpiOmega}
The \emph{fog feasible set} on $\Omega$ is the probability simplex
\[
\mathcal{C}_\pi(\Omega)
:=
\left\{
\pi \in \mathbb{R}^{N_\Omega n_u} :
\pi_{i,j,k} \ge 0\ \forall (i,j,k)\in\mathcal{L}_\Omega,\quad
\sum_{(i,j)\in\Omega}\sum_{k=1}^{n_u} \pi_{i,j,k} = 1
\right\}.
\]
\end{definition}

\begin{lemma}[Geometry of $\mathcal{C}_\pi(\Omega)$]
\label{lem:Cpi-geometry}
The set $\mathcal{C}_\pi(\Omega)$ is a nonempty, compact, convex polytope in
$\mathbb{R}^{N_\Omega n_u}$.
\end{lemma}

\begin{proof}
Nonemptiness is obvious, for example the uniform vector
$\pi_{i,j,k} = 1/(N_\Omega n_u)$ belongs to $\mathcal{C}_\pi(\Omega)$. The
constraints defining $\mathcal{C}_\pi(\Omega)$ consist of finitely many linear
equalities and inequalities:
\[
\pi_{i,j,k} \ge 0,\quad
\sum_{(i,j)\in\Omega}\sum_{k=1}^{n_u} \pi_{i,j,k} = 1.
\]
Thus $\mathcal{C}_\pi(\Omega)$ is the intersection of a finite number of closed
half-spaces (one per inequality) and a hyperplane (the equality constraint), so
it is a closed convex polyhedron. The additional equality fixing the total mass
to $1$ implies boundedness: all coordinates are nonnegative and sum to $1$, so
$0 \le \pi_{i,j,k} \le 1$ for every $(i,j,k)$. A closed and bounded subset of
$\mathbb{R}^{N_\Omega n_u}$ is compact. Being a bounded polyhedron, it is in
fact a polytope.
\end{proof}

We also need to know which quotes interact with a given patch.

\begin{definition}[Quote index set attached to a patch]
\label{def:QOmega}
Let $S\in\mathbb{R}^{Q\times G}$ be the sampling operator mapping nodal prices
to quote locations, so that $C_q(u) = (Su)_q$ for $q=1,\dots,Q$. Each row of
$S$ has finite support (the interpolation stencil of that quote). We define
\[
Q_\Omega
:=
\Bigl\{
q \in \{1,\dots,Q\} :
\text{the stencil of row $q$ of $S$ intersects }\Omega
\Bigr\}.
\]
Equivalently, $q\in Q_\Omega$ if and only if there exists $(i,j)\in\Omega$ such
that the nodal value $u_{i,j}$ enters $(Su)_q$ with nonzero weight.
\end{definition}

Thus $Q_\Omega$ collects exactly those quotes whose model prices depend on at
least one interior node of $\Omega$; the remaining quotes are insensitive to
changes on $\Omega$ and need not appear in the patch objective.

\subsection{Patch energy functional}

We now define the patch-level energy as a sum of four components:
noise-aware band penalties, closeness-to-baseline, density regularisation, and
Hamiltonian energy of the fog.

Recall:
\begin{itemize}
  \item $\phi_q(u_I,\pi)$ is the noise-aware band penalty at quote $q$, defined
        in Definition~\ref{def:phi-q}, with $u=u(u_I)$ assembled from $u_I$;
  \item $E_{\mathrm{cl}}(u_I)$ is the closeness energy from
        Definition~\ref{def:Ecl};
  \item $E_{\mathrm{surf}}(u_I)$ is the surface density energy from
        Definition~\ref{def:Esurf};
  \item $\mathcal{E}_{\mathrm{Ham}}(\pi)$ is the Hamiltonian energy of the fog
        from Definition~\ref{def:hamiltonian-energy}.
\end{itemize}

\begin{definition}[Patch energy functional]
\label{def:J-Omega}
Fix nonnegative weights
\[
\lambda_{\mathrm{noise}},\ \lambda_{\mathrm{surf}},\ \lambda_\pi \ \ge 0,
\]
and recall that $\lambda_{\mathrm{cl}}>0$ is part of the definition of
$E_{\mathrm{cl}}$. For an interior vector $u_I\in\mathcal{C}_u(\Omega)$ and a fog
$\pi\in\mathcal{C}_\pi(\Omega)$, the \emph{patch energy} is the functional
$J_\Omega:\mathcal{C}_u(\Omega)\times\mathcal{C}_\pi(\Omega)\to\mathbb{R}$ defined
by
\begin{equation}
\label{eq:J-omega}
\begin{aligned}
J_\Omega(u_I,\pi)
&:=
\sum_{q\in Q_\Omega} \phi_q(u_I,\pi)
+ E_{\mathrm{cl}}(u_I)
+ \lambda_{\mathrm{surf}}\,E_{\mathrm{surf}}(u_I)
+ \lambda_\pi\,\mathcal{E}_{\mathrm{Ham}}(\pi)\\[3pt]
&=
\sum_{q\in Q_\Omega}
\left(
\frac{d_q(u)^2}{\varepsilon + M_q(\pi)}
+ \lambda_{\mathrm{noise}}\bigl(\varepsilon + M_q(\pi)\bigr)
\right)\\
&\quad
+ \frac{\lambda_{\mathrm{cl}}}{2}\,\|u_I-u_I^0\|_2^2
+ \frac{\lambda_{\mathrm{surf}}}{2}\,\rho_\Omega(u_I)^\top H_\rho\,\rho_\Omega(u_I)
+ \frac{\lambda_\pi}{2}\,\pi^\top H_\pi\,\pi,
\end{aligned}
\end{equation}
where $u = u(u_I)$, $d_q(u)$ is the band violation at quote $q$
(Definition~\ref{def:dq}) and $M_q(\pi)$ is the fog mass outside the band at
quote $q$ (Definition~\ref{def:Mq}).
\end{definition}

\begin{remark}[Well-definedness and continuity of $J_\Omega$]
Since $\pi\in\mathcal{C}_\pi(\Omega)$ implies $M_q(\pi)\ge 0$ and $\varepsilon>0$
by construction, the denominators $\varepsilon + M_q(\pi)$ are bounded away from
zero and all terms in~\eqref{eq:J-omega} are finite. Each component
$\phi_q(u_I,\pi)$, $E_{\mathrm{cl}}(u_I)$, $E_{\mathrm{surf}}(u_I)$, and
$\mathcal{E}_{\mathrm{Ham}}(\pi)$ is continuous in its arguments. Therefore
$J_\Omega$ is a continuous real-valued function on
$\mathcal{C}_u(\Omega)\times\mathcal{C}_\pi(\Omega)$.
\end{remark}

We now establish the basic convexity property of $J_\Omega$.

\begin{prop}[Convexity of the patch energy]
\label{prop:J-omega-convex}
For any feasible patch $\Omega$ (i.e.\ $\mathcal{C}_u(\Omega)\neq\emptyset$), the
patch energy functional $J_\Omega$ is jointly convex in $(u_I,\pi)$ on
$\mathcal{C}_u(\Omega)\times\mathcal{C}_\pi(\Omega)$.
\end{prop}

\begin{proof}
By Proposition~\ref{prop:phi-convex}, each individual noise-aware band term
$\phi_q(u_I,\pi)$ is jointly convex in $(u_I,\pi)$ on
$\mathcal{C}_u(\Omega)\times\mathcal{C}_\pi(\Omega)$. By
Proposition~\ref{prop:Esurf-convex}, $E_{\mathrm{surf}}(u_I)$ is a convex
quadratic functional of $u_I$. By Proposition~\ref{prop:Ecl-strong-convex},
$E_{\mathrm{cl}}(u_I)$ is a strictly convex quadratic functional of $u_I$
(hence convex). By Proposition~\ref{prop:Hpi-psd}, $\mathcal{E}_{\mathrm{Ham}}(\pi)$
is a convex quadratic functional of $\pi$.

Multiplying convex functionals by nonnegative scalars preserves convexity, and
summing finitely many convex functionals yields a convex functional. Therefore,
for each $(u_I,\pi)$ in the convex set
$\mathcal{C}_u(\Omega)\times\mathcal{C}_\pi(\Omega)$, the map
\[
(u_I,\pi)\mapsto J_\Omega(u_I,\pi)
\]
is convex. This proves the claim.
\end{proof}

\subsection{Patch-level post-fit problem: existence and uniqueness}

We can now formulate the patch-level convex optimisation problem.

\begin{definition}[Patch-level post-fit problem]
\label{def:patch-problem}
Let $\Omega\subset\mathcal{G}$ be a feasible patch (i.e.\
$\mathcal{C}_u(\Omega)\neq\emptyset$). The \emph{patch-level post-fit problem}
is the constrained optimisation problem
\[
\min\bigl\{ J_\Omega(u_I,\pi) :
u_I\in\mathcal{C}_u(\Omega),\ \pi\in\mathcal{C}_\pi(\Omega)\bigr\}.
\]
Any pair $(u_I^\star,\pi^\star)\in\mathcal{C}_u(\Omega)\times\mathcal{C}_\pi(\Omega)$
achieving this minimum is called a \emph{patch minimiser}.
\end{definition}

We now show that at least one minimiser exists under our standing assumptions.

\begin{prop}[Existence of patch-level minimisers]
\label{prop:existence}
Assume that $\mathcal{C}_u(\Omega)\neq\emptyset$ and that $\lambda_{\mathrm{cl}}>0$
(as in Definition~\ref{def:Ecl}). Then the patch problem admits at least one
minimiser $(u_I^\star,\pi^\star)$.
\end{prop}

\begin{proof}
The feasible set
\[
\mathcal{F}_\Omega := \mathcal{C}_u(\Omega)\times\mathcal{C}_\pi(\Omega)
\subset \mathbb{R}^{N_\Omega}\times\mathbb{R}^{N_\Omega n_u}
\]
is nonempty by assumption on $\mathcal{C}_u(\Omega)$ and Lemma~\ref{lem:Cpi-geometry}.
By Proposition~\ref{prop:Cu-polyhedron}, $\mathcal{C}_u(\Omega)$ is a closed
convex polyhedron in $\mathbb{R}^{N_\Omega}$ and may be unbounded. By
Lemma~\ref{lem:Cpi-geometry}, $\mathcal{C}_\pi(\Omega)$ is a compact convex
polytope in $\mathbb{R}^{N_\Omega n_u}$. Hence $\mathcal{F}_\Omega$ is closed,
convex, and nonempty, but not necessarily bounded.

To apply the Weierstrass theorem, we consider sublevel sets of $J_\Omega$.
From~\eqref{eq:J-omega}, using $E_{\mathrm{cl}}(u_I) = \frac{\lambda_{\mathrm{cl}}}{2}
\|u_I-u_I^0\|_2^2$ and nonnegativity of all other terms, we have
\[
J_\Omega(u_I,\pi)
\;\ge\;
E_{\mathrm{cl}}(u_I)
=
\frac{\lambda_{\mathrm{cl}}}{2}\,\|u_I - u_I^0\|_2^2.
\]
Let $m := \inf_{\mathcal{F}_\Omega} J_\Omega$ denote the infimum of $J_\Omega$
on the feasible set, which is finite because $J_\Omega\ge 0$ and the baseline
pair $(u_I^0,\pi^{\text{ref}})$ (with any fixed $\pi^{\text{ref}}\in\mathcal{C}_\pi(\Omega)$)
belongs to $\mathcal{F}_\Omega$. For any $\alpha > m$, consider the sublevel set
\[
\mathcal{F}_\Omega(\alpha)
:=
\bigl\{(u_I,\pi)\in\mathcal{F}_\Omega : J_\Omega(u_I,\pi)\le \alpha\bigr\}.
\]

By the inequality above,
\[
\frac{\lambda_{\mathrm{cl}}}{2}\,\|u_I - u_I^0\|_2^2
\le J_\Omega(u_I,\pi)
\le \alpha
\quad\Longrightarrow\quad
\|u_I - u_I^0\|_2^2 \le \frac{2\alpha}{\lambda_{\mathrm{cl}}}.
\]
Thus, for any $(u_I,\pi)\in\mathcal{F}_\Omega(\alpha)$, the interior vector
$u_I$ lies in the closed Euclidean ball of radius
$\sqrt{2\alpha/\lambda_{\mathrm{cl}}}$ centred at $u_I^0$. The fog variable
$\pi$ always lies in the compact set $\mathcal{C}_\pi(\Omega)$ by definition
of $\mathcal{F}_\Omega$. It follows that $\mathcal{F}_\Omega(\alpha)$ is
bounded.

Moreover, $\mathcal{F}_\Omega(\alpha)$ is closed: it is the intersection of the
closed set $\mathcal{F}_\Omega$ with the closed inverse image
$\{(u_I,\pi): J_\Omega(u_I,\pi)\le \alpha\}$ of $(-\infty,\alpha]$ under the
continuous map $(u_I,\pi)\mapsto J_\Omega(u_I,\pi)$. Hence
$\mathcal{F}_\Omega(\alpha)$ is compact.

By construction $m$ is the infimum of $J_\Omega$ over $\mathcal{F}_\Omega$, so
there exists a sequence $(u_I^{(n)},\pi^{(n)})$ in $\mathcal{F}_\Omega$ such
that $J_\Omega(u_I^{(n)},\pi^{(n)})\downarrow m$ as $n\to\infty$. All but
finitely many of these points lie in $\mathcal{F}_\Omega(\alpha)$ for any
fixed $\alpha>m$. By compactness of $\mathcal{F}_\Omega(\alpha)$, the sequence
has a convergent subsequence $(u_I^{(n_k)},\pi^{(n_k)})$ with limit
$(u_I^\star,\pi^\star)\in\mathcal{F}_\Omega(\alpha)\subset\mathcal{F}_\Omega$.
Continuity of $J_\Omega$ implies
\[
J_\Omega(u_I^\star,\pi^\star)
=
\lim_{k\to\infty} J_\Omega(u_I^{(n_k)},\pi^{(n_k)})
= m.
\]
Thus $(u_I^\star,\pi^\star)$ attains the infimum and is a minimiser of
$J_\Omega$ on $\mathcal{F}_\Omega$.
\end{proof}

We now provide sufficient conditions for uniqueness of the patch minimiser.

\begin{prop}[Uniqueness under strict convexity]
\label{prop:uniqueness}
Suppose, in addition to the assumptions of Proposition~\ref{prop:existence}, that:
\begin{enumerate}
  \item The quadratic form in $u_I$ given by
  \[
  u_I \mapsto E_{\mathrm{cl}}(u_I) + \lambda_{\mathrm{surf}} E_{\mathrm{surf}}(u_I)
  \]
  is strictly convex on the affine hull of $\mathcal{C}_u(\Omega)$; equivalently,
  its Hessian $\lambda_{\mathrm{cl}} I_{N_\Omega} + \lambda_{\mathrm{surf}} Q_\Omega$
  is positive definite on the tangent cone of $\mathcal{C}_u(\Omega)$.
  \item $\lambda_\pi>0$ and the Hamiltonian matrix $H_\pi$ is positive definite
  on the affine hull of $\mathcal{C}_\pi(\Omega)$; equivalently, the quadratic
  form $\pi\mapsto \pi^\top H_\pi\pi$ is strictly convex on $\mathcal{C}_\pi(\Omega)$.
\end{enumerate}
Then the minimiser $(u_I^\star,\pi^\star)$ of $J_\Omega$ on
$\mathcal{C}_u(\Omega)\times\mathcal{C}_\pi(\Omega)$ is unique.
\end{prop}

\begin{proof}
By assumption (1), the map
\[
u_I \mapsto E_{\mathrm{cl}}(u_I) + \lambda_{\mathrm{surf}} E_{\mathrm{surf}}(u_I)
\]
is strictly convex on $\mathcal{C}_u(\Omega)$. By assumption (2), the map
\[
\pi \mapsto \lambda_\pi\,\mathcal{E}_{\mathrm{Ham}}(\pi)
= \frac{\lambda_\pi}{2}\,\pi^\top H_\pi \pi
\]
is strictly convex on $\mathcal{C}_\pi(\Omega)$. The remaining contribution to
$J_\Omega$ is the sum of noise-aware band terms
\[
(u_I,\pi)\mapsto \sum_{q\in Q_\Omega} \phi_q(u_I,\pi),
\]
which is convex by Proposition~\ref{prop:J-omega-convex} (and does not affect
strict convexity, since adding a convex function to a strictly convex one
preserves strict convexity).

To see that $J_\Omega$ is strictly convex on the product set
$\mathcal{C}_u(\Omega)\times\mathcal{C}_\pi(\Omega)$, let
$(u_I^{(1)},\pi^{(1)})$ and $(u_I^{(2)},\pi^{(2)})$ be two distinct feasible
points, and let $\theta\in(0,1)$. Then at least one of the components $u_I^{(1)}$
and $u_I^{(2)}$ differs, or $\pi^{(1)}$ and $\pi^{(2)}$ differ. If
$u_I^{(1)}\neq u_I^{(2)}$, strict convexity of the $u_I$-quadratic implies
\[
E_{\mathrm{cl}}(\theta u_I^{(1)} + (1-\theta)u_I^{(2)})
+ \lambda_{\mathrm{surf}} E_{\mathrm{surf}}(\theta u_I^{(1)} + (1-\theta)u_I^{(2)})
\]
\[
<
\theta\bigl(E_{\mathrm{cl}}(u_I^{(1)}) + \lambda_{\mathrm{surf}} E_{\mathrm{surf}}(u_I^{(1)})\bigr)
+ (1-\theta)\bigl(E_{\mathrm{cl}}(u_I^{(2)}) + \lambda_{\mathrm{surf}} E_{\mathrm{surf}}(u_I^{(2)})\bigr).
\]

If instead $u_I^{(1)} = u_I^{(2)}$ but $\pi^{(1)}\neq \pi^{(2)}$, strict
convexity of $\lambda_\pi \mathcal{E}_{\mathrm{Ham}}$ on $\mathcal{C}_\pi(\Omega)$
implies a strict inequality in the $\pi$-component. In either case, adding the
convex sum of band terms preserves strict inequality:
\[
J_\Omega\bigl(\theta u_I^{(1)} + (1-\theta)u_I^{(2)},\ \theta \pi^{(1)} + (1-\theta)\pi^{(2)}\bigr)
\]
\[
<\theta J_\Omega(u_I^{(1)},\pi^{(1)}) + (1-\theta) J_\Omega(u_I^{(2)},\pi^{(2)}).
\]
Thus $J_\Omega$ is strictly convex on the convex feasible set
$\mathcal{C}_u(\Omega)\times\mathcal{C}_\pi(\Omega)$. A strictly convex function
on a convex set has at most one minimiser. Combined with existence
(Proposition~\ref{prop:existence}), this implies that the minimiser of $J_\Omega$
is unique.
\end{proof}

\begin{remark}[Non-quadratic but convex structure]
\label{rem:non-QP}
The patch energy $J_\Omega$ is convex but not quadratic in the joint variables
$(u_I,\pi)$. The non-quadratic structure arises from the perspective-type terms
$d_q(u)^2/(\varepsilon+M_q(\pi))$ in the noise-aware band penalties $\phi_q$,
which couple the surface misfit and the fog mass in a nonlinear way. Introducing
additional slack variables to eliminate the perspective structure would break
the natural probabilistic interpretation of $\pi$ and $\nu_q(\pi)$, and is not
pursued here. Consequently, the patch-level post-fit is formulated and solved
as a general convex optimisation problem, rather than as a quadratic program.
\end{remark}

\section{Global post-fit across patches and dates}

We now describe how the patch-level post-fit is assembled into a global
arbitrage-free surface on each date, and state conditions under which global
static no-arbitrage is preserved.

Throughout this section we fix a calendar date $t$ and suppress explicit $t$-dependence
in the notation when no ambiguity arises. All objects (quotes, bands, forwards,
baseline $u^0$, operators $\ell_\alpha,r_\alpha$, etc.) are understood to be
associated with this fixed date.

\subsection{Patch decomposition and compatibility with no-arbitrage stencils}

Recall that the global discrete static no-arbitrage constraints on the nodal
grid $\mathcal{G}$ are encoded by the index set $\mathcal{I}$ and linear
inequalities
\begin{equation}
\label{eq:global-noarb-ineq-again}
\ell_\alpha^\top u \le r_\alpha,\qquad \alpha\in\mathcal{I},
\end{equation}
as in Definition~\ref{def:global-noarb-operators} and
equation~\eqref{eq:Cglobal}. For each $\alpha\in\mathcal{I}$, the \emph{support}
of the stencil is
\[
\mathrm{supp}(\ell_\alpha)
:=
\bigl\{ g \in \{1,\dots,G\} : (\ell_\alpha)_g \neq 0 \bigr\}.
\]
Equivalently, $\mathrm{supp}(\ell_\alpha)$ is the set of nodal indices at which
$u$ enters the $\alpha$-th constraint with nonzero coefficient.

Let $\{\Omega_p\}_{p\in\mathcal{P}}$ be a finite family of pairwise disjoint
patches in $\mathcal{G}$, i.e.
\[
\Omega_p \subset \mathcal{G},\quad
\Omega_p \cap \Omega_{p'} = \emptyset\ \text{for }p\neq p'.
\]
Define their union and complement by
\[
\Omega_{\mathrm{all}}
:=
\bigcup_{p\in\mathcal{P}} \Omega_p,
\qquad
\Omega_{\mathrm{off}}
:=
\mathcal{G}\setminus \Omega_{\mathrm{all}}.
\]

We explicitly assume that the patch decomposition is compatible with the global
no-arbitrage stencils in the following sense.

\begin{assump}[Stencil compatibility of the patch decomposition]
\label{ass:stencil-compatibility}
For every $\alpha\in\mathcal{I}$, the support of $\ell_\alpha$ is either
contained entirely in one patch or entirely outside all patches; that is, for
each $\alpha\in\mathcal{I}$ there exists either:
\begin{itemize}
  \item a patch index $p\in\mathcal{P}$ such that
        $\mathrm{supp}(\ell_\alpha) \subset \Omega_p$, or
  \item no patch index with this property, in which case
        $\mathrm{supp}(\ell_\alpha) \subset \Omega_{\mathrm{off}}$.
\end{itemize}
Equivalently, there is no $\alpha\in\mathcal{I}$ such that
$\mathrm{supp}(\ell_\alpha)$ intersects both $\Omega_p$ and
$\mathcal{G}\setminus\Omega_p$ for some $p$.
\end{assump}

\noindent
In words, no-arbitrage stencils do not “straddle” patch boundaries: each
discrete bound, monotonicity, convexity, or calendar constraint is supported
either entirely on a single patch, or entirely outside the union of patches.
This is a slightly stronger version of the patch feasibility condition discussed
after Definition~\ref{def:patch-feasible}, and is natural in view of the local
construction of patches from the badness field.

\subsection{Global post-fit surface on a fixed date}

For a fixed date $t$, the patch-level post-fit yields, for each patch
$\Omega_p$, a pair $(u_{I,p}^\star,\pi_p^\star)$ solving the patch problem
(Definition~\ref{def:patch-problem}) on that patch, i.e.
\[
(u_{I,p}^\star,\pi_p^\star)
\in \arg\min\bigl\{
J_{\Omega_p}(u_I,\pi) :
u_I\in\mathcal{C}_u(\Omega_p),\ \pi\in\mathcal{C}_\pi(\Omega_p)
\bigr\}.
\]
By construction, $u_{I,p}^\star\in\mathcal{C}_u(\Omega_p)$, so the assembled
surface $u^{(p)} := u(u_{I,p}^\star)$ (defined as in
Section~\ref{sec:patch-assembly}) satisfies all global no-arbitrage inequalities
\eqref{eq:global-noarb-ineq-again} with off-patch nodes fixed to their baseline
values.

We now combine all patch-level interior solutions into a single global nodal
surface $u^\star$ for date $t$.

\begin{definition}[Global post-fit surface on a date]
\label{def:global-post-fit}
Let $u^0\in\mathbb{R}^G$ be the baseline nodal surface for date $t$, and let
$\{\Omega_p\}_{p\in\mathcal{P}}$ be a stencil-compatible patch decomposition
(Assumption~\ref{ass:stencil-compatibility}) with corresponding interior
solutions $\{u_{I,p}^\star\}_{p\in\mathcal{P}}$. The \emph{global post-fit
nodal surface} $u^\star\in\mathbb{R}^G$ for date $t$ is defined componentwise
by
\[
u^\star_{i,j}
:=
\begin{cases}
(u_{I,p}^\star)_{i,j}, & \text{if }(i,j)\in \Omega_p\ \text{for some }p\in\mathcal{P},\\[3pt]
u^0_{i,j}, & \text{if }(i,j)\in \Omega_{\mathrm{off}}.
\end{cases}
\]
Equivalently, $u^\star$ coincides with the patch-level interior solutions on
each $\Omega_p$ and with the baseline on all nodes outside the union of
patches.
\end{definition}

\noindent
We emphasise that the fog fields $\pi_p^\star$ remain patch-local; they are not
assembled into a single global fog, since only $u^\star$ is used in further
pricing and calibration.

\subsection{Global static no-arbitrage and locality}

We now show that under Assumption~\ref{ass:stencil-compatibility}, the global
post-fit surface $u^\star$ is statically no-arbitrage on $\mathcal{G}$, and that
nodes outside the patches are unchanged.

\begin{prop}[Global static no-arbitrage and locality]
\label{prop:global-noarb-local}
Fix a date $t$ and suppose:
\begin{enumerate}[label=(\roman*)]
  \item the patch decomposition $\{\Omega_p\}_{p\in\mathcal{P}}$ satisfies
        Assumption~\ref{ass:stencil-compatibility};
  \item each patch $\Omega_p$ is feasible in the sense of
        Definition~\ref{def:patch-feasible} and admits a patch-level solution
        $(u_{I,p}^\star,\pi_p^\star)$ as in Definition~\ref{def:patch-problem};
  \item the baseline nodal surface $u^0$ is globally statically
        no-arbitrage, i.e.\ $u^0\in\mathcal{C}_{\mathrm{glob}}$.
\end{enumerate}
Let $u^\star$ be the global post-fit surface defined in
Definition~\ref{def:global-post-fit}. Then:
\begin{enumerate}
  \item $u^\star$ is statically no-arbitrage on $\mathcal{G}$, i.e.\
        $u^\star\in\mathcal{C}_{\mathrm{glob}}$; and
  \item locality holds: $u^\star_{i,j} = u^0_{i,j}$ for all
        $(i,j)\in\Omega_{\mathrm{off}}$.
\end{enumerate}
\end{prop}

\begin{proof}
Part (b) (locality) is immediate from the definition of $u^\star$: by
Definition~\ref{def:global-post-fit}, for $(i,j)\in\Omega_{\mathrm{off}}$ we set
$u^\star_{i,j} := u^0_{i,j}$. Hence $u^\star$ agrees with the baseline on all
off-patch nodes.

We now prove (a). It suffices to show that all global no-arbitrage inequalities
\eqref{eq:global-noarb-ineq-again} hold for $u^\star$; that is, we must verify
\[
\ell_\alpha^\top u^\star \le r_\alpha,\qquad \forall\alpha\in\mathcal{I}.
\]

Fix $\alpha\in\mathcal{I}$ and consider the support
$\mathrm{supp}(\ell_\alpha)$. By Assumption~\ref{ass:stencil-compatibility},
there are two mutually exclusive cases:

\medskip
\noindent
\emph{Case 1:} $\mathrm{supp}(\ell_\alpha)\subset \Omega_{\mathrm{off}}$.

In this case, the $\alpha$-th inequality involves only off-patch nodes. On
$\Omega_{\mathrm{off}}$, we have $u^\star_{i,j} = u^0_{i,j}$, so the $\alpha$-th
constraint evaluated at $u^\star$ is identical to that evaluated at $u^0$:
\[
\ell_\alpha^\top u^\star
=
\ell_\alpha^\top u^0.
\]
By assumption (iii), $u^0\in\mathcal{C}_{\mathrm{glob}}$, so
$\ell_\alpha^\top u^0 \le r_\alpha$. Hence
$\ell_\alpha^\top u^\star \le r_\alpha$ in Case~1.

\medskip
\noindent
\emph{Case 2:} There exists $p\in\mathcal{P}$ such that
$\mathrm{supp}(\ell_\alpha)\subset \Omega_p$.

In this case, the $\alpha$-th constraint involves only nodes inside the single
patch $\Omega_p$. Let $u^{(p)}$ denote the full nodal surface corresponding to
the patch-level interior solution $u_{I,p}^\star$, i.e.\ the assembled surface
obtained by replacing $u^0$ by $u_{I,p}^\star$ on $\Omega_p$ and keeping all
other nodes at their baseline values. By definition of $\mathcal{C}_u(\Omega_p)$
(Definition~\ref{def:CuOmega}), we have $u_{I,p}^\star\in\mathcal{C}_u(\Omega_p)$,
hence
\[
\ell_\alpha^\top u^{(p)} \le r_\alpha,\qquad \forall \alpha\in\mathcal{I}.
\]
In particular, for the specific index $\alpha$ under consideration,
\[
\ell_\alpha^\top u^{(p)} \le r_\alpha.
\]

We now compare $u^{(p)}$ and $u^\star$ on the support of $\ell_\alpha$. On
$\Omega_p$, both $u^{(p)}$ and $u^\star$ take the same nodal values, namely
$(u_{I,p}^\star)_{i,j}$; on $\mathcal{G}\setminus\Omega_p$, the $\alpha$-th
constraint has zero coefficients (since
$\mathrm{supp}(\ell_\alpha)\subset \Omega_p$). Therefore
\[
\ell_\alpha^\top u^\star
=
\ell_\alpha^\top u^{(p)}.
\]
Hence
\[
\ell_\alpha^\top u^\star
=
\ell_\alpha^\top u^{(p)} \le r_\alpha.
\]

\medskip
\noindent
In both cases we have shown $\ell_\alpha^\top u^\star \le r_\alpha$. Since
$\alpha\in\mathcal{I}$ was arbitrary, it follows that
$\ell_\alpha^\top u^\star \le r_\alpha$ for all $\alpha\in\mathcal{I}$, i.e.
$u^\star\in\mathcal{C}_{\mathrm{glob}}$. This proves (a).
\end{proof}

\begin{remark}[Independence across dates]
\label{rem:date-independence}
The above argument is purely cross-sectional and is applied separately on each
date $t$. There is no coupling in the static no-arbitrage constraints between
different dates, so the global post-fit surfaces $\{u_t^\star\}_t$ across all
dates are obtained by applying the per-date patch decomposition and assembly
independently. Provided that the assumptions of
Proposition~\ref{prop:global-noarb-local} hold for each date, the family
$\{u_t^\star\}_t$ is statically no-arbitrage on every date, and coincides with
the baseline surfaces outside the union of patches on each date.
\end{remark}

\chapter{Conclusion and outlook}\label{sec:conclusion}

We have presented a convex-programming framework for constructing arbitrage-free
option price surfaces based on a global Chebyshev representation on a warped
log-moneyness domain. By encoding static no-arbitrage inequalities as
linear constraints on a dense collocation grid, and fitting directly to prices via
a coverage-seeking quadratic objective, the method yields a surface that is
both smooth and internally consistent.

On the empirical side, our implementation attains high inside-spread coverage and
low rates of static no-arbitrage violations across a multi-year panel of equity
options. These results suggest that Chebyshev/QP formulations, combined with spectral-geometry
and transport-type regularisers, are a viable and competitive alternative to more
widely used parametric and spline-based approaches, particularly when tight control
over arbitrage metrics is required.

Beyond the global QP backbone, we have formulated a local post-fit layer in which
a discrete fog of risk-neutral densities on $(m,\tau,u)$ is endowed with a
Hamiltonian-type energy. On each problematic patch of the $(m,\tau)$-plane, this
fog is coupled convexly to a nodal price field that remains globally
arbitrage-free. The resulting patch problems are jointly convex in the surface
and fog variables and yield noise-aware corrections that improve local band
coverage in stressed regions while preserving static no-arbitrage and locality.

Several limitations and directions for further work remain. First, our study
focuses on a particular choice of warping, regularisation and grid design; different
markets or underlyings may benefit from alternative configurations, and a more
systematic comparison against SVI-type and deep-learning-based surfaces would be
informative. This will be tackled in a separate paper. Second, we have evaluated performance primarily through static
diagnostics (spread coverage, violation rates, smoothness) and local band metrics
on patches. A natural next step is to examine the impact on hedging performance
and risk measures, for example via delta-hedging backtests or scenario analysis of
risk-neutral densities, both for the baseline QP and for the fog-corrected surface.

Finally, the Hamiltonian fog layer is implemented here in a finite-dimensional,
patch-wise discretisation. From a mathematical standpoint, it suggests a continuous
framework in which a fog density $\pi(m,\tau,u)$ on a three-dimensional manifold
evolves under a Hamiltonian or transport-type metric, with the option surface
appearing as a constrained ``sheet'' inside this geometry. Developing this
continuous theory-including PDE and variational formulations, existence and
uniqueness questions, and connections to optimal transport on the space of
risk-neutral measures is beyond the scope of the present paper and will be
pursued in separate work. We view the discrete constructions in this article as a practical, convex
realisation of that program: the global Chebyshev/QP fit provides a transparent
arbitrage-free backbone, and the patch-wise Hamiltonian fog post-fit offers a
local, noise-aware refinement that remains compatible with production-style
constraints and solvers.

\section*{Acknowledgments}

The author made limited use of an AI language model (ChatGPT by OpenAI) as a writing and brainstorming aid; all models, proofs and numerical results presented are the author's own work and have been independently verified.

\nocite{*}
\bibliographystyle{abbrv} 
\bibliography{sample}

\end{document}